%% file: Augmenting_the_Algebraic_Connectivity_of_Graphs.tex
\documentclass[11pt,a4paper]{article}
\usepackage{bbding}
\fussy
\usepackage[pagebackref=false,citecolor=newblue,colorlinks=true,linkcolor=newblue,bookmarks=false]{hyperref}

\usepackage[font=scriptsize,bf]{caption}
\usepackage{epsfig,amssymb,amsfonts,amsmath,amsthm}
\usepackage[margin=1in]{geometry}

\usepackage{algorithm,algorithmicx}
\usepackage[noend]{algpseudocode}
\usepackage{amsopn}

\usepackage{xcolor}

\usepackage{rotating}
\usepackage{enumerate}
\definecolor{ocre}{rgb}{0.72,0,0} 
\definecolor{MyMango}{rgb}{1.00, 0.47, 0.20}
\definecolor{brickred}{rgb}{0.8, 0.25, 0.33}
\definecolor{newblue}{rgb}{0.2,0.2,0.6} 
\usepackage{mdframed}
\newenvironment{fminipage}%
  {\begin{Sbox}\begin{minipage}}%
  {\end{minipage}\end{Sbox}\fbox{\TheSbox}}

\usepackage{tikz}
\usetikzlibrary{shapes.geometric,positioning}
\usetikzlibrary{arrows.meta}
\definecolor{ocre}{rgb}{0.72,0,0} 
\definecolor{babyblueeyes}{rgb}{0.63, 0.79, 0.95}
\definecolor{newgreen}{rgb}{0.53,0.66,0.42} 
\definecolor{newred}{rgb}{0.67,0.16,0}

\usepackage{algpseudocode}
\usepackage{xparse}
\usepackage{ulem}

\newcommand{\Pro}{\mathbf{P}}
\newcommand{\Ex}{\mathbf{E}}

\newcommand{\Tr}{\mathrm{tr}}
\newcommand{\tr}{\mathrm{tr}}
\newcommand{\eps}{\varepsilon}
\newcommand{\OPT}{\textrm{OPT}}

\newcommand{\1}{\textbf{1}}

\newcommand{\poly}{\ensuremath{\mathrm{poly}}}

\newcommand{\rot}{\intercal}
\newcommand{\barr}{\overline}

\newcommand{\lPot}[2]{\Tr \lp \PV(#2 - #1 I)\PV \rp ^{\dag q}}
\newcommand{\uRes}[2]{\left(#1 I - #2\right)^{-1}}
\newcommand{\lRes}[2]{Z \lp \PV(#2 - #1 I)\PV \rp ^{\dag} Z}

\newcommand{\PL}[1]{P_{L(#1)}}
\newcommand{\PV}{P_{\mathcal{V}}}
\newcommand{\lp}{\left (}
\newcommand{\rp}{\right )}
\newcommand{\lsp}{\left [}
\newcommand{\rsp}{\right ]}

\newcommand{\Diag}{\ensuremath{\textbf{\textrm{Diag}}}}
\newcommand{\PSDP}{\ensuremath{\textsf{P-SDP}}}

\newcommand{\DSDP}{\ensuremath{\textsf{D-SDP}}}

\newcommand{\lambdaMin}{\lambda_{\mathrm{min}}}
\newcommand{\lambdaMax}{\lambda_{\mathrm{max}}}
\newtheorem{definition}{Definition}[section] 
\newtheorem{remark}[definition]{Remark}

\newtheorem{fact}[definition]{Fact}
\newtheorem{lemma}[definition]{Lemma}

\newtheorem{claim}[definition]{Claim}
\newtheorem{theorem}[definition]{Theorem}

\newtheorem{assumption}[definition]{Assumption}

 \providecommand{\norm}[1]{\lVert#1\rVert}
\providecommand{\mat}[1]{{#1}}

\renewcommand{\tilde}{\widetilde}
\renewcommand{\hat}{\widehat}

\makeatletter
\NewDocumentCommand{\lplabel}{o m}{%
	\makebox[0pt][r]{#2\hspace*{2em}}%
	\IfNoValueF{#1}
	{\def\@currentlabel{#2}\ltx@label{#1}}
}
\makeatother

\allowdisplaybreaks

\title{Augmenting the Algebraic Connectivity of Graphs}

\author{ Bogdan-Adrian Manghiuc\footnote{School of Informatics,   University of Edinburgh, UK. \url{b.a.manghiuc@sms.ed.ac.uk}. This work is supported by an EPSRC Doctoral Training Studentship (EP/R513209/1)} \and Pan Peng\footnote{Department of Computer Science, University of Sheffield, UK. \url{p.peng@sheffield.ac.uk}. \href{https://orcid.org/0000-0003-2700-5699}{ORCID~iD: 0000-0003-2700-5699}.
} \and He Sun\footnote{School of Informatics,   University of Edinburgh, UK. \url{h.sun@ed.ac.uk}. Part of the work is supported by an EPSRC Early Career Fellowship~(EP/T00729X/1).}  }
\date{}

\begin{document}
\maketitle
 
\thispagestyle{empty}

\setcounter{page}{0}

\input{Abstract}
\newpage
\thispagestyle{empty}
\setcounter{page}{0}
\tableofcontents
\newpage

\input{Introduction}
\input{Preliminaries}

\input{AlgConn}
\input{Subgraph_Sparsification_Description}

\input{Subgraph_Sparsification_Analysis}
\input{MainResult}
\newpage

\bibliographystyle{alphabetic}
\bibliography{effresistance,reference}

\newpage

\input{appendix}

\end{document}

%% file: Abstract.tex
\begin{abstract}

For any undirected graph $G=(V,E)$ and a set $E_W$ of candidate edges with $E\cap E_W=\emptyset$, the $(k,\gamma)$-spectral augmentability problem is to find a set $F$ of $k$ edges from $E_W$ with appropriate weighting, such that the algebraic connectivity of the resulting graph $H=(V,E\cup F)$ is least $\gamma$. Because of a tight connection between the algebraic connectivity and many other graph parameters, including the graph's conductance and the mixing time  of random walks in a graph, maximising the resulting graph's algebraic connectivity by adding a small number of edges
has been studied 
over the past 15 years,   and has many practical applications in network optimisation.

In this work we present an approximate and  efficient algorithm for the   $(k,\gamma)$-spectral augmentability  problem, and our algorithm runs in almost-linear time under a wide regime of parameters. Our main algorithm is based on   the following two novel techniques developed in the paper, which might have applications beyond   the   $(k,\gamma)$-spectral augmentability problem:
\begin{itemize}
    \item We present a fast algorithm for solving a feasibility version of an \textsf{SDP} for the algebraic connectivity maximisation problem from \cite{GB06:growing}. Our algorithm is based on the classic primal-dual framework for solving \textsf{SDP}, which in turn uses the multiplicative weight update algorithm. We present a novel approach of unifying \textsf{SDP} constraints of different matrix and vector variables and give a good separation oracle accordingly.
    \item  We present an efficient algorithm for the subgraph sparsification problem, and for a wide range of parameters our algorithm runs in almost-linear time, in contrast to the previously best known algorithm running in at least $\Omega(n^2mk)$ time~\cite{KMST10:subsparsification}.  
    Our analysis shows  how the randomised BSS framework can be generalised in the setting of subgraph sparsification, and how the potential functions can be applied to approximately  keep track of different subspaces.
\end{itemize}
  
\end{abstract}

%% file: Introduction.tex
\section{Introduction}\label{sec:introduction}
\normalem

Graph expansion is the metric quantifying   how well vertices are connected in a graph, and
has applications in many important problems of  computer science:  in complexity theory, graphs with good expansion  are   used to construct error-correcting codes~\cite{SipserS96,Zemor01} and pseudorandom generators~\cite{ImpagliazzoNW94}; in network design, expander graphs have been applied   in constructing super concentrators~\cite{Valiant76}; in probability theory,  graph expansion is closely related to the behaviours of random walks in a graph~\cite{mihail1989conductance,SinclairJ89}. On the other side, as most graphs occurring in practice might not be expander graphs and a subset of vertices of low expansion is usually viewed as the bottleneck of a graph, finding the set of vertices with minimum expansion  has many practical applications including image segmentation~\cite{MeilaShiNIPS01}, community detection~\cite{CSWZ,nips02,PSZ}, ranking web pages, among many others. Because of these, both the  approximation  algorithms for   the graph expansion  problem and the computational complexity of the problem itself  have  been extensively studied over the past three decades. 

In this paper we study the following graph expansion optimisation problem: given    an undirected and weighted graph $G=(V,E,w)$, a set $E_W$ of candidate edges, and a parameter $k\in\mathbb{N}$ as input, we are  interested in (i) finding a set $F\subseteq E_W$ of $k$ edges and their weights such that the resulting  graph $H=(V,E\cup F,w')$ with weight function $w': E\cup F\rightarrow \mathbb{R}_{\geq 0}$ has good expansion, or (ii) showing that it's impossible to significantly improve the graph's expansion by adding $k$ edges from $E_W$. Despite sharing many similarities with the sparest cut problem, our proposed problem has many of its own applications: for example,  assume that  the underlying graph $G$ is a  practical traffic or communication  network and, due to physical constraints, only certain links can be used to improve the network's connectivity. For any given $k$ and a set of feasible links, finding the best  $k$  links to optimise the  network's connectivity is   exactly the objective of our graph expansion optimisation problem.

To formalise the problem,  we follow the work of \cite{fiedler1973algebraic,GB06:growing,KMST10:subsparsification} and define the \emph{algebraic connectivity} of $G$ by the second smallest eigenvalue $\lambda_2(L_G)$ of the Laplacian matrix $L_G$ of $G$ defined by $L_G\triangleq D_G - A_G$, where $D_G$ is the diagonal matrix consisting of the degrees of the vertices and $A_G$ is the adjacency matrix of $G$.  Given   an undirected and weighted graph $G=(V,E,w)$ with $n$ vertices,  $O(n)$ edges\footnote{Since a spectral sparsifier of $G$ with $O(n)$ edges preserves the eigenvalues of the Laplacian matrix of $G$, we   assume that $G$ has $O(n)$ edges throughout the paper. Otherwise one can always run the algorithm in \cite{LS17} to get a spectral sparsifier of $G$ with $O(n)$ edges and use this spectral sparsifier as  the input of our problem. Because of this, the number of edges in $G$ will not be mentioned in our paper to simplify the notation.},
a set $E_W$ of candidate edges defined on $V$    satisfying  $E_W\cap E=\emptyset$ and a parameter $k$,  we say that $G$ is $(k,\gamma)$-\emph{spectrally-augmentable with respect to} $W=(V,E_W)$, if there is $F\subseteq E_W$ with $|F|=k$ together with  edge weights $\{w_e\}_{e\in F}$ such that  $H=(V, E\cup F, w)$ satisfies $\lambda_2(L_H)\geq \gamma$.  The main result of our work is an almost-linear time\footnote{We say that a graph  algorithm runs in almost-linear time  if the algorithm's runtime is $O((m+n)^{1+c})$ for an arbitrary small constant $c$, where $n$ and $m$ are the number of vertices and edges of $G$, respectively. Similarly, we say that a graph algorithm runs in nearly-linear time if the algorithm's runtime is $O((m+n)\cdot \log^c(n))$ for some constant $c$.} algorithm that either (i) finds a set of $O(k)$ edges from $E_W$ if $G$ is $(k,\gamma)$-spectrally augmentable for some  $\gamma\geq \Delta\cdot n^{-1/q}$, or (ii) returns ``no'' if $G$ is not $(O(kq), O(\Delta\cdot n^{-2/q}))$-spectrally augmentable, where $\Delta$ is an upper bound of both the maximum degree of $G$ and $W$. The formal description of our result is as follows:

\begin{theorem}\label{thm:mainAlgConn}
Let $q\geq 10$ be an integer. Let $G=(V,E,w)$ be a base graph with $n$ vertices, $O(n)$ edges, and weight function $w:E\rightarrow \mathbb{R}_{\geq 0}$, and let $W=(V, E_W)$ be the candidate graph of $m$ edges such that the maximum degrees of $G$ and $W$ is at most $\Delta$.  
Then, there is an algorithm such that for any integer $k\geq 1$, the following statements hold: 
\begin{itemize}
    \item if $G$ is $\left(k,  \Delta\cdot n^{-1/q}\right)$-spectrally-augmentable with respect to $W$, then 
    the algorithm finds a set $F\subseteq E_W$ of edges and a set of edge weights $\{w(e):e\in F\}$ such that $|F|=O(qk)$, $\sum_{e\in F} w(e)\leq O(k)$, and the resulting graph $H=(V,E \cup F)$ satisfies that $\lambda_2(L_{H})\geq c \lambda_{\star}^2 \Delta$, for some constant $c>0$, where $\lambda_{\star}\cdot \Delta$ is the optimum solution\footnote{Note that since $G$ is $(k,n^{-1/q}\cdot \Delta)$-spectrally-augmentable with respect to $W$, it always holds that $\lambda_{\star}\geq n^{-1/q}$.}.
    \item if $G$ is not $\left(O(kq), O(\Delta\cdot n^{-2/q})\right)$-spectrally-augmentable with respect to $W$, then the algorithm  rejects the input $G, W$.
\end{itemize}
Moreover, the algorithm runs in 
$\widetilde{O} \left(\min \left\{ qn^{\omega+O(1/q)}, q(m+n)n^{O(1/q)}k^2\right\}\right)$ time. Here, the $\widetilde{O}(.)$ notation hides $\mathrm{poly}\log n$  factors, and $\omega$ is the constant for matrix multiplication. 
\end{theorem}

 We remark that the most typical  application of our problem is the scenario in which only a low number of edges are needed such that the resulting graph enjoys good expansion, and these correspond to the regime of $k=n^{o(1)}$ and $\lambda_{\star}\in (n^{-1/q},O(1))$~\cite{conf/soda/GharanT14}, under which our algorithm runs in almost-linear time and achieves an $\Omega(\lambda_\star)$-approximation. In particular, when it is possible to augment $G$ to be an expander graph, i.e.~$\lambda_{\star}=\Theta(1)$, our algorithm achieves a  constant-factor approximation. 
Our algorithm runs much faster than the previously best-known algorithm for a similar problem that runs in \emph{at least} $\Omega\left(n^2mk\right)$ time~\cite{KMST10:subsparsification}, 
though their algorithm solves the more general problem: for any instance $G,W,k$, if the optimum solution is $\lambda_\star \Delta$, i.e., $G$ is $(k,\lambda_\star\Delta)$-spectrally-augmentable with respect to $W$, for \emph{any} $\lambda_\star\in [0,1)$, then their algorithm finds a graph $H=(V,E\cup F)$ with $\lambda_2(L_H)\geq c\lambda_{\star}^2 \Delta$ such that $|F|=O(k)$ and the total sum of weights of edges in $F$ is at most $k$. Our algorithm can only find a graph $H$ when $\lambda_{\star}\in (n^{-1/q},1)$. 


To give an overview of the proof technique for    Theorem~\ref{thm:mainAlgConn}, notice that  our problem is closely linked to the \emph{algebraic connectivity maximisation problem} studied  in \cite{GB06:growing}, which looks for $k$ edges from the candidate set to maximise $\lambda_2(L_H)$ of the resulting graph $H$.
It is known that   the algebraic connectivity maximisation problem is $\textbf{NP}$-hard~\cite{mosk2008maximum}, and Ghosh and Boyd~\cite{GB06:growing} show that this problem  can be formulated as an \textsf{SDP}, which we call the \textsf{GB-SDP}.
Inspired by  this, we study  the following \textsf{P-SDP}, which is the feasibility version of the \textsf{GB-SDP} parameterised by some $\gamma$. Here, $P_{\bot}$ is the   projection on the space orthogonal to $\1\triangleq (1,\dots,1)^{\rot}$, i.e., $P_{\bot}=I-\frac{1}{n}\1\1^{\rot}$. 

 ${\PSDP(G,W,k,\gamma)}$ 
\(\displaystyle
\begin{aligned}[c]
 & & \lambda\geq \gamma &  \\
 & & \mat{L}_G +\sum_{e\in E_W} w_e \mat{L}_e \succeq \lambda \Delta\mat{P}_{\bot}  & \\
 && k-\sum_{e \in E_W} w_e \geq 0 & \\
 && 1-w_e\geq 0, & \qquad \forall e\in E_W\\
&& w_e\geq 0,   &\qquad \forall e \in E_W\\
&&\gamma \geq 0.&
\end{aligned}
\)

Notice that, if $G$ is $(k,\gamma \Delta)$-spectrally-augmentable with respect to $W$, then there is a set $F$ of $k$ edges such that, by setting $w_e=1$ if $e\in F$ and $w_e=0$ otherwise, it holds that   $L_G+\sum_{e\in E_W}L_e \succeq {\gamma} \Delta P_\bot$. Therefore, there is a feasible solution of $\PSDP(G,W,k,{\gamma})$. Our algorithmic result for solving the $\PSDP$ is summarised as follows:

\begin{theorem}\label{thm:mainSDPlambda2}
Let $\delta'>0$ be any constant. 
There exists an algorithm running in $\tilde{O}(( m+n)/\gamma^2)$ time  that 
either finds a solution to   $\PSDP(G,W,k,(1-\delta')\gamma)$ or certifies that there is no feasible solutions for $\PSDP(G,W,k,\gamma)$.  
\end{theorem} 

Since the solution to the \PSDP\ only guarantees that the total weights of the selected edges are at most $k$ if $G$ is $(k,\gamma\Delta)$-spectrally augmentable, following   \cite{KMST10:subsparsification} we use a subgraph sparsification algorithm to round our \textsf{SDP} solution, such that  there are only $O(k)$ edges selected in the end. To give a high-level overview of this rounding step, we redefine the set $E_W$ of candidate edges, and assume that $E_W$ consists of the edges  whose weight from the \PSDP\ solution is non-zero. Therefore, our objective is to find $O(k)$ edges from $E_W$ and new weights, which form an edge set $F$, such that
the Laplacian matrix $L_{H}$ of 
the resulting graph $H=(V, E\cup F)$ is close to $L_{G+W}$. That is, the subgraph sparsification problem asks for a sparse representation of $G+W$ while keeping the entire base graph $G$  in the resulting   representation. Our improved algorithm shows that, as long as $k=n^{o(1)}$, a subgraph sparsifier can be computed in almost-linear time\footnote{We remark that, when $k=\Theta(n)$, our problem can be solved directly by using a spectral sparsifier $\widetilde{W}$ of the graph $W$ with $O(n)$ edges, which can be computed in nearly-linear time. This will imply that the two Laplacians $L_{G + \widetilde{W}} = L_G + L_{\widetilde{W}}$ and $L_{G + W} = L_G + L_W$ are close.}.
Our result on subgraph sparsification will be formally described in  Theorem~\ref{thm:main-formal}.

\subsection{Our techniques}

In this section  we will explain the techniques used to design the  fast algorithm for the $\PSDP$, and an almost-linear time  algorithm for subgraph sparsification.  

\paragraph*{Faster algorithm for solving the $\PSDP$.}  Our efficient $\mathsf{SDP}$ solver is based on    the primal-dual framework developed    in \cite{AK16:combinatorial}, which has been     used in many other works~\cite{jain2011qip,OV11:towards}. In this primal-dual framework, we  will    work on both the original \textsf{SDP} $\PSDP(G,W,k,\gamma)$ and its dual $\DSDP(G,W,k,\gamma)$ defined as follows:

${\DSDP(G,W,k,\gamma)} $
\(\displaystyle
\begin{aligned}[r]
&&  \mat{Z}\bullet \mat{L}_G+ k v + \sum_{e\in E_W} \beta_e <\gamma &\\
& & 
 \mat{Z}\bullet \Delta \mat{P}_{\bot} = 1 &\\
& &  \mat{Z}\bullet \mat{L}_{e}  \leq v + \beta_e, &\qquad \forall e \in E_W \\
& & \mat{Z} \succeq \mat{0} & \\
& & \beta_e \geq 0, & \qquad\forall e \in E_W\\
&& v\geq 0.&
\end{aligned}
\)

We then apply the matrix multiplicative weight update (MWU) algorithm. Formally speaking, starting with some initial embedding $X^{(1)}$, for each $t\geq 1$ our algorithm 
iteratively uses a carefully constructed oracle \textsc{Oracle} for $\DSDP(G,W,k,\gamma)$ to check whether the current embedding $X^{(t)}$ is good or not: 
\begin{itemize}  
    \item If  $X^{(t)}$   satisfies some condition, denoted by $\mathcal{C}\left(X^{(t)}\right)$, then the oracle fails, which implies that we can find a feasible solution from $X^{(t)}$ to $\DSDP(G,W,k,\gamma)$. This   implies that the primal SDP $\PSDP(G,W,k,\gamma)$ has no feasible solution, which certifies that $G$ is not $(k,\gamma)$-spectrally-augmentable with respect to $W$.
\item If  $X^{(t)}$ does not satisfy the condition $\mathcal{C}\left(X^{(t)}\right)$, then the oracle does not fail, which certifies that the current solution from $X^{(t)}$ is not feasible for $\DSDP(G,W,k,\gamma)$, and will output a set of numbers $\left(\lambda^{(t)},w^{(t)}\right)$ for updating the embedding.
\end{itemize} 
The procedure above will be iterated  for   $T$ times, for some $T$ depending on the oracle and the approximate parameter $\delta'>0$: if the oracle fails in any iteration, then $\PSDP$ is infeasible; otherwise, the oracle does not fail for all $T$ iterations and we find a feasible solution to $\PSDP(G,W,k,(1-\delta') \gamma)$. 

The main challenge for applying the above framework in our setting is to construct the \textsc{Oracle} and deal with the complicated constraints in our \textsf{SDP}s, which include both matrix inequality constraint and vector inequality constraints of different variables. To work with these constraints, our strategy is to unify them through a diagonal block matrix $X$, and through this we turn all individual constraints into a single matrix constraint. 
The embedding in each iteration is constructed in nearly-linear time in $n+m$ by the definition of the embedding.  
To construct the \textsc{Oracle}, we carefully design the condition $\mathcal{C}(X)$ with the intuition that if the candidate solution corresponding to $X$ has a relatively small dual objective value, then a re-scaling of $X$ gives a feasible solution to $\DSDP$. Then we use a case analysis to show that if $\mathcal{C}(X)$ is not satisfied, we can very efficiently find updating numbers $(\lambda^{(t)},w^{(t)})$ by distinguishing edges satisfying one constraint (in $\DSDP$) from those that do not satisfy it and assigning different weights $w^{(t)}$ to them accordingly. 

\paragraph*{Faster algorithm for  subgraph sparsification.} 

The second component behind proving our main result is an efficient algorithm for the subgraph sparsification problem.  
Our algorithm is inspired by the    the original deterministic algorithm for subgraph sparsification~\cite{KMST10:subsparsification} and the almost-linear time algorithm for constructing linear-sized spectral sparsifiers~\cite{LS15:linearsparsifier}. In particular, both algorithms follow the   BSS framework, and proceed in iterations: it is shown that, with the careful choice of barrier values $u_j$ and $\ell_j$ in each iteration $j$ and the associated potential functions, one or more vectors can be selected in each iteration and the final barrier values can be used to bound the approximation ratio of the constructed sparsifier.

However, in contrast to most algorithms for   linear-sized spectral sparsifiers~\cite{zhu15,LS17,LS15:linearsparsifier}, both the barrier values and the  potential functions  in \cite{KMST10:subsparsification} are employed for a slightly  different purpose. In particular, instead of expecting the final constructed ellipsoid to be close to being a sphere, the final constructed ellipsoid for subgraph sparsification could be still very  far from being a sphere, since the total number of added edges is $O(k)$. Because of this, the two potential functions in \cite{KMST10:subsparsification} are used to quantify the contribution of the added vectors towards \emph{two different subspaces}: one fixed $k$-dimensional subspace denoted by $S$, and one variable space defined with respect to the currently constructed matrix.  Based on analysing two different subspaces for every added vector, which is computationally expensive, the   algorithm in \cite{KMST10:subsparsification}
ensures that the added vectors will significantly benefit the ``worst subspace'', the subspace in $\mathbb{R}^n$ that limits the approximation ratio of the final constructed sparsifier. 

Because of these different roles of the potential functions in \cite{KMST10:subsparsification} and \cite{BSS,LS15:linearsparsifier}, when applying the randomised BSS framework~\cite{LS15:linearsparsifier} for the subgraph sparsification problem, more technical issues need to take into account: 
(1)  Since \cite{KMST10:subsparsification} crucially depends on some projection matrix denoted by $P_S$, of which the exact computation is expensive,  to obtain an efficient algorithm for subgraph sparsification one needs to obtain some projection matrix close to $P_S$ and such a projection matrix can be computed efficiently. (2)   Since the upper and lower potential functions keep track of different subspaces whose dimensions are of different orders in most regimes, analysing the impact of multiple added vectors to the potential functions are significantly more challenging than \cite{LS15:linearsparsifier}. 

To address the first issue, we show that the problem of computing an approximate projection close to $P_S$ while preserving relevant proprieties can be reduced to the generalised eigenvalue problem, which in turn can be efficiently approximated   by a recent algorithm~\cite{pmlr-v70-allen-zhu17b}. For the second issue, we meticulously bound the \emph{intrinsic dimension} of the matrix corresponding to the multiple added vectors, and by a more refined matrix analysis than \cite{LS15:linearsparsifier} we show that the potential functions and the relative effective resistances decease  in each iteration.
We highlight that   developing a fast procedure to computing all the quantities that involve a fixed projection matrix and analysing the impact of multiple added vectors with respect to two different subspaces constitute the most challenging part of the design of our algorithm.  


Finally, we remark that, although the almost-linear time algorithm~\cite{LS15:linearsparsifier} has been improved by the subsequent paper~\cite{LS17}, it looks more challenging to adapt the technique developed in \cite{LS17} for the setting of subgraph sparsification. In particular, since the two potential functions in \cite{LS17} are used to analyse the same space $\mathbb{R}^n$, it is shown in \cite{LS17} that it suffices to analyse the one-sided case through a one-sided oracle. However, the   two potential functions defined in our paper are used to analyse two different subspaces, and it remains unclear whether we can reduce our problem to the one-sided case. We will leave this for future work.

\subsection{Other related work}

Spielman and Teng~\cite{spielman_teng:SS11}  present the first algorithm for constructing spectral sparsifiers: for any parameter $\varepsilon\in(0,1)$, and 
any undirected graph $G$ of $n$ vertices and $m$ edges, they prove that a spectral sparsifier of $G$ with $\widetilde{O}\left(n/\varepsilon^2\right)$ edges exists, and can be constructed in $\widetilde{O}\left(m/\varepsilon^2 \right)$ time. Since then, there has been extensive studies on different variants of spectral sparsifiers and their efficient constructions in various settings. In addition to several results on   several constructions of linear-sized spectral sparsifiers mentioned above, there are many   studies on constructing spectral sparsifiers in streaming and dynamic settings~\cite{AbrahamDKKP16,KapralovLMMS17,kelner2013spectral}. 
The subgraph sparsification problem has many applications, including constructing
precondtioners and  nearly-optimal ultrasparsifiers~\cite{KMST10:subsparsification,Peng2013},  optimal approximate matrix product~\cite{CohenNW16}, and some network optimisation problems \cite{mellette2019expanding}. 
Our work is also related to a sequence  of research on network design, in which the goal is to find minimum cost subgraphs under some ``connectivity constraints''. Typical examples include constraints on vertex connectivity \cite{chakraborty2008network,cheriyan2014approximating, chuzhoy2009k,fakcharoenphol2012log,kortsarz2004hardness,laekhanukit2014parameters}, shortest path distances \cite{dinitz2016approximating,dodis1999design},
and spectral information \cite{allen2017near,boyd2004fastest, GBS08:effresist,nikolov2019proportional}.
 

%% file: Preliminaries.tex
\section{Preliminaries}
\label{sec:preliminaries}

In this section, we list all the notation used in our paper, and the lemmas used for proving the main results.

 \subsection{Notation}
 For any   symmetric matrix $A\in\mathbb{R}^{n\times n}$, let $\lambda_{\min}(A)=\lambda_1(A)\leq \cdots\leq \lambda_n(A)=\lambda_{\max}(A)$ be the eigenvalues of $A$,  where $\lambda_{\min}(A)$ and $\lambda_{\max}(A)$ represent the    smallest and    largest eigenvalues of $A$. 
For any subspace $S$ of dimension $k$, let  $P_S$ be the orthogonal projection onto  $S$.   For any matrix $A$, let $A\big|_{S}$ be the restriction of $A$ to   $S$. Note that $A\big|_S$ is a $k \times k$ matrix.

We call a matrix $A$  \emph{positive semi-definite~(\textsf{PSD})} if  $x^{\rot}Ax\geq 0$ holds for any $x\in\mathbb{R}^n$, and a matrix $A$ \emph{positive definite}  if  $x^{\rot}Ax> 0$ holds for any $x\in \mathbb{R}^n \setminus \{0\}$.  
 For any \textsf{PSD} matrix $A$,  let $A^{\dagger}$ be the pseudoinverse of $A$, and the \emph{intrinsic dimension} of $A$ is defined by 
$$\mathrm{intdim}(A) \triangleq \frac{\Tr(A)}{\norm{A}},$$
where 
$\|A \|$ is the spectral norm of $A$.
The number of non-zero entries of any matrix $A$ is denoted by $\mathsf{nnz}(A)$. 
For any positive definite matrix $A$, we define the corresponding ellipsoid by
$$
\mathsf{Ellip}(A)\triangleq\left\{ x :~ x^{\rot}A^{-1}x\leq 1 \right\}.
$$
For any two matrices $A$ and $B$, we write $A\preceq B$ to represent $B-A$ is \textsf{PSD}, and $A\prec B$ to represent $B-A$
 is positive definite.   For two matrices $A$ and $B$ and positive scalar $\varepsilon$ we write $A \approx_{\varepsilon} B$ if 
$
    (1-\varepsilon) \cdot B \preceq A \preceq (1+\varepsilon) \cdot B$.
The trace of matrix $A$ is denoted by $\tr(A)$, and we use  $A \bullet B$ to denote the entry-wise products of $A$ and $B$, i.e., $A\bullet B=\sum_{ij}A_{ij} B_{ij}$, which implies that  $\Tr \left(A^{\rot}B\right)= A\bullet B$.

  For any connected and undirected graph  
  $G=(V,E,w)$ with $n$ vertices, $m$ edges, and weight function $w: E\rightarrow \mathbb{R}_{\geq 0}$, we  
   fix an arbitrary orientation of the edges in $G$, and let $B\in\mathbb{R}^{m\times n}$ be the \emph{signed edge-vertex incidence matrix} defined by 
 \begin{equation} \nonumber
B_G(e,v)\triangleq \left\{ \begin{aligned}
         1 & \qquad \mbox{if $v$ is $e$'s head,} \\
         -1 & \qquad \mbox{if $v$ is $e$'s tail,} \\
                 0&\qquad \mbox{otherwise.}
                          \end{aligned} \right.
                          \end{equation}
    We define $b_e\in\mathbb{R}^n$ for every edge $e=\{u,v\}$, i.e., $b_e(u)=1, b_e(v)=-1$, and $b_e(w)=0$ for any vertex $w$ different from $u$ and $v$.  
We  define an $m\times m$ diagonal matrix $W_G$ by $W_G(e,e)=w_e$ for any edge $e\in E[G]$, and assume that the value of $w_e$ is polynomially bounded in $n$ for every edge $e$. 
 The Laplacian matrix of $G$ is an $n\times n$ matrix $L$ defined by
\begin{equation} \nonumber
L_G(u,v)\triangleq \left\{ \begin{aligned}
         -w(u,v) & \qquad \mbox{if $u\sim v$,} \\
         \deg(u) & \qquad \mbox{if $u=v$,} \\
                 0&\qquad \mbox{otherwise,}
                          \end{aligned} \right.
                          \end{equation}
where $\deg(v)=\sum_{u\sim v} w(u,v)$.
It is easy to verify 
 that \[
x^{\rot}L_G x = x^{\rot}B_G^{\rot}W_GB_Gx =\sum_{u\sim v} w_{u,v}(x_u-x_v)^2 \geq 0
\] holds
for any $x\in \mathbb{R}^n$. Hence, the Laplacian matrix of any undirected graph is a positive semi-definite matrix.  
Notice that, by setting $x_u=1$ if $u\in S$ and $x_u=0$ otherwise, $x^{\rot}L_Gx$ equals to the value of the cut between $S$ and $V\setminus S$. Hence, a spectral sparsifier is a stronger notion than a cut sparsifier.

\subsection{Useful facts in matrix analysis}
The following lemmas will be used in our analysis.

\begin{lemma}[Sherman-Morrison Formula]\label{lem:woodbury}
Let $\mat{A}\in \mathbb{R}^{n\times n}$ be an invertible matrix, and $u,v\in\mathbb{R}^n$. Suppose that $1+v^{\rot}A^{-1}u\neq 0$. Then it holds that
\[
(\mat{A}+u v^\rot)^{-1} = \mat{A}^{-1} - \frac{\mat{A}^{-1} u v^\rot\mat{A}^{-1}}{1+v^\rot \mat{A}^{-1} u}.
\]
\end{lemma}

\begin{lemma}[Lemma~3.5, \cite{KMST10:subsparsification}]\label{lem:pseud-inv SM}
For any   symmetric (possibly singular) matrix $A\in\mathbb{R}^{n\times n}$ and  $Y = vv^{\rot}$, it holds that 
  \[
    (A + PYP)^{\dag} = A^{\dag} - \frac{A^{\dag} Y A^{\dag}}{1 + v^{\rot} A^{\dag} v},
  \]
  where $P$ is the orthogonal projectin on $\mathrm{Im}(A)$.
\end{lemma}

 \begin{lemma}[Araki-Lieb-Thirring Inequality,  \cite{A07:audenaert2007araki}]\label{lem:ALT_ineq}
For $A, B \succeq 0, q \geq 0$ and for $r \geq 1$, the following inequality holds:
\[
  \tr (ABA)^{rq} \leq \tr(A^r B^r A^r)^q.
\]
\end{lemma}

\begin{lemma}[Corollary~3.9, \cite{KMST10:subsparsification}] \label{lem: majorisation}
  Let $A$ be a \textsf{PSD} matrix such that $\Tr(A) \leq T \in \mathbb{N}$ and $A \preceq I_n$. Then, for every symmetric positive semidefinite matrix $U$, it holds that
  \[
    U  \bullet A = \Tr(U A) \leq \sum_{i=n-T+1}^{n} \lambda_i(U).
  \]
\end{lemma}

\begin{lemma}[Lemma~3.6, \cite{KMST10:subsparsification}]\label{lem:majorisation eigenvalues}
    For every \textsf{PSD} matrix $A$, every projection $P$ and every $r \in \{1, \dots, n\}$ it holds that
    \[
        \sum_{i=n-r+1}^n \lambda_i(A) \geq \sum_{i=n-r+1}^n \lambda_i(PAP)
    \]
 \end{lemma}

\begin{lemma}\label{fact:PSD-invariant}
For any matrices $A,B$ satisfying $A \preceq B$ and any matrix $C$, it holds that $C^{\rot}AC \preceq C^{\rot}BC$. In particular,  when $C$ is symmetric, we have that $CAC \preceq CBC$.
 \end{lemma}

\begin{proof} 
Since it holds for any vector $w$ and $u=Cw$ that 
\[
w^{\rot}C^{\rot}ACw = u^{\rot} A u \leq u^{\rot} B u =w^{\rot}C^{\rot}BCw,
\]
the first statement holds. The second statement follows by the fact that $C$ is symmetric.
\end{proof}

\begin{lemma} \label{lem:proj perm}
For any \textsf{PSD} matrix $A$, a projection matrix $P$ and a positive integer $q$, the following statements hold:
  \begin{enumerate}
    \item $\lp PAP \rp^{q} \preceq P A^{q} P \preceq A^q$;
    \item $\Tr \lsp \lp AP \rp^q \lp PA\rp^q \rsp \leq \Tr \lsp PA^{2q}P\rsp.$
  \end{enumerate}
\end{lemma} 

\begin{proof} 
We prove the first statement by induction. The statement holds trivially when $q=1$. For $q=2$, we have that $A P A \preceq A^2$ because of $P \preceq I$. This implies that $ PAPAP \preceq PA^2P$. Combining this with the fact of $P^2=P$ proves the case of $q=2$.
For the inductive step, we assume that $q>2$ and have that 
\[
    \lp PAP\rp^q = PA \lp PAP \rp^{q-2} AP \preceq PA A^{q-2} AP = PA^q P,
\]
which the inequality above is based on the inductive hypothesis. With this we proved the first statement.

Now we prove the second statement. Because of  $P^2=P$ and the first statement, it holds that 
\[
    (AP)^q (PA)^q = A (PAP)^{2(q-1)} A \preceq  A P A^{2(q-1)}PA,
\]
which implies that 
\begin{align*}
    \Tr \lsp \lp AP \rp^q \lp PA\rp^q \rsp  
    &\leq \Tr \lsp A P A^{2(q-1)} P A\rsp
    = \Tr \lsp P A^2 P A^{2(q-1)} P\rsp 
    \leq \Tr \lsp A^2 P A^{2(q-1)}\rsp \\
    &= \Tr \lsp P A^{2q} P\rsp.
\end{align*}

With this we proved the second statement.
\end{proof}

      \begin{lemma}\label{lem:spec decomp preserved under unitary}
        For any matrix $A$ having spectral decomposition 
        \[
            A = \sum_{i=1}^k \lambda_i f_i f_i^{\rot}
        \]
        and a unitary matrix $U$ such that $U^{\rot}U = I$, the matrix $UAU^{\rot}$ has spectral decomposition
        \[
            UAU^{\rot} = \sum_{i=1}^k \lambda_i \lp Uf_i\rp \lp Uf_i\rp^{\rot}
        \]
    \end{lemma}

\begin{proof} 
        The statement  simply follows from the fact that the set $\{Uf_i\}_{i=1}^k$ forms   an orthonormal set.
    \end{proof}

\begin{lemma}\label{cor:trace ineq proj}
    Let $A$ be a \textsf{PSD} matrix and $P_L$ be the projection on the top $T$ eigenspace of $A$. Moreover, suppose that $u > \lambdaMax(A)$. Then, for any projection $P$ on a $T$-dimensional space, it holds that
    \[
        \Tr \lsp  P_L (u I - A)^{-1} P_L \rsp \geq \Tr \lsp P (u I - A)^{-1} P\rsp.
    \]
\end{lemma}

\begin{proof} 
We first recall the following  Karamata Majorisation inequality, which will be used in our proof:  for any  non-increasing sequences $x_1, x_2, \dots x_m$ and $y_1, y_2, \dots, y_m$ such that 
    \begin{equation}\label{eq:Karamata}
        \sum_{i=1}^r x_i \geq \sum_{i=1}^r y_i
    \end{equation}
      for any $1\leq r\leq m$, and a convex function $f$, it holds that 
    \[
        \sum_{i=1}^r f(x_i) \geq \sum_{i=1}^r f(y_i).
    \]
    We apply the inequality above by setting $x_i = \lambda_{n-i+1}(A)$,  $y_i = \lambda_{n-i+1}(PAP)$, and $m = T$. By  Lemma~\ref{lem:majorisation eigenvalues}, we know that 
    the two sequences $\{x_i\}$ and $\{y_i\}$ satisfy 
    \eqref{eq:Karamata}. We further set $f(x) = \frac{1}{u - x}$, and this gives us that 
    \begin{align*}
        \Tr \lsp P_L (u I - A)^{-1} P_L \rsp 
        & = \sum_{i=n-T+1}^n \frac{1}{u - \lambda_i(A)}
        = \sum_{i=n-T+1}^n f(\lambda_i(A)) 
        \geq \sum_{i=n-T+1}^n f(\lambda_i(PAP))\\
        & = \sum_{i=n-T+1}^n \frac{1}{u - \lambda_i(PAP)}
        =\Tr \lsp P (u I - A)^{-1} P \rsp,
    \end{align*}
    which proves the statement.
\end{proof}

%% file: AlgConn.tex
\section{A fast SDP solver }\label{sec:proofmainSDPlambda2}

We use the primal-dual framework introduced in \cite{AK16:combinatorial} to solve the  SDP~$\PSDP(G,W,k,\gamma)$ and prove Theorem~\ref{thm:mainSDPlambda2}. The framework is based on the matrix multiplicative weight update (MWU) algorithm on both the primal SDP $\PSDP(G,W,k,\gamma)$ and its dual $\DSDP(G,W,k,\gamma)$.


%

\paragraph*{Notation.}
  For any given  vector $\beta$, we use $\Diag(\beta)$ to denote the diagonal matrix such that each diagonal entry $[\Diag(\beta)]_{ii} = \beta_i$. Given  matrix $Z$, scalar $v$ and vector $\beta$, we use $\Diag(Z,v,\beta)$ to  denote the diagonal $3$-block matrix with blocks $Z$, $v$ and $\Diag(\beta)$.  We use  $I_V$ and $I_{E_W}$ to denote the identity matrices on vertex set $V$ and edge set $E_W$ with $|E_W|=m$, respectively. We further define 
\begin{eqnarray}
&& E\triangleq \Diag(\Delta\cdot I_V, m, I_{E_W}), \quad \Pi\triangleq \Diag(P_\bot, 1, I_{E_W}), \label{eqn:defEPIN}\\ &&N\triangleq \Diag(\Delta\cdot P_\bot, m, I_{E_W})=E^{1/2}\Pi E^{1/2}. \nonumber
\end{eqnarray}


For any given  parameter $\lambda$ and  vector $w$, we define \[V(\lambda, {w})\triangleq \lambda,  \qquad  
A(\lambda, {w}) \triangleq \mat{L}_G +\sum_{e\in E_W} w_e \mat{L}_e - \lambda \Delta \mat{P}_{\bot},  \qquad 
B(\lambda, {w}) \triangleq k- \sum_{e\in E_W}w_e.\] 
Let $c=c(\lambda,w)\in \mathbb{R}^{m}$ denote the vector with $c_e=1-w_e$ for each $e\in E_W$, and  $C=C(\lambda, {w})=\Diag(c(\lambda,w))$ be the diagonal $m\times m$ matrix with the diagonal entry $1- {w}_e$ corresponding to edge $e$. Then we define 
\begin{eqnarray}
M(\lambda, {w}) \triangleq\Diag\left(A(\lambda,w), B(\lambda,w), C(\lambda,w)\right)= \begin{bmatrix}
      A(\lambda, {w}) & 0 & 0           \\[0.3em]
      0& B(\lambda, {w})           & 0 \\[0.3em]
      0           & 0 & C(\lambda, {w})
     \end{bmatrix}.\label{eqn:M-definition}
\end{eqnarray}

	\begin{definition}
		An $(\ell,\rho)$-oracle for $\DSDP(G,W,k,\gamma)$ is an algorithm that on input $\langle Z,v,\beta\rangle$ with $ \Diag(Z,v,\beta) \bullet N  =1$, either fails or outputs $(\lambda, {w})$ with $\lambda\geq 0$, $ {w}\in \mathbb{R}_{\geq 0}^{m}$ that satisfies
\[
V(\lambda, {w})\geq \gamma, \qquad
A(\lambda, {w}) \bullet Z + B(\lambda, {w}) \cdot v +  c(\lambda, {w}) \cdot \beta \geq 0, 
\qquad-\ell N \preceq M(\lambda, {w})\preceq \rho N.
\]
	\end{definition}
	
We have the following simple fact and we defer its proof to Appendix~\ref{sec:OmittedSDP}.
\begin{fact}\label{fact:infeasibledual}
	If an $(\ell,\rho)$-oracle for $\DSDP(G,W,k,\gamma)$ does not fail on input $\langle Z, v,\beta \rangle$ with \[ \Diag(Z,v,\beta) \bullet N  =1,\] then $\langle Z, v,\beta \rangle$ is infeasible for $\DSDP(G,W,k,\gamma)$.
\end{fact}


In order to apply the MWU algorithm, in the following we use the $U_{\varepsilon}(A)$ to denote the matrix 
			\[
			U_\varepsilon(A)\triangleq \frac{E^{-1/2}(1-\varepsilon)^{E^{-1/2}AE^{-1/2}}E^{-1/2}}{\Pi\bullet (1-\varepsilon)^{E^{-1/2}AE^{-1/2}}},
		\]
		where $E, \Pi$ are matrices as defined in Equation (\ref{eqn:defEPIN}).  


\subsection{The MWU algorithm} 
In the framework of MWU for solving our SDP, we sequentially produce candidate dual solutions $\langle Z^{(t)},v^{(t)},\beta^{(t)}\rangle$ such that $\Diag(Z^{(t)},v^{(t)},\beta^{(t)})\bullet N=1$ for any $t$. Specifically, for any given $k,\gamma$, we start with a solution $Z^{(1)}=\frac{1}{\Delta(n-1)}I$, $v^{(1)}=\frac{2}{n-1}$ and $\beta_e^{(1)}=0$ for any $e\in E_W$. At each iteration $t$, we invoke a good separation oracle that takes $\Diag(Z^{(t)},v^{(t)},\beta^{(t)})$ as input, and then either guarantees that $\Diag(Z^{(t)},v^{(t)},\beta^{(t)})$ is already good for dual SDP (and thus certifies infeasibility of primal SDP), or outputs $(\lambda^{(t)},w^{(t)})$ certifying the infeasibility of $\Diag(Z^{(t)},v^{(t)},\beta^{(t)})$.

If $\left(\lambda^{(t)},w^{(t)}\right)$ is returned by the oracle, then the algorithm   updates the next candidate solution based on 
\begin{eqnarray*}
X^{(t)} = U_{\varepsilon} \left(\frac{1}{2\rho}\sum_{s=1}^{t-1}M^{(s)}\right), 
\end{eqnarray*}
where $M^{(s)} \triangleq M\left(\lambda^{(s)},{w}^{(s)}\right)$ is as  defined before and  $\varepsilon$ is a parameter of the algorithm. By definition, we have that 
$X^{(t)}\bullet N=1$. Moreover, since $M^{(t)}$ can be viewed as a $3$-block diagonal matrix with diagonal entries $A^{(t)},B^{(t)},C^{(t)}$,  $\exp(M^{(t)})=\Diag\left(\exp\left(A^{(t)}\right),\exp\left(B^{(t)}\right),\exp\left(C^{(t)}\right)\right)$. Therefore, we can decompose $X^{(t)}$ as $$X^{(t)}=\Diag(Z^{(t)},v^{(t)},\beta^{(t)}).$$ 
Note that $X^{(t)}\bullet N = 1$ is equivalent to
\[
\Delta \cdot Z^{(t)}\bullet P_\bot +m\cdot  v^{(t)} +\sum_{e\in E_W}\beta_e^{(t)}
=1.\]

The following theorem guarantees that, after a small number of iterations, the algorithm either finds a good enough dual solution, or a feasible solution to the primal SDP. The proof of the theorem is built upon a result in \cite{Orecchia:11} and is deferred to Appendix \ref{sec:OmittedSDP}.
	\begin{theorem}\label{thm:sdp_MWU}
Let \textsc{Oracle} be an $(\ell,\rho)$-oracle for $\DSDP(G,W,k,\gamma)$, and let $\delta>0$. Let $N$,   $X^{(t)}$, and  $M^{(t)}$ be   defined as  above, for any $t\geq 1$. Let $\varepsilon=\min\{1/2,\delta/2\ell \}$. 
Suppose that  \textsc{Oracle} does not fail for $T$ rounds, where 
\begin{eqnarray*}
T=O\left(\frac{\rho \log n}{\delta \varepsilon}\right) \leq \max\left\{O\left(\frac{\rho\log n}{\delta}\right), O\left(\frac{\rho \ell \log n }{\delta^2}\right) \right\}, 
	\end{eqnarray*}
		then  $(\bar{\lambda}-{3\delta },\bar{w}-{\delta})$ is a feasible solution to $\PSDP(G,W,k,\gamma-{3\delta })$, where
		$	\bar{\lambda}\triangleq\frac{1}{T}\sum_{t=1}^T\lambda^{(t)}$ and $ \bar{w}\triangleq\frac{1}{T}\sum_{t=1}^{T}w^{(t)}. 
		$
	\end{theorem}
%

\paragraph*{Approximate computation.} By applying the Johnson-Linderstrauss (JL) dimensionality  reduction to the embedding corresponding to $U_\varepsilon$, we can approximate $X^{(t+1)}$ by $\tilde{X}^{(t+1)}$ while preserving the relevant properties. Specifically, let $\tilde{U}_\varepsilon$ be a randomised approximation to $U_\varepsilon$ from applying the JL Lemma (see \cite{OV11:towards}), and we  compute in nearly-linear time the matrix 
$\tilde{X}^{(t+1)} =\tilde{U}_{\varepsilon} \left(\frac{1}{2\rho}\sum_{s=1}^{t-1}M^{(s)}\right)$ and 
decompose it into $3$ blocks:
$$\tilde{X}^{(t+1)}=\Diag\left(\tilde{Z}^{(t+1)},\tilde{v}^{(t+1)},\Diag\left(\tilde{\beta}^{(t+1)}\right)\right).$$ Moreover, $\tilde{X}^{(t+1)}\bullet L_H$ well approximates $X^{(t+1)}\bullet L_H$ for any graph $H$, which suffices for our oracle. Hence, we assume that the oracle receives $\tilde{X}^{(t+1)}$ as input instead of $X^{(t+1)}$. 

Formally, we need the following lemma which follows directly from Lemma 2.3 in \cite{OV11:towards}.

\begin{lemma}[Lemma 2.3, \cite{OV11:towards}]\label{lem:ov11}
Let $\varepsilon>0$ be a sufficiently small constant. Let $M\in \mathbb{R}^{(n+1+m)\times (n+1+m)}, M\succeq 0$ be a matrix such that $M=\Diag(A,B,C)$, where $A\in \mathbb{R}^{n\times n}, B \in \mathbb{R}$ and $C\in \mathbb{R}^{m\times m}$ is a diagonal matrix. Let $\tilde{X}=\tilde{U}_\varepsilon(M)$, and $X=U_\varepsilon(M)$. Then, the following statements hold:
\begin{enumerate}
\item $\tilde{X}\succeq 0$ and $\tilde{X}\bullet N=1$. Furthermore, $\tilde{X}$ can be decomposed into $3$ blocks such that $\tilde{X}=\Diag(\tilde{Z},\tilde{v},\Diag(\tilde{\beta}))$.
\item The embedding $\left\{\tilde{v}_i\in \mathbb{R}^d: i\in V\right\}$ corresponding to $\tilde{X}$ can be represented in $d=O(\log n)$ dimensions. 
\item The embedding $\left\{\tilde{v}_i\in \mathbb{R}^d: i\in V\right\}$ can be computed in   $\tilde{O}(t_M+n+m)$ time, where   $t_M$ is the running time for performing matrix-vector multiplication by $M$. 
\item For any graph $H=(V,E_H)$, with high probability, we have 
$$\left(1-\frac{1}{64}\right) L_H\bullet Z - \tau \leq L_H\bullet \tilde{Z} \leq \left(1+\frac{1}{64}\right) L_H\bullet Z+\tau$$
where $\tau=O(1/\poly(n))$ and $Z$ is the block matrix in the decomposition of $X$ such that $X=\Diag(Z,v,\Diag(\beta))$. 
\end{enumerate}
\end{lemma}

In our algorithm, we only need to compute $L_H\bullet \tilde{Z}$ and all the diagonal entries on $\tilde{v}$ and $\tilde{\beta}$. These quantities can be computed efficiently by the above embedding.

%

\subsection{The oracle} Now we are ready to present the oracle for our SDP $\DSDP(G,W,k,\gamma)$,  and our result is summarised as follows:
\begin{theorem}\label{thm:sdporacle}
On input $\tilde{X}^{(t)}$, there exists an algorithm \textsc{Oracle} that runs in   $\tilde{O}(n+m)$ time and is an $(\ell,\rho)$-oracle for SDP $\DSDP(G,W,k,\gamma)$, where $\ell=1$ and $ \rho=3$. 
\end{theorem}

For the simplicity of presentation, we abuse notation and  use $X=\Diag(Z,v,\beta)$ to denote the input to the oracle, although it should be clear that the input is the approximate embedding $\tilde{X}=\Diag(\tilde{Z},\tilde{v},\Diag(\tilde{\beta}))$ of $X$. Our oracle is described in Algorithm \ref{alg:oracle_sdp}. 
\begin{algorithm}[!h]
	\caption{\textsc{Oracle} for SDP $\DSDP(G,W,k,\gamma)$}\label{alg:oracle_sdp}
	\begin{algorithmic}[1]
		\Require Candidate solution $\langle Z,v, \beta\rangle$ with ${\Delta}\cdot Z\bullet P_\bot + m\cdot v +\sum_{e\in E_W}\beta_e=1$, target value $\gamma$
\State{Let $B:=\{e: v+ \beta_e<L_e\bullet Z \}$, $\Gamma:=\sum_{e\in B} (L_e\bullet Z - v- \beta_e)$, and $T:=Z\bullet \Delta P_{\bot}$.}
\State{Let $T_{\textrm{tol}}:= L_G\bullet Z + k v + \sum_{e\in E_W} \beta_e$.}
\If{$\Gamma\leq T\gamma-T_{\textrm{tol}}$}
		\State{Output ``fail''.}
		\Comment{In this case, $\langle Z,v, \beta\rangle$ is ``good'' enough}
\ElsIf{$T_{\textrm{tol}}>\gamma m- {\gamma}\sum_{e\in E_W} Z\bullet L_e$}
		\State{\Return $w_e={\gamma}$, and $\lambda=\gamma$.} 
\Else     
		\State{\Return $w_e=1$ for $e\in B$, $w_e=0$ for $e\in E_W\setminus B$ and $\lambda=\gamma$}    
		\EndIf
	\end{algorithmic}
\end{algorithm}

In order to prove Theorem~\ref{thm:sdporacle}, we give two lemmas in the following. 
We first show that if the \textsc{Oracle} fails, then we can find a dual feasible solution for $\DSDP(G,W,k,\gamma)$. 

\begin{lemma}\label{lem:goodsolu}
Let $\langle Z,v, \beta\rangle$ be a candidate solution. 
Suppose that 
for 
\[
B\triangleq\{e:  v+\beta_e-L_e\bullet Z  < 0 \}, \quad T\triangleq Z\bullet \Delta P_\bot, \quad T_{\textrm{tol}}\triangleq L_G \bullet Z + k v + \sum_{e\in E_W} \beta_e,\] 
it holds that 
\[
\Gamma\triangleq\sum_{e\in B}  (L_e\bullet Z - v- \beta_e)\leq T\gamma-T_{\textrm{tol}}.
\] 
Moreover, by setting  $Z'=\frac{Z}{T}$, $v'=\frac{v}{T}$, and $\beta_e'= \frac{\beta_e}{T}$ if $e\in E_W\setminus B$ and $\beta_e'=\frac{L_e\bullet Z- v}{T}$ if $e\in B$, we have that  $\langle Z',v', \beta'\rangle$ is a dual feasible for $\DSDP(G,W,k,\gamma)$.%
\end{lemma}
\begin{proof}
By definition, it holds that  $Z'\bullet \Delta P_\bot=\frac{Z\bullet \Delta P_\bot}{T}=1$, as $T=Z\bullet \Delta P_\bot$. Moreover, we have  \[Z'\bullet L_e = \frac{Z\bullet L_e}{T}\leq \frac{v}{T}+\frac{\beta_e}{T} = v'+\beta_e'\] for any $e\in E_W\setminus B$ and \[Z'\bullet L_e = \frac{Z\bullet L_e}{T}= \frac{v}{T} + \frac{L_e\bullet Z -v}{T}= v'+\beta_e'\] for $e\in B$. 
We also note that
\begin{eqnarray*}
&&Z'\bullet L_G+kv'+\sum_{e\in E_W}\beta_e'\\
&=&\frac{Z}{T}\bullet L_G+k\frac{v}{T}+\sum_{e\in E_W\setminus B}\frac{\beta_e}{T} + \sum_{e\in B}\frac{L_e\bullet Z-v}{T}\\
&=&\frac{1}{T}\left(Z\bullet L_G+kv+\sum_{e\in E_W\setminus B}\beta_e + \sum_{e\in B} (\beta_e + L_e\bullet Z- v -\beta_e)\right)\\
&=&\frac{1}{T}\left(  T_{\textrm{tol}}+ \sum_{e\in B} (L_e\bullet Z- v -\beta_e)\right) \\
&\leq& \frac{1}{T}(T_{\textrm{tol}} + T\gamma-T_{\textrm{tol}})\\
&=& 
\gamma,
\end{eqnarray*}
where the last inequality follows by our assumption that $\sum_{e\in B}(L_e\bullet Z-v-\beta_e)\leq T \gamma-T_{\textrm{tol}}$.
%
%
%
%
\end{proof}

We then show that if \textsc{Oracle} does not fail, then it returns $(\lambda,w)$ that satisfies the properties of $(\ell,\rho)$-oracle for $\DSDP(G,W,k,\gamma)$ for appropriate $\ell, \rho$. 
\begin{lemma}\label{lemma:width-oracle}
When \textsc{Oracle} described in the algorithm does not fail, it returns a vector $w$ and value $\lambda$ such that $V(\lambda,w)\geq \gamma$, and for the matrix 
$M(\lambda,w)=\Diag(A(\lambda,w), B(\lambda,w),C(\lambda,w))$, it holds that \[A(\lambda,w)\bullet Z+B(\lambda,w)\cdot v+C(\lambda,w)\cdot \beta\geq 0.\] Moreover, it holds that  $-N\preceq M(\lambda,w)\preceq 3N$.
\end{lemma}

\begin{proof}
When \textsc{Oracle} does not fail, it returns a vector $w$ and a value $\lambda$. The proof is based on case distinction. 

\textbf{Case 1:} Consider the case that $$T_{\textrm{tol}}=L_G \bullet Z + k v + \sum_{e\in E_W} \beta_e>{\gamma} m- {\gamma}\sum_{e\in E_W} Z\bullet L_e.$$ Then, we have  by the algorithm description that $w_e={\gamma}$
and $\lambda=\gamma$. Then $V(\lambda,w)=\lambda= 
\gamma$. 
Recall that $\Delta\cdot Z\bullet P_\bot + m\cdot v +{\sum_{e\in E_W}\beta_e}=1$.  
Now we have 
\begin{eqnarray*}
&&A(\lambda,w)\bullet Z+B(\lambda,w)v+C(\lambda,w)\cdot \beta \\
&=& \left(L_G+\sum_{e\in E_W}w_e L_e-\lambda\Delta P_\bot\right) \bullet Z + \left(k-\sum_{e\in E_W}w_e\right)v + \sum_{e\in E_W} (1-w_e)\cdot \beta_e \\
&= &Z\bullet  L_G+kv+\sum_{e \in E_W}\beta_e - \sum_{e \in E_W} w_e\cdot (v+\beta_e-Z\bullet L_e)-\lambda \Delta P_\bot \bullet Z\\
&=&T_{\textrm{tol}} - {\gamma}\left(\sum_{e\in E_W}(v+\beta_e) +\Delta\cdot P_\bot \bullet Z\right) +{\gamma}\sum_{e\in E_W} Z\bullet L_e \\
&=& T_{\textrm{tol}} -{\gamma} m+ {\gamma}\sum_{e\in E_W} Z\bullet L_e>0,
\end{eqnarray*}
by the assumption on $T_{\textrm{tol}}$. 
Since both $G=(V,E)$ and $(V,E_W)$ have the maximum degree at most $\Delta$, we have that
$$A(\lambda,w)=L_G+\sum_{e\in E_W}w_e L_e -\lambda \Delta P_\bot \\
\preceq 2\Delta P_\bot + 2\gamma \Delta P_{\bot} - \gamma \Delta P_\bot \preceq 3\Delta P_\bot ,$$
and 
$$ -\gamma \Delta P_\bot = -\lambda\Delta P_\bot \preceq L_G+\sum_{e\in E_W}w_e L_e -\lambda P_\bot =A(\lambda,w).$$
Furthermore, $-\gamma m< B(\lambda,w)=k-\gamma m= k-\sum_{e\in E_W}w_e\leq k\leq m$, and $0 < 1-w_e \leq 1 $ for any $e\in E_W$. Thus, $0\preceq C(\lambda,w)\preceq {I_{E_W}}$. Therefore, it holds that $-\gamma N \preceq M(\lambda,w)\preceq 3 N$.


\textbf{Case 2:} Consider the case that $$T_{\textrm{tol}}=L_G \bullet Z + k v + \sum_{e\in E_W} \beta_e\leq {\gamma} m- {\gamma}\sum_{e\in E_W} Z\bullet L_e,$$ and $\Gamma =\sum_{e\in B} (L_e\cdot Z - v- \beta_e) >T\gamma-T_{\textrm{tol}}$, where $B\triangleq\{e: v+ \beta_e < L_e\bullet Z \}$, and $T\triangleq Z\bullet \Delta P_\bot$. Then, by our algorithm, it holds that   $w_e=1$ for $e\in B$ and $w_e=0$ for $e\in E_W\setminus B$, and $\lambda=\gamma$. Then $V(\lambda,w)=\lambda= 
\gamma$.

Now we have that
\begin{eqnarray*}
&&A(\lambda,w)\bullet Z+B(\lambda,w)v+C(\lambda,w)\cdot \beta\\
&= &Z\bullet L_G+kv+\sum_{e \in E_W}\beta_e - \sum_{e \in E_W} w_e\cdot (v+\beta_e-Z\bullet L_e)-\lambda \Delta P_\bot \bullet Z\\
	&=&T_{\textrm{tol}}+ \sum_{e\in B}\left(Z\bullet L_e - v -\beta_e\right) -\gamma T\\
	&= & T_{\textrm{tol}}+\Gamma-\gamma T\\
	&>&T_{\textrm{tol}} + T\gamma  - T_{\textrm{tol}} -\gamma T > 0,
\end{eqnarray*}
where the last inequality follows from the fact that $\Gamma > T\gamma- T_{\textrm{tol}}$.
Furthermore, 
$$A(\lambda, {w})=\mat{L}_G +\sum_{e\in E_W} w_e \mat{L}_e - \lambda \mat{P}_{\bot}\preceq 3\Delta P_\bot, \quad A(\lambda, {w})\succeq -\gamma P_\bot.$$ 
In addition, we have $$-m\leq -\sum_{e\in E_W}w_e\leq {B(\lambda,w)}=k-\sum_{e\in E_W}w_e\leq k\leq m,$$ $\beta_e=1-w_e=1$ if $w\in E_W\setminus B$ and $\beta_e=1-w_e=0$ if $w\in B$. Therefore, it holds that $- N\preceq M(\lambda,w)\preceq 3 N$.
\end{proof}

Finally, we note that Theorem \ref{thm:sdporacle} will follow from the above two lemmas.


%

\subsection{Proof of Theorem \ref{thm:mainSDPlambda2}}
\begin{proof}[Proof of Theorem \ref{thm:mainSDPlambda2}]
Let $\delta'>0$ be any constant. We specify $\delta=\frac{\delta'\gamma}{3}$ in our MWU algorithm, which is  described in the previous subsections. 
 We set $\rho=3$ and $\ell=1$, and let
\[
T\triangleq O\left(\frac{\rho\ell\log n}{\delta^2}\right)=O\left(\frac{\log n}{(\delta')^2\gamma^2}\right)=O\left(\frac{\log n}{\gamma^2}\right).
\]
In the MWU algorithm, if the \textsc{Oracle} fails in the $t$-th iteration for some $1\leq t\leq T$, then the corresponding embedding $\tilde{X}^{(t)}=\Diag(\tilde{Z}^{(t)},\tilde{v}^{(t)},\Diag(\beta^{(t)}))$ provides a good enough solution:  the precondition of Lemma~\ref{lem:goodsolu} is satisfied, which further implies that $\tilde{X}^{(t)}$ can be turned into a dual feasible solution with objective at most $\gamma$, i.e., we find a solution to $\DSDP(G,W,k,\gamma)$. Therefore, the primal SDP $\PSDP(G,W,k,\gamma)$ is infeasible.

Otherwise, the \textsc{Oracle} does not fail for $T$ iterations, and by Theorem~\ref{thm:sdp_MWU} and Lemma~\ref{lemma:width-oracle} we know that for 
\[\bar{\lambda}\triangleq \frac{1}{T}\sum_{t=1}^T\lambda^{(t)}, \bar{w}\triangleq \frac{1}{T}\sum_{t=1}^{T}w^{(t)},\] $(\bar{\lambda}-{3\delta }, \qquad\qquad \bar{w}-{\delta})$ is a feasible solution for $\PSDP(G,W,k,\gamma-{3\delta })=\PSDP(G,W,k,\gamma-\delta'\gamma)$. 

For the running time, by Lemma \ref{lem:ov11}, 
  in each iteration  the \textsc{Oracle} can be implemented in $\tilde{O}(n+m)$ time and the approximation embedding can be found in $\tilde{O}(t_M+n+m)=\tilde{O}(n+m)$ time, as $t_M=\tilde{O}(|E(G)|+m+n)=\tilde{O}(n+m)$. Thus, in $\tilde{O}((n+m)/\gamma^2)$ time we either find a solution to our SDP with objective value at least $(1-\delta')\gamma$, for any constant $\gamma'>0$ or we certify that there is $\PSDP(G,W,k,\gamma)$ is infeasible (in case the \textsc{Oracle} fails.)  
\end{proof}

%% file: Subgraph_Sparsification_Description.tex
\section{Subgraph sparsificaion \label{sec:alg}}


 Now we give an overview of  our efficient algorithm for constructing subgraph sparsifiers, and discuss its connection to spectral sparsifiers. 
 Recall that,  for any $k\in\mathbb{N}$, parameter $\kappa\geq 1$, and two weighted graphs  $G=(V, E)$ and
  $W=(V, E_W)$, the subgraph sparsification problem  is to find a
  set $F\subseteq E_W$ of $|F|=O(k)$ edges with weights $\{w_e\}_{e\in F}$, such that the resulting graph $H=(V, E+F)$ is a $\kappa$-approximation of $G+W$, i.e.,  
  \begin{equation}\label{eq:kapt}
     L_{G+W}  \preceq    L_{G}  +\sum_{e\in F} w_e b_e b_e^{\rot}   \preceq \kappa\cdot    L_{G+W}.   
  \end{equation}
To  construct the required edge set $F$, we apply the standard reduction for constructing graph sparsifiers by setting $v_e\triangleq L_{G+W}^{\dagger/2} b_e$ for every $e\in E_W$, and  (\ref{eq:kapt})
is equivalent to  
\[
    I_{\mathsf{im}(L_{G+W})}\preceq L^{\dagger/2}_{G+W} L_G L^{\dagger/2}_{G+W} 
    + \sum_{e\in F} w_e v_e v_e^{\rot} \preceq \kappa\cdot I_{\mathsf{im}(L_{G+W})},
\]
where $I_{\mathsf{im}(L_{G+W})}$ is the identity on $\mathsf{im}(L_{G+W})$. Our main result is summarised as follows:


\begin{theorem}\label{thm:main-formal}
Let $\varepsilon$ and $q$ be arbitrary constants such that
$\varepsilon\leq 1/20$ and $q\geq 10$.
Then, there is a randomised algorithm such that, for any two graphs $G=(V,E)$ and $W=(V,E_W)$ defined on the same vertex set as input, by defining $X=\left(L^{\dagger/2}_{G+W} L_G L^{\dagger/2}_{G+W}\right)\Big|_{\mathrm{Im}(L_{G+W})}$ and $\overline{M}\triangleq \left(\sum_{i=1}^m v_i v_i^{\rot}\right)|_{\mathrm{Im}(L_{G+W})} $  where every $v_i$ is of the form $L_{G+W}^{\dagger/2} b_e$ for some edge $e\in E_W$,   the algorithm outputs a set of non-negative coefficients $\{c_i\}_{i=1}^m$ with
$|\{ c_i\ |\ c_i\neq 0\} |=K$ for some $K=O \left(qk/\varepsilon^2\right)$
such that  it holds for some constant $C$ that
\[
C \cdot(1-O(\varepsilon))\cdot \min\{ 1,  K/T\}\cdot \lambda_{k+1}(X)\cdot I\preceq  
X+ \sum_{i=1}^m c_i  v_iv_i^{\rot} 
\preceq  (1 + O(\varepsilon))\cdot I,
\]
where $T\triangleq \left\lceil \mathrm{tr}\left(\overline{M}\right)\right\rceil$. Moreover, if we assume that   every $v_i$ is associated with some cost denoted by $cost_i$ such that  $\sum_{i=1}^m cost_i=1$, then with constant probability the coefficients $\{c_i\}_{i=1}^m$ returned by the algorithm satisfy  
$\sum_{i=1}^m c_i\cdot cost_i \leq O(1/\varepsilon^2)\cdot  \min\{ 1, k/T\}$. The algorithm runs in time
\[
    \widetilde{O} \lp \min \left\{ n^{\omega},  \frac{mk + nk^2}{\sqrt{\lambda_{k+1}(X)}}\right\} + q\cdot n^{O(1/q)}
    \lp \frac{mn^{2/q}}{\varepsilon^{2+2/q}} + \min \left\{ n^{\omega}, mk + nk^2 + k^{\omega}\right\} \rp \Big/\varepsilon^5 \rp.
\]
\end{theorem}

Without loss of generality, we assume that $\overline{M}$ is a full-rank matrix, which can be achieved by adding $n$ self-loops each of small weight $\gamma=\Theta(1/\mathrm{poly}(n))$, so that with constant probability  these self-loops will not be sampled by the algorithm.

\subsection{Overview of our algorithm\label{sec:BSS}}

\paragraph*{The BSS framework.}
At a high level, our algorithm follows the  BSS framework for constructing spectral sparsifiers~\cite{BSS}. The BSS algorithm proceeds by iterations: in each iteration $j$ the algorithm chooses one or more vectors, denoted by $v_{j_1},\ldots, v_{j_k}$, and adds $\Delta_j = \sum_{i=1}^{k} c_{j_i}v_{j_i} v_{j_i}^{\rot}$ to the currently constructed matrix by setting $A_j = A_{j-1} + \Delta_j$, where $c_{j_1},\ldots, c_{j_k}$ are scaling factors, and $A_0=\mathbf{0}$ initially. Moreover, two barrier values, the \emph{upper barrier} $u_j$ and the \emph{lower barrier} $\ell_j$, are maintained such that the constructed ellipsoid  $\mathsf{Ellip}(A_j)$  is sandwiched between the outer sphere $u_j\cdot I$ and the inner sphere $\ell_j\cdot I$ for any iteration $j$. To ensure this, all the previous analysis uses a potential function $\Phi(A_j, u_j, \ell_j)$ defined by
$$
\Phi\left(A_j, u_j, \ell_j\right) = \tr [f(u_j I - A_j)]  + \tr [f( A_j -\ell_j I)]
$$
for some function $f$, and a bounded value of $\Phi\left(A_j, u_j, \ell_j\right)$ implies that
\begin{equation}\label{eq:two-sided}
\ell_j\cdot I \prec A_j \prec u_j\cdot I.
\end{equation}
After each iteration,  the two barrier values $\ell_j$ and $u_j$ are increased properly   by setting
\[
u_{j+1} = u_j + \delta_{u,j}, \qquad 
\ell_{j+1 } = \ell_j + \delta_{\ell,j}
\]
for some positive values $\delta_{u,j}$ and $\delta_{\ell,j}$. The careful choice of $\delta_{u,j}$ and $\delta_{\ell,j}$ ensures that    after $\tau$ iterations $\mathsf{Ellip}(A_{\tau})$ is close to being a sphere, which implies that $A_{\tau}$ is a spectral sparsifier of $\ell_{\tau}\cdot I$, see Figure~\ref{fig:bsspic} for illustration.


\begin{center}
\begin{figure}[t]
\centering
\begin{tikzpicture}[xscale=0.35,yscale=0.35]

   
  \filldraw[ball color = cyan, opacity = 1] (0,0) ellipse (1.2cm and 1.9cm);
   \draw (0,0) ellipse (1.2cm and 1.9cm);
 

  \draw (-2,0) arc (180:360:2 and 0.6);
  
  \draw[dashed] (2,0) arc (0:180:2 and 0.6);
  \filldraw[ball color = MyMango, opacity = 0.9] (0,0) circle (1cm);
  \draw (0,0) circle (1cm);
  
  \draw (-1,0) arc (180:360:1 and 0.3);
  \draw[dashed] (1,0) arc (0:180:1 and 0.3);
  
   \draw (0,0) circle (2cm);
   \filldraw[ball color = white, opacity = 0.1] (0,0) circle (2cm);
 
\node () at (0,-5) {\footnotesize{Iteration $j$}};


   \filldraw[ball color = cyan, opacity = 1] (8,0) ellipse (2.2cm and 1.7cm);
   \draw (8,0) ellipse (2.2cm and 1.7cm);
 

 \draw[-stealth, color=black!70, line width=2pt] (3,0) -- (4.5,0);

 \draw[dashed] (10.5,0) arc (0:180:2.5 and 0.75);
  \filldraw[ball color = MyMango, opacity = 0.9] (8,0) circle (1.5cm);
  \draw (8,0) circle (1.5cm);
  \draw (6.5,0) arc (180:360:1.5 and 0.45);
  \draw[dashed] (9.5,0) arc (0:180:1.5 and 0.45);
  \draw (8,0) circle (2.5cm);
  \filldraw[ball color = white, opacity = 0.1] (8,0) circle (2.5cm);
  \draw[color=black!80] (5.5,0) arc (180:360:2.5 and 0.75);

\node () at (8,-5) {
\footnotesize{Iteration $j+1$}
};



  \filldraw[ball color = cyan, opacity = 1] (22,0) ellipse (3.3cm and 3.7cm);
   \draw (22,0) ellipse (3.3cm and 3.7cm);


 \draw[-stealth, color=black!70, line width=2pt] (12,0) -- (13.5,0);

\draw[-stealth, color=black!70, line width=2pt] (15,0) -- (16.5,0);

 \draw[dashed] (26,0) arc (0:180:4 and 1.2);
  
 \filldraw[ball color = MyMango, opacity = 0.9] (22,0) circle (3cm);
 \draw (22,0) circle (3cm);
  \draw (19,0) arc (180:360:3 and 0.9);
  \draw[dashed] (25,0) arc (0:180: 3 and 0.9);
  \draw (22,0) circle (4cm);
  \filldraw[ball color = white, opacity = 0.1] (22,0) circle (4cm);
  
   \draw[color=black!80] (18,0) arc (180:360:4 and 1.2);

\node () at (22,-5) {
\footnotesize{Final iteration $\tau$}
};

\end{tikzpicture}
\caption{Illustration of the BSS framework: the light grey  and orange balls in  iteration $j$ represent the spheres $u_j\cdot I$ and $\ell_j\cdot I$, and the cyan  ellipsoid sandwiched between the two balls corresponds to the constructed ellipsoid in iteration $j$. After each iteration $j$,  the algorithm increases the value of $\ell_j$ and $u_j$
by some $\delta_{\ell,j}$ and $\delta_{u,j}$
 so that the invariant \eqref{eq:two-sided}  holds in iteration $j+1$. 
This process is repeated for $\tau$ iterations, so that the final constructed ellipsoid  is close to be a sphere.
\label{fig:bsspic}} 

\end{figure}
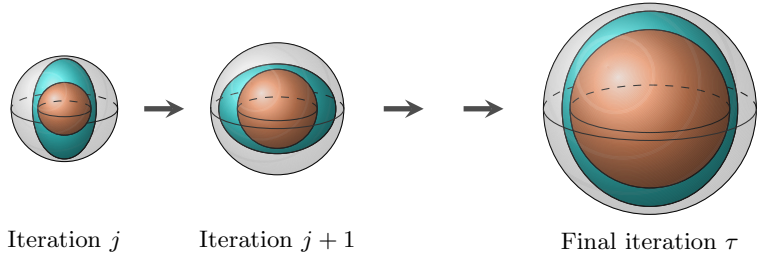

\end{center}

\vspace{-0.9cm}


\paragraph*{The BSS framework for subgraph sparsification.}
The BSS framework ensures that, when starting with the zero matrix, after choosing $O(n)$ vectors, the final constructed matrix is close to $I$. 
However, applying the BSS framework to construct   a subgraph sparsifier is significantly  more challenging due to the following two  reasons:
\begin{itemize}
\item Instead of starting with the zero matrix, we need to start with some \emph{non-zero} matrix $A_0=X$, and the number of added vectors is $K=O(k)$, which could be much smaller than $n$. This implies that the ellipsoid corresponding to the final constructed matrix could be still very far from being a sphere.
\item Because of this and every rank-one update has  different contribution towards each direction in $\mathbb{R}^n$, to ``optimise'' the contribution of $O(k)$ rank-one updates we have to ensure that
the added vectors will significantly benefit the ``worst subspace'', the subspace in $\mathbb{R}^n$ that limits the approximation ratio of the final constructed sparsifier. 
\end{itemize}
To address these two challenges, in the celebrated paper 
 Kolla et al.~\cite{KMST10:subsparsification}  propose to keep track of the algorithm's progress  with respect to two subspaces, each of which is measured by some potential function.  Specifically, in each iteration $j$ they define
$
A_j \triangleq X+\sum_{i} c_i v_iv_i^{\rot}$,
where $\sum_i c_i v_iv_i^{\rot}$ is the sum of currently picked rank-one matrices after reweighting during the first $j$ iterations. 
For the upper barrier value $u_j$ in iteration $j$,  they define the upper potential function
\[
\Phi^{u_j} (A_j) \triangleq \tr\left( P_{L(A_j)} \left( u_j I - A_j \right)  P_{L(A_j)} \right)^{\dagger},
\]
where $L(A_j)$ is the $T$-dimensional subspace of $A_j$ spanned by the $T$ largest eigenvectors of $A_j$  and $\PL{A_j}$ is the projection onto that subspace. 
Notice that   $\Phi^{u_j}(A_j)$ is defined with respect to a variable space $L(A_j)$ that \emph{changes after every rank-one update}, in order to upper bound the maximum eigenvalue of the final constructed matrix in the entire space. 
Similarly, for the same matrix $A_j$ and
lower barrier $\ell_j$ in iteration $j$, they define the lower potential function by
\[
\Phi_{\ell_j}(B_j) \triangleq \tr\left( P_S (B_j - \ell_j I)P_S \right)^{\dagger},
\]
where $P_S$ is the orthogonal projection onto $S$, the subspace generated by the bottom $k$ eigenvectors of $X$, and 
the matrix $B_j$ is defined by 
$
B_j = Z(A_j-X)Z,
$ 
for
$
Z  = \left(P_S (I-X) P_S\right)^{\dag/2}.
$
Since the total number of chosen vectors is $K=O(k)$, instead of expecting the final constructed matrix $A_{\tau}$   approximating the identity matrix, the objective of the subgraph sparsification is to   find coefficients $\{c_i\}$ with $K = O(k)$ non-zeros such that the following two conditions hold:
\begin{enumerate}
  \item $X + \sum_{i=1}^m c_i v_i v_i^{\rot} \preceq \theta_{\mathrm{max}} I$, and
  \item $\sum_{i=1}^m c_i Z v_iv_i^{\rot}Z \succeq \theta_{\mathrm{min}} P_S$,
\end{enumerate}
for some positive constants $\theta_{\mathrm{min}}, \theta_{\mathrm{max}}$. Informally, the first condition above states that  the length of any axis of $\mathsf{Ellip}(A_j)$ is upper bounded, and the second condition ensures that the final matrix $A_{\tau}$ has significant contribution towards the bottom $k$ eigenspace $X$.
In other words,  
instead of ensuring $\ell_j \cdot I\prec A_j\prec u_j\cdot I$, 
$\Phi^{u_j} (A_j)$ and  $\Phi_{\ell_j}(B_j) $ are used to ``quantify'' the shapes of the two ellipsoids with different dimensions:
\begin{itemize}\itemsep -0.4pt
    \item  The function 
$\Phi^{u_j} (A_j)$ studies the ellipsoid $A_j$ projected onto its own top eigenspaces, the subspace that \emph{changes} after each iteration;
\item The function $\Phi_{\ell_j} (B_j)$ studies $A_j-X$ projected onto the bottom $k$ eigenspace of $X$, the subspace that remains 
\emph{fixed} during the entire BSS process. 
\end{itemize}
Proving the existence of some vector in each iteration  so that the algorithm will make progress is much more involved, and constitutes one of the key lemmas used in    \cite{KMST10:subsparsification} for constructing a subgraph sparsifier. 
We remark that the subgraph sparsification algorithm presented in \cite{KMST10:subsparsification} requires the computation of the projection matrices $P_{L(A_j)}$ in each iteration.
Because of this, the algorithm presented in \cite{KMST10:subsparsification} runs in time $\Omega\left(n^2mk\right)$.

\paragraph*{Our approach.} At a very high level, our algorithm and its analysis can be viewed as a neat combination of the algorithm presented in \cite{LS15:linearsparsifier} and the algorithm presented in \cite{KMST10:subsparsification}. Specifically, for any iteration $j$  with the constructed matrix  $A_j$,  we set 
$
B_j \triangleq Z(A_j-X)Z,
$
where
$
Z \triangleq \left(\PV (I-X) \PV\right)^{\dag/2},
$
and define the two potential functions by 
\[
\Phi^{u_j} (A_j) \triangleq \tr\left( P_{L(A_j)} \left( u_j I - A_j \right)  P_{L(A_j)} \right)^{\dagger q},
\]
\[
\Phi_{\ell_j}(B_j) \triangleq \lPot{\ell_j}{B_j}
\]
for some fixed projection matrix $\PV$,  projecting on a $k$-dimensional subspace $S'$. Similar with \cite{LS15:linearsparsifier}, with the help of $q$-th power in the definition of $\Phi^{u_j} (A_j)$ and $\Phi_{\ell_j}(B_j)$ we   show that the eigenvalues of  our constructed matrices $A_j$ and $B_j$ are never very close to the two barrier values $u_j$ and $\ell_j$. Moreover, although the top $T$-eigenspace of the currently constructed matrix $A_j$ changes after every rank-one update,   multiple vectors can still be selected according to some probability distribution in each iteration. 

However,   when combining the randomised BSS framework~\cite{LS15:linearsparsifier} with the algorithm presented in \cite{KMST10:subsparsification}, we have to take many challenging technical issues into account. In particular, we need to address the following issues: 
\begin{itemize}
    \item Both the algorithm and its analysis in \cite{KMST10:subsparsification} crucially depend on the projection matrix $P_S$, of which the exact computation is expensive. Therefore, in order to obtain an efficient algorithm for subgraph sparsification, one needs to obtain some projection matrix close to $P_S$ and such projection matrix can be computed efficiently.
    \item As  indicated by  our definition of $\Phi_{\ell_j}(B_j)$ above, developing a fast subgraph sparsification algorithm would require efficient approximation of   polynomials of the matrix $(\PV(B_j - \ell_j I)\PV)^q$. In comparison to \cite{LS15:linearsparsifier},  the fixed projection matrix $\PV$ sandwiched between two consecutive $\left(B_j-\ell_j I\right)$ makes computing the required quantities much more challenging.
\end{itemize}

To address these issues, we   prove that there is a $k$-dimensional subspace $S'$ close to $S$, and all of our required quantities that involve the projection onto
$S'$, denoted by $\PV$, can be computed   efficiently.  
Moreover, we   prove that   the quality of  our constructed subgraph sparsifer  based on the ``approximate projection'' $\PV$ is the same as the one constructed by \cite{KMST10:subsparsification}, in which the ``optimal projection'' $P_S$ is needed.  Our result on the approximate subspace $S'$ is summarised as follows:

\begin{lemma}\label{thm:our projection informal}
There is an algorithm that computes a matrix $\mathbb{V} = L^{-1/2} V$ for
matrix $V$ in $ \min \Big\{O(n^{\omega}), \tilde{O}\lp\frac{mk + nk^2}{\sqrt{\lambda^*}}\rp \Big\}$ time, such that  with constant probability the following two properties hold:
\begin{enumerate}
    \item $\PV = VV^{\rot}$ is a projection matrix  on a $k$-dimensional subspace $S'$ of $\mathbb{R}^n$;
    \item For any $u\in\mathbb{R}^n$ satisfying $u^{\rot}V = 0$ we have that  
    \[
        \frac{u^{\rot}Xu}{u^{\rot}u} \geq \frac{\lambda_{k+1}(X)}{2} = \frac{\lambda^*}{2}.
    \]
\end{enumerate}
\end{lemma}

We highlight that, in comparison to \cite{LS15:linearsparsifier}, in our setting the upper and lower potential functions \emph{keep track of two different subspaces whose dimensions are of different orders in most regimes},  i.e., $k$ versus $T$, and this makes our analysis much more involved than \cite{LS15:linearsparsifier}.  
On the other side, we also show that the algorithm in \cite{LS15:linearsparsifier} can be viewed as a special case of our algorithms, and from this aspect our algorithm presents a general framework for constructing spectral sparsifiers and subgraph sparsifiers.

At the end of this subsection, we mention  the following fact about $\PV$, which will be extensively used in the remaining part of our analysis.

\begin{remark}
It is important to remember that $\PV$ is a fixed projection. Moreover, it holds by definition that $Z \cdot \PV = \PV \cdot Z = Z$, and $
  \sum_{i=1}^m Z v_i v_i^{\rot} Z = Z \barr{M} Z = \PV$.
\end{remark}


\subsection{Description of our algorithm\label{sec:ouralgo}}
Our algorithm proceeds in iterations in which multiple vectors are sampled with different probabilities. In each iteration $j$, $A_j$ is updated by setting $A_{j+1}=A_{j}+\Delta_j$, where $\Delta_j$ is the sum of the sampled rank-one matrices with reweighting. To compensate for this change, the two barriers $u_j$ and $\ell_j$ are increased by $\delta_{u,j}$ and $\delta_{\ell,j}$. The algorithm terminates when the difference of the barriers is greater than $\alpha$, defined by
$
    \alpha \triangleq  4k/ \Lambda.
$
Specifically, in the initialisation step, the algorithm sets
$
 A_0\triangleq X$,$u_0 \triangleq 2+\lambda_{\max}(X)$, $\ell_0 \triangleq -2k/\Lambda,
$
where 
$
\Lambda \triangleq \max\{k, T\}$.
In each iteration $j$, the algorithm keeps track of the currently constructed matrix $A_j$ and hence, also of the matrix $B_j \triangleq Z(A_j-X) Z$, where $Z \triangleq (\PV (I-X) \PV)^{\dagger/2}$ for some  fixed projection matrix $\PV$. Intuitively, the projection matrix $\PV$ used here is close to $P_S$, but can be approximated more efficiently than computing   $P_S$ precisely.  A detailed discussion regarding the exact computation of $\PV$ is presented in Section~\ref{sec:approximated projection}.
In each iteration $j$, the algorithm starts by computing the \emph{relative effective resistances}, which is defined as
\[
    R_i(A_j, B_j, u_j, \ell_j) \triangleq v_i^{\rot}\uRes{u_j}{A_j} v_i + v_i^{\rot}\lRes{\ell_j}{B_j}v_i,
\]
for all vectors $v_i$. Then, the algorithm computes the number of vectors $N_j$ that will be sampled, which can be written as 
\begin{align*} 
  N_j \triangleq  &\lp \frac{\varepsilon}{4\rho_j}\cdot \lambdaMin\lsp(u_j I-A_j)^{-1}\barr{M}\rsp \cdot \frac{\lambdaMax\lp(u_j I - A_j)^{-1}\barr{M}\rp}{\Tr \lsp(u_j I - A_j)^{-1}\barr{M}\rsp}\rp ^{2\varepsilon/q} \cdot \rho_j \\
  &\cdot \mathrm{min} \left\{ \frac{1}{\lambdaMax\lp (u_j I - A_j)^{-1}\barr{M}\rp}, \frac{1}{\lambdaMax\lp \PV (B_j - \ell_j I) \PV \rp^{\dag}} \right\},
\end{align*}
where 
\begin{align*}
    \rho_j &\triangleq \sum_{t=1}^m R_t(A_j, B_j, u_j, \ell_j) = \Tr \lsp (u_j I - A_j)^{-1}\barr{M}\rsp + \Tr \lsp \PV (B_j - \ell_j I ) \PV\rsp^{\dag}.
\end{align*}
Next, the algorithm samples $N_j$ vectors such that every $v_i$ is sampled with probability proportional to $R_i(A_j, B_j, u_j, \ell_j)$, i.e., the sampling probability of every $v_i$ is defined by
\[
p(v_i) \triangleq \frac{ R_i(A_j, B_j, u_j, \ell_j)}{ \sum_{t=1}^ m  R_t(A_j, B_j, u_j, \ell_j)}.
\]
For every sampled $v_i$, the algorithm scales it to 
\[
    w_i\triangleq \sqrt{\frac{\varepsilon}{ q\cdot R_i (A,B, u,\ell)}}\cdot v_i,
\]
and gradually adds $w_iw_i^{\rot}$ to $A_j$. After each rank-one update, the algorithm increases the barrier values by  the average increases
\[
    \overline{\delta}_{u, j} \triangleq   \frac{(1+3\varepsilon)\cdot \varepsilon}{q\cdot \rho_j}  \qquad \textrm{and} \qquad \overline{\delta}_{\ell, j} \triangleq \frac{(1-3\varepsilon)\cdot  \varepsilon}{q\cdot \rho_j},
\]
and checks whether the terminating condition of the algorithm is satisfied. Note that between two consecutive iterations $j$ and $j+1$ of the algorithm, the two barriers $u_j$ and $\ell_j$ are increased by 
$
   \delta_{u, j} \triangleq  N_j \cdot \overline{\delta}_{u, j}$  and $\delta_{\ell, j} \triangleq N_j \cdot \overline{\delta}_{\ell, j},
   $
respectively.

The formal description of our algorithm is presented in Algorithm~\ref{alg:subgraph sparsifier}. We remark that, in contrast to the algorithm for constructing a spectral sparsifier~\cite{LS15:linearsparsifier}, the total number of vectors  needed in the final iteration $j$ could be much smaller than $O(N_j)$.  This is why our algorithm performs a sanity check in line~15 after every   rank-$1$ update $w_iw_i^{\rot}$.

\begin{algorithm}[!h]
	\caption{Algorithm for constructing subgraph spectral sparsifiers}\label{alg:subgraph sparsifier}
	\begin{algorithmic}[1]
		\Require $\eps \leq  1/20, q \geq  10$ 
		
		\State $u_0 = 2+ \lambda_{\max}(X)$, $\ell_0 = -2k/\Lambda$
		\Comment{Here $u_0$ and $\ell_0$ are the initial barrier values}
		\State $\widehat{u} = u_0$ and $\widehat{\ell} = \ell_0$ 
		\Comment{Here $\widehat{u}$ and $\widehat{\ell}$ are the current barrier values} 
		
		\State $j=0$
		\Comment{ $j$ will be the index of the current iteration}
		\State $A_0=X$, $B_0=\mathbf{0}$
		\While{$\widehat{u} - \widehat{\ell} > \alpha + u_0 - \ell_0$}
		    \Comment{ Start of iteration $j$}
		    \State Compute $R_t(A_j, B_j, \ell_j, u_j)$ and hence $p(v_t)$ for all vectors $v_t$
		    \State Compute $N_j$
		    \State Sample $N_j$ vectors $v_1, \dots v_{N_j}$ according to $p$
		    \State Set $W_j \leftarrow 0$
		    \For{every subphase $i = 1 \dots N_j$}
		        \Comment{Start of subphase $i$}
		        \State $w_i \leftarrow \sqrt{\frac{\varepsilon}{q \cdot R_i(A_j, B_j, u_j, \ell_j)}} \cdot v_i$
		        \State $W_j \leftarrow W_j + w_iw_i^{\rot}$ 
    		    \State $\widehat{u} \leftarrow \widehat{u} +  \barr{\delta}_{u, j}$
    		    \State $\widehat{\ell} \leftarrow \widehat{\ell} + \barr{\delta}_{\ell, j}$
		        \If{$\widehat{u} - \widehat{\ell} > \alpha + u_0 - \ell_0$} 
		            \State Stop at the current subphase
		            \Comment{End of subphase $i$}
		        \EndIf
		    \EndFor
		        
	        \State $A_{j+1} \leftarrow A_j + W_j$
	        \State $B_{j+1} \leftarrow Z(A_{j+1} - X)Z$
	        \State $j = j+1$
	        \Comment{End of iteration $j$}
		\EndWhile
		\State \textbf{Return} $M = A_{j}$
	\end{algorithmic}
\end{algorithm}

%% file: Subgraph_Sparsification_Analysis.tex
\section{Analysis of the subgraph sparsification algorithm}\label{sec:analysis}

In this section we analyse Algorithm~\ref{alg:subgraph sparsifier} and prove Theorem~\ref{thm:main-formal}. The section will be organised as follows: in  subsection~\ref{sec:one iteration} we
analyse  the  sequence of upper and lower potential functions, and prove that  the  total effective resistances decrease after every iteration. In subsection~\ref{sec:total number of iterations and vectors}, we analyse the number of iterations needed before the algorithm terminates,   and the total number of vectors picked by the algorithm. In subsection~\ref{sec:runtime analysis} we present an efficient subroutine  to approximately compute all the quantities used in the algorithm.  In subsection~\ref{sec:approx guarantee},  we analyse the spectral properties of the final constructed matrix. Finally,  we combine everything together and prove Theorem~\ref{thm:main-formal} in subsection~\ref{sec:proof of main theorem}.


\subsection{Analysis within each iteration}\label{sec:one iteration}
We analyse the sampling scheme within a single iteration\footnote{Throughout the rest of the paper we will interchangeably use the words ``iteration'' and ``phase'' when referring to one run of the while loop (lines 6-19) in Algorithm~\ref{alg:subgraph sparsifier}.}, and drop the subscript representing the iteration $j$ for simplicity. We assume that in each iteration the algorithm samples $N$ vectors independently from $\{v_i\}_{i=1}^m $, where every vector is sampled with probability proportional to their relative effective resistance. We use $v_1,\ldots, v_N$ to denote these $N$ sampled vectors, and let 
\[
w_i\triangleq \sqrt{\frac{\varepsilon}{ q\cdot R_i (A,B, u,\ell)}}\cdot v_i
\]
for any $1\leq i\leq N$.
Let 
\[
W \triangleq \sum_{i=1}^N w_i w_i^{\rot}.
\]

\begin{lemma}\label{lem:expectation w_i}
It holds for any $1\leq i\leq N$ that 
\[
    \Ex \lsp w_i{w_i}^{\rot} \rsp = \frac{\varepsilon}{q\cdot \rho} \cdot \barr{M},
\]
and hence 
\[
    \Ex \lsp W \rsp = \frac{\varepsilon\cdot N}{q\cdot \rho}   \cdot \barr{M}.
\]
\end{lemma}
\begin{proof}
We have for any $1\leq i \leq N$ that 
\[
    \Ex \left[ w_i w_i^{\rot} \right]  =   \sum_{i=1}^m \frac{R_i(A, B, u, \ell)}{\rho } \cdot \frac{\varepsilon}{q \cdot R_i(A, B, u, \ell)} \cdot v_i{v_i}^{\rot} = \frac{\varepsilon}{q \cdot \rho } \cdot \barr{M},
\]
and hence
\[
    \Ex \lsp W\rsp  = \frac{\varepsilon \cdot N}{q \cdot \rho } \cdot \barr{M}.
\]
\end{proof} 
Next, we will show that, as long as the total number of sampled vectors are not too large, it holds with probability at least $1-\varepsilon/(2n)$ that $\mathbf{0} \preceq W\preceq (1/2) (uI-A)$. To prove this, we recall the following stronger version of the matrix Chernoff bound.

\begin{lemma}[\cite{Tropp15}]\label{lem:intdim MCB}
Let $\{X_i\}$ be a finite sequence of independent, Hermitian matrices of the same size, such that $\lambda_{\min}(X_i)\geq 0$ and $\lambda_{\max}(X_i)\leq D$ hold for every matrix $X_i$.
Moreover, assume that  $\mathbf{E} \lsp \sum_i X_i\rsp \preceq U,$
and let  
$  \mu \triangleq \lambda_{\mathrm{max}}(U).$
Then, it holds for any $\delta$ with $\delta \geq D/\mu$  that 
\[
    \Pro \lsp \lambda_{\mathrm{max}}\lp \sum_i X_i \rp \geq (1 + \delta)\mu \rsp \leq 2\cdot \mathrm{intdim}(U) \cdot \lp \frac{\mathrm{e}^\delta}{(1+\delta)^{1+\delta}}\rp^{\mu/D}.
\]
\end{lemma}
\begin{lemma} \label{lem:bounded W}
Assume that the number of samples satisfies
\[ 
    N  \leq \lp \frac{\varepsilon}{4\rho}\cdot \lambdaMin\lsp(u I - A)^{-1}\barr{M}\rsp \cdot \frac{\lambdaMax\lp (u I - A)^{-1}\barr{M}\rp }{\Tr \lsp (u I - A)^{-1}\barr{M}\rsp}\rp^{2\varepsilon/q} \cdot \rho \cdot \frac{1}{\lambdaMax\lp (u I - A)^{-1}\barr{M}\rp} 
\]
Then it holds that 
\[
    \Pro \left[ \mathbf{0} \preceq W \preceq \frac{1}{2}(u I - A)\right] \geq 1-  \frac{\varepsilon}{2 \rho} \cdot \lambdaMin \lsp (u I - A)^{-1}\barr{M}\rsp \geq 1 - \frac{\varepsilon}{2 n}.
\]
\end{lemma}
\begin{proof}
  Let $z_i = (u I - A)^{- 1/2} w_i$. We will prove this result using Lemma~\ref{lem:intdim MCB} for the sequence of matrices $z_i z_i^{\rot}$. We have that
  \begin{align*}
    \Tr(z_i z_i^{\rot}) &= z_i^{\rot}z_i \\
                        &= {w_i}^{\rot}(uI - A)^{-1}w_i \\
                        &= \frac{\varepsilon}{q} \cdot \frac{{v_i}^{\rot}(uI - A)^{-1}v_i}{R_i(A, B, u, \ell)} \\
                        &= \frac{\varepsilon}{q} \cdot \frac{{v_i}^{\rot}\uRes{u}{A}v_i}{{v_i}^{\rot}\uRes{u}{A}v_i + {v_i}^{\rot}\lRes{\ell}{B}v_i} \\
                        &\leq \frac{\varepsilon}{q},
  \end{align*}
  which implies that  $\lambdaMax(z_i z_i^{\rot}) \leq  \varepsilon/q$. Therefore,  we have that
  \begin{align*}
    \Ex\left(\sum_{i=1}^N z_i z_i^{\rot}\right)  
    &= \Ex \lsp(u I - A)^{-1/2} \lp\sum_{i=1}^N w_i {w_i}^{\rot}\rp (u I - A)^{-1/2}\rsp \\
    &= \frac{\varepsilon \cdot N}{q \cdot \rho} \cdot  (uI - A)^{-1/2} \barr{M} (uI - A)^{-1/2},
  \end{align*}
  which implies that  
  \begin{align*}
    \lambdaMax\lp \Ex \lsp \sum_{i=1}^N z_i z_i^{\rot}\rsp \rp 
    &= \frac{\varepsilon \cdot N}{q \cdot \rho} \cdot \lambdaMax \lp (uI - A)^{-1/2} \barr{M} (uI - A)^{-1/2}\rp  \\
    &= \frac{\varepsilon \cdot N}{q \cdot \rho} \cdot \lambdaMax \lp (uI - A)^{-1}\barr{M}\rp,
  \end{align*}
  where that last equality follows by the fact that the eigenvalues are invariant under circular permutations.
  
  Now we apply the Matrix Chernoff Bound~(Lemma~\ref{lem:intdim MCB}) to analyse $W$. To this end, we define $D\triangleq \varepsilon/q$,
  \[
  U \triangleq \Ex\left(\sum_{i=1}^N z_i z_i^{\rot}\right) =\frac{\varepsilon \cdot N}{q \cdot \rho} \cdot  (uI - A)^{-1/2} \barr{M} (uI - A)^{-1/2},
  \]
 \[  \mu = \lambda_{\max}(U)= \frac{\varepsilon \cdot N}{q \cdot \rho} \cdot \lambdaMax \lp (uI - A)^{-1} \barr{M}\rp , \]
  and  \[   
     \mathrm{intdim}(U) = \frac{\Tr(U)}{\norm{U}} 
      = \frac{\Tr \lp (uI - A)^{-1/2} \barr{M} (uI - A)^{-1/2}\rp}{\norm{(uI - A)^{-1/2} \barr{M} (uI - A)^{-1/2}}}
      = \frac{\Tr \lp (uI - A)^{-1} \barr{M}\rp}{\lambdaMax\lp (uI - A)^{-1} \barr{M} \rp,
    }
  \]
  using the fact that both the trace and the eigenvalues are invariant under circular permutations. We also set $\delta$ such that 
  \begin{equation*}
    1+\delta = \frac{1}{2\mu} =\frac{q \cdot \rho}{2\varepsilon \cdot N} \cdot \frac{1}{\lambdaMax\lp (uI - A)^{-1} \barr{M}\rp}.
  \end{equation*}
  Together with our choice of $\varepsilon$ and $q$, it is  easy to check that this value of  $\delta$ satisfies  $\delta\geq D/\mu$, which is required for Lemma~\ref{lem:intdim MCB}.  Hence, by applying   the Matrix Chernoff Bound (Lemma~\ref{lem:intdim MCB}) we have that
  \begin{align}
    &\mathbf{P}\left[\lambdaMax\left(\sum_{i=1}^N z_i z_i^{\rot}\right) \geq(1+\delta)\mu\right]  \leq   2\cdot \mathrm{intdim}(U) \cdot \left(\frac{\mathrm{e}^{\delta}}{(1+\delta)^{1+\delta}}\right)^{\mu q/\varepsilon} \nonumber\\
     &\leq   \frac{2\cdot \Tr \lp (uI - A)^{-1} \barr{M}\rp}{\lambdaMax\lp (uI - A)^{-1} \barr{M} \rp} \cdot \lp \frac{\mathrm{e}}{1+\delta} \rp^{(1+\delta)\mu q /\varepsilon} \nonumber\\
     &= 2 \cdot \frac{\Tr \lp (uI - A)^{-1} \barr{M}\rp}{\lambdaMax\lp (uI - A)^{-1} \barr{M} \rp} \cdot \lp \frac{2 \mathrm{e}  \varepsilon N}{q\rho}   \cdot \lambdaMax\lp (uI - A)^{-1} \barr{M} \rp \rp^{q/2\varepsilon} \nonumber\\
     &\leq 2\cdot \frac{\varepsilon}{4 \rho} \cdot \lambdaMin \lsp (u I - A)^{-1}\barr{M}\rsp \label{eq:choice of N}\\
     &\leq \frac{\varepsilon}{2} \cdot \frac{\lambdaMin \lsp(u I - A)^{-1}\barr{M}\rsp}{\Tr \lsp (u I - A)^{-1}\barr{M}\rsp} \nonumber\\
     &\leq \frac{\varepsilon}{2n}, 
  \end{align}
  where \eqref{eq:choice of N} follows from the choice of $N$ and the fact that $\left (2\varepsilon \mathrm{e}/q \right)^{q/2\varepsilon} \leq 1$.
  Hence, with probability at least $1 - \varepsilon/(2n)$  it holds that 
  \[
    \lambdaMax\lp \sum_{i=1}^N z_i z_i^{\rot}\rp \leq (1+\delta)\mu = \frac{1}{2},
  \]
  which  implies that $\mathbf{0} \preceq \sum_{i=1}^N z_i z_i^{\rot} \preceq (1/2)\cdot I$ and thus $\mathbf{0} \preceq W \preceq (1/2)\cdot (uI - A)$.
\end{proof}

By  Lemma~\ref{lem:bounded W} we know that the sampled matrix $W$ satisfies $W \preceq \frac{1}{2} \lp u I - A\rp $ with probability $
1-  \varepsilon/(2n) $. In the following, conditioning on this event we will show that the expected value of the potential function decreases. To this end, we introduce the conditional expectation $\widetilde{\Ex}$ defined by 
\[
\widetilde{\Ex}[ \cdot ] = \Ex \lsp \ \cdot \  \Big| W \preceq \frac{1}{2} \lp u I - A\rp \rsp. 
\]
The following estimate tells us that $w_iw_i^{\rot}$ does not change too much under the new expectation.
\begin{lemma}\label{lem:bound expectation} It holds for any vector $w$ defined above that
  \[
    (1-\varepsilon/2) \cdot \Ex \lsp ww^{\rot} \rsp \preceq \widetilde{\Ex}\lsp ww^{\rot}\rsp \preceq (1+\varepsilon/2) \cdot \Ex \lsp ww^{\rot}\rsp.
  \]
\end{lemma}
\begin{proof}
    For the second inequality, we have that
    \begin{align*}
      \widetilde{\Ex}\lsp ww^{\rot}\rsp 
      &= \Ex \lsp ww^{\rot} \big| W \preceq \frac{1}{2}(uI-A) \rsp 
      \preceq \frac{\Ex \lsp ww^{\rot} \rsp}{\Pro \lsp W \preceq \frac{1}{2}(uI-A)\rsp}\\
      &\preceq  \frac{\Ex \lsp ww^{\rot} \rsp}{1- \frac{\varepsilon}{2n}}
      \preceq (1+\varepsilon/2)\cdot \Ex \lsp ww^{\rot} \rsp,
    \end{align*}
    where the second last inequality comes from Lemma~\ref{lem:bounded W}.
    
    To prove  the first inequality, notice  that
    \begin{align*}
      \widetilde{\Ex} \lsp ww^{\rot}\rsp &= \Ex \lsp ww^{\rot}\big|W\preceq \frac{1}{2}(uI-A)\rsp \\
      &\succeq \Ex \lsp ww^{\rot}\rsp - \Pro\lsp W \npreceq \frac{1}{2}(uI-A)\rsp \cdot \Ex \lsp ww^{\rot} \big| W \npreceq \frac{1}{2}(uI-A)\rsp  \\
      &\succeq \Ex \lsp ww^{\rot}\rsp - \frac{\varepsilon}{2\rho}\cdot \lambdaMin\lsp (uI-A)^{-1}\barr{M}\rsp \cdot \frac{\varepsilon}{q} \cdot (uI-A).
    \end{align*}
    Therefore, to prove the statement it suffices to show that
    \[
      \frac{\varepsilon}{2\rho}\cdot \lambdaMin\lsp (uI-A)^{-1}\barr{M}\rsp \cdot \frac{\varepsilon}{q} \cdot (uI-A) \preceq (\varepsilon/2) \cdot \Ex\lsp ww^{\rot}\rsp.
    \]
    Based on   Lemma~\ref{lem:expectation w_i}, this  is equivalent to show that 
    \[
      \frac{\varepsilon}{2\rho}\cdot \lambdaMin\lsp (uI-A)^{-1}\barr{M}\rsp \cdot \frac{\varepsilon}{q} \cdot (uI-A) \preceq (\varepsilon/2) \cdot \frac{\varepsilon}{q\cdot \rho} \cdot \barr{M},    
    \]
    or alternatively
    \[
      \lambdaMin\lsp (uI-A)^{-1}\barr{M}\rsp \cdot I \preceq (uI-A)^{-1/2}\barr{M}(uI-A)^{-1/2},
    \]
    which holds by the fact that the eigenvalues are invariant under circular permutations. 
\end{proof}

To analyse the spectral properties of the sampled matrix,  recall that our potential functions are defined by  
\[
\Phi^{u}(A) = \tr\left( P_{L(A)} \left( u I - A \right)  P_{L(A)} \right)^{\dagger q} = \sum_{i=n-T+1}^n \lp\frac{1}{u-\lambda_i(A)} \rp^q,
\]
\[
\Phi_{\ell}(B) = \lPot{\ell}{B} = \sum_{i=1}^k \lp\frac{1}{\lambda_i(B\big|_{S'}) - \ell}\rp^q,
\] 
where $S'$ is the $k$-dimensional space where $\PV$ is projecting onto. We further write \[
\Phi_{u, \ell}(A, B) \triangleq \Phi^{u}(A) + \Phi_{\ell}(B).
\]
Note that the total numbers of the terms involved in $\Phi^{u}(A)$ and $\Phi_{\ell}(B)$ respectively are  different for most settings, and this is another reason why our analysis is more involved. 

To analyse the change of the potential functions and relative effective resistances, we divide each iteration into  subphases, and analyse the change of the potential function after each rank-one update.  Without loss of generality, we assume that $v_1,\ldots, v_N$ are the $N$ sampled vectors in this iteration.  
We introduce matrices $A^{(i)}$ and $B^{(i)}$ defined by 
\[ A^{(i)} = A + \sum_{t=1}^{i} w_t w_t^{\rot} \quad \text{and} \quad B^{(i)} = B + \sum_{t=1}^{i} Zw_t w_t^{\rot}Z,
\]
 for every $0 \leq i \leq N$. 
Recall that we defined the average change of the barrier values by 
\[
\barr{\delta_u} = \frac{\delta_{u}}{N} = \frac{(1+3\varepsilon)\cdot \varepsilon}{q \cdot \rho} \quad \text{and} \quad \barr{\delta_\ell} = \frac{\delta_{\ell}}{N} = \frac{(1-3\varepsilon)\cdot \varepsilon}{q \cdot \rho}, 
\]
and let 
\[
\widehat{u}_i \triangleq u + i \cdot \barr{\delta_u} \quad \text{and} \quad \widehat{\ell}_i \triangleq \ell + i \cdot \barr{\delta_\ell}.
\]
For each intermediate subphase we define
\[
  \rho^{\widehat{u}_i}\left(A^{(i)}\right) \triangleq \Tr \lsp \left(\widehat{u}_i I - A^{(i)}\right)^{-1} \barr{M}\rsp
\]
and
\[
  \rho_{\widehat{\ell}_i}\left(B^{(i)}\right) \triangleq \Tr \lsp \PV \left(B^{(i)} - \widehat{\ell}_i I\right) \PV \rsp^{\dag}.
\]
We also define that 
\[
  \rho^{(i)} \triangleq \rho^{\widehat{u}_i}\left(A^{(i)}\right) +  \rho_{\widehat{\ell}_i}\left(B^{(i)}\right).
\]

The following lemma relates the properties of the sampled matrix $W$ to the individual rank-one update within each iteration.
 \begin{lemma}\label{lem:first bound w_i}
 If   $W \preceq (1/2)\cdot (u I - A)$, then it holds for any $0\leq i\leq N-1$  that  
    \[
    w_{i+1}^{\rot}\left(\widehat{u}_{i+1}I - A^{(i)} \right)^{-1}w_{i+1} \leq \frac{2\varepsilon}{q}
    \]
     and
    \[
    w_{i+1}^{\rot}\lRes{\widehat{\ell}_{i+1}}{B^{(i)}}w_{i+1} \leq \frac{2\varepsilon}{q}.
    \]
  \end{lemma}
\begin{proof}
Since $W \preceq \frac{1}{2} (u I - A)$, it holds for any $0\leq i \leq N-1$ that
\[
A^{(i)} - A \preceq W \preceq \frac{1}{2} (uI-A),
\]
which implies that
\[
2A^{(i)} - A \preceq uI \preceq \left( 2\widehat{u}_{i+1} - u \right) I,
\]
i.e.,
\begin{equation}\label{eq:rewritecondition}
uI-A\preceq 2\left( \widehat{u}_{i+1} I- A^{(i)}\right).
\end{equation}
To prove the first statement, we notice that 
\begin{align*}
w_{i+1}^{\rot}\left(\widehat{u}_{i+1}I - A^{(i)} \right)^{-1}w_{i+1} 
      &= \frac{\varepsilon}{q} \cdot \frac{v_{i+1}^{\rot} 
      \uRes{\hat{u}_{i+1}}{A ^{(i)}}v_{i+1}}{R_{i+1}(A, B, u, \ell)}\\ 
      &\leq \frac{\varepsilon}{q} \cdot \frac{v_{i+1}^{\rot}\uRes{\hat{u}_{i+1}}{A^{(i)}}v_{i+1}}{v_{i+1}^{\rot}\uRes{u}{A}v_{i+1}}.
\end{align*}

Hence, it suffices to show that
\[
 \uRes{\hat{u}_{i+1}}{A^{(i)}} \preceq 2 \uRes{u}{A},
\]
which  holds by (\ref{eq:rewritecondition}).  

Now we prove the second statement. Since $B-\ell I$ is positive definite and $\PV Z=Z$, it holds that $v_{i+1}^{\rot}\lRes{\ell}{B}v_{i+1} = 0$ if $v_{i+1}$ is not in $S'$, in which case the statement holds trivially.  Therefore, we only need to study the case that $v_{i+1}^{\rot}\lRes{\ell}{B}v_{i+1} >0$, under which condition it holds that 
\begin{align*} 
      w_{i+1}^{\rot}\lRes{\hat{\ell}_{i+1}}{B^{(i)}}w_{i+1} 
      &= \frac{\varepsilon}{q} \cdot \frac{v_{i+1}^{\rot}\lRes{\widehat{\ell}_{i+1}}{B^{(i)}}v_{i+1}}{R_{i+1}(A, B, u, \ell)} \\
      &\leq \frac{\varepsilon}{q} \cdot \frac{v_{i+1}^{\rot}\lRes{\widehat{\ell}_{i+1}}{B^{(i)}}v_{i+1}}{v_{i+1}^{\rot}\lRes{\ell}{B}v_{i+1}}.
    \end{align*}
    To prove the statement, it suffices to show that   
    \[
     \lp \PV \lp B^{(i)} - \widehat{\ell}_{i+1}I \rp \PV \rp^{\dag} \preceq 2 \lp \PV \lp B - \ell I \rp \PV \rp^{\dag},
     \]
    which is equivalent to showing that 
\begin{equation}\label{eq:cond2}
\PV \lp B^{(i)} - \widehat{\ell}_{i+1}I \rp \PV \succeq \frac{1}{2} \cdot \PV \lp B - \ell I \rp \PV.
\end{equation}
By the definition of $\delta_{\ell}$, we have that 
\begin{align*}
\delta_{\ell} &= \frac{(1- 3\varepsilon) \cdot \varepsilon\cdot N}{q\cdot \rho}  \\
&\leq \frac{(1 -3\varepsilon)\cdot \varepsilon}{q} \cdot \mathrm{min} \left\{ \frac{1}{\lambdaMax\lp (u I - A)^{-1}\barr{M}\rp}, \frac{1}{\lambdaMax\lp \PV (B - \ell I) \PV\rp^{\dag}} \right\} \\
&\leq \frac{1}{2} \cdot \frac{1}{\lambdaMax\lp \PV (B - \ell I) \PV\rp^{\dag}},
\end{align*}
where the last inequality follows by our choice of $\varepsilon$ and $q$.
Hence we have that
\[
	\delta_\ell \PV \preceq \frac{1}{2} \cdot \PV \lp B - \ell I \rp \PV.
\]
Therefore, it holds that
  \begin{align*} 
      \PV \lp B^{(i)} - \widehat{\ell}_{i+1}I \rp \PV 
      &\succeq \PV \lp B - \widehat{\ell}_{i+1}I \rp \PV \\
      &= \PV \lp B - \ell I \rp \PV - (i+1)\cdot \barr{\delta_\ell} \PV\\
      &\succeq \PV \lp B - \ell I \rp \PV - \delta_{\ell} \PV\\
      & \succeq \frac{1}{2} \cdot \PV \lp B - \ell I \rp \PV,
    \end{align*}
which proves (\ref{eq:cond2}) and  the second statement of the lemma.
\end{proof}

\begin{lemma}\label{lem:bound separate potentials}
    Assuming $q \geq 10, \varepsilon \leq  1/20$, 
    \[w_{i+1}^{\rot}\uRes{\hat{u}_{i+1}}{A^{(i)}}w_{i+1} \leq \frac{2\varepsilon}{q}\] and 
    \[w_{i+1}^{\rot}\lRes{\hat{\ell}_{i+1}}{B^{(i)}}w_{i+1} \leq \frac{2\varepsilon}{q},\]
    it holds that 
    \[
    \Phi^{\hat{u}_{i+1}}\left(A^{(i+1)}\right) \leq \Phi^{\hat{u}_{i+1}}\left(A^{(i)}\right) + q (1+2\varepsilon) \cdot w_{i+1}^{\rot} \lp \hat{u}_{i+1} I - A^{(i)}\rp^{-(q+1)}w_{i+1}
    \]
    and
    \[
    \Phi_{\hat{\ell}_{i+1}}\left(B^{(i+1)}\right) \leq \Phi_{\hat{\ell}_{i+1}}\left(B^{(i)}\right) - q(1-2\varepsilon) \cdot w_{i+1}^{\rot} Z \lp \PV \lp B^{(i)} - \hat{\ell}_{i+1} I \rp \PV \rp^{\dag (q+1)}Z w_{i+1}.
    \]
  \end{lemma}
\begin{proof}
  To prove the first statement, for simplicity we define 
  $Y = \hat{u}_{i+1} I - A^{(i)}$, $P = \PL{A^{(i+1)}}$, and  $w = w_{i+1}$. Then it holds that  
    \begin{align}
      \Phi^{\hat{u}_{i+1}}\left(A^{(i+1)}\right)
      &= \Tr \lsp \PL{A^{(i+1)}} \lp \hat{u}_{i+1} I - A^{(i)} - w_{i+1}w_{i+1}^{\rot}\rp^{-q} \PL{A^{(i+1)}} \rsp \nonumber\\ 
      &= \Tr \lsp P \lp Y - ww^{\rot} \rp^{-q} P \rsp \nonumber\\
      &= \Tr \lsp P \lp Y^{-1} + \frac{Y^{-1}ww^{\rot}Y^{-1}}{1 - w^{\rot}Y^{-1}w}\rp^q P\rsp,  \nonumber
        \end{align}
 where the last line follows  by  the Sherman-Morrison formula~(Lemma~\ref{lem:woodbury}). Since $P$ is the projection onto $L\left(A^{(i+1)} \right)$, the bottom $T$ eigenspace of $(Y-ww^{\rot})^{-q}$, by applying the argument from the spectral decomposition of $(Y-ww^{\rot})^{-q}$ we have that 
       \begin{align}
  \Phi^{\hat{u}_{i+1}}\left(A^{(i+1)}\right)     &= \Tr \lsp \lp P \lp Y^{-1} + \frac{Y^{-1}ww^{\rot}Y^{-1}}{1 - w^{\rot}Y^{-1}w}\rp P \rp^q \rsp \nonumber \\
            &= \Tr \lsp \lp PY^{-1/2} \lp I + \frac{Y^{-1/2}ww^{\rot}Y^{-1/2}}{1 - w^{\rot}Y^{-1}w} \rp Y^{-1/2}P \rp^q \rsp \nonumber\\
      &\leq \Tr \lsp \lp PY^{-1/2}\rp^q \lp I + \frac{Y^{-1/2}ww^{\rot}Y^{-1/2}}{1 - w^{\rot}Y^{-1}w} \rp^q \lp Y^{-1/2}P \rp^q \rsp \nonumber\\
      &= \Tr \lsp \lp Y^{-1/2}P \rp^q \lp PY^{-1/2}\rp^q \lp I + \frac{Y^{-1/2}ww^{\rot}Y^{-1/2}}{1 - w^{\rot}Y^{-1}w} \rp^q \rsp, \label{eq:step1}
    \end{align}
    where the last inequality follows by
    the Araki-Lieb-Thirring inequality~(Lemma~\ref{lem:ALT_ineq}).   
    Now we   use the Taylor expansion of matrices to upper bound the 
      second matrix above. Let  $D = Y^{-1/2}ww^{\rot}Y^{-1/2}$. Since $D \preceq  (2\varepsilon)/q \cdot I, q \geq 10$ and $\varepsilon \leq 1/20$,  we have that
    \begin{align}
     &\lp I + \frac{Y^{-1/2}ww^{\rot}Y^{-1/2}}{1 - w^{\rot}Y^{-1}w} \rp^q  
     \preceq 
      \lp I + \frac{D}{1-2\varepsilon/q} \rp^q \notag  \\
      &\preceq I + \frac{qD}{1-2\varepsilon/q} + \frac{q(q-1)}{2} \cdot \lp 1 + \frac{2\varepsilon/q}{1-2\varepsilon/q}\rp^{q-2} \lp \frac{D}{1-2\varepsilon/q}\rp^2 \notag \\
      &\preceq I + q \lp 1 + 1.1 \frac{2\varepsilon}{q}\rp D + 1.4 \frac{q(q-1)}{2} D^2 \notag \\
      &\preceq I + q(1+2\varepsilon)D. \label{eq:step2}
    \end{align}
   Combining \eqref{eq:step1} with \eqref{eq:step2}, we have that 
    \begin{align}
    &\Phi^{\hat{u}_{i+1}}\left(A^{(i+1)}\right)     \leq \Tr \lsp \lp Y^{-1/2}P \rp^q \lp PY^{-1/2}\rp^q \lp I + \frac{Y^{-1/2}ww^{\rot}Y^{-1/2}}{1 - w^{\rot}Y^{-1}w} \rp^q \rsp \nonumber\\
      &\leq \Tr \lsp \lp Y^{-1/2}P \rp^q \lp PY^{-1/2}\rp^q \lp I + q(1+2\varepsilon)D \rp \rsp \notag\\
      &=    \Tr \lsp \lp Y^{-1/2}P \rp^q \lp PY^{-1/2}\rp^q \rsp + q(1+2\varepsilon) \Tr \lsp \lp Y^{-1/2}P \rp^q \lp PY^{-1/2}\rp^q D\rsp \nonumber\\
      &\leq \Tr \lsp  PY^{-q}P \rsp + q(1+2\varepsilon) \Tr \lsp \lp Y^{-1/2}P \rp^q \lp PY^{-1/2}\rp^q Y^{-1/2}ww^{\rot}Y^{-1/2}\rsp \label{eq:newlabel}\\
      &= \Tr \lsp P Y^{-q} P\rsp + q(1+2\varepsilon) w^{\rot}Y^{-1/2} \lp Y^{-1/2}P \rp^q \lp PY^{-1/2}\rp^q Y^{-1/2}w \nonumber\\
      &\leq \Tr \lsp \PL{A^{(i+1)}} \lp \hat{u}_{i+1} I - A^{(i)} \rp^{-q} \PL{A^{(i+1)}} \rsp 
      + q(1+2\varepsilon) w^{\rot}Y^{-1}Y^{-(q-1)} Y^{-1}w  \notag\\
      &\leq \Tr \lsp \PL{A^{(i)}} \lp \hat{u}_{i+1} I - A^{(i)} \rp^{-q} \PL{A^{(i)}} \rsp 
      + q(1+2\varepsilon) w^{\rot}Y^{-(q+1)}w \label{eq:proj props}\\
      &= \Phi^{\hat{u}_{i+1}}\left(A^{(i)}\right) + q (1+2\varepsilon) \cdot w_{i+1}^{\rot} \lp \hat{u}_{i+1} I - A^{(i)}\rp^{-(q+1)}w_{i+1}, \nonumber
    \end{align}
    where \eqref{eq:newlabel} follows by Lemma~\ref{lem:proj perm}, and 
  \eqref{eq:proj props} follows by Lemma~\ref{cor:trace ineq proj}.
    
Next we prove the second statement of the lemma. For simplicity we assume that   $\barr{Y} = \PV \lp B^{(i)} - \hat{\ell}_{i+1} I \rp \PV$ and $w = Zw_{i+1}$. Then it holds that 
    \begin{align*}
      \Phi_{\hat{\ell}_{i+1}}\left(B^{(i+1)}\right)
      &= \Tr \lsp \PV \lp B^{(i)} + Zw_{i+1}w_{i+1}^{\rot}Z - \hat{\ell}_{i+1} I \rp \PV\rsp^{\dag q} \\
      &= \Tr \lsp \barr{Y} + ww^{\rot}\rsp^{\dag q} \\
      &= \Tr \lsp \barr{Y}^{\dag} - \frac{\barr{Y}^{\dag}ww^{\rot}\barr{Y}^{\dag}}{1 + w^{\rot}\barr{Y}^{\dag}w}\rsp^q \\
             &= \Tr \lsp \barr{Y}^{\dag 1/2} \lp I - \frac{\barr{Y}^{\dag 1/2}ww^{\rot}\barr{Y}^{\dag 1/2}}{1   + w^{\rot}\barr{Y}^{\dag}w} \rp \barr{Y}^{\dag 1/2}\rsp^q \\
              &\leq \Tr \lsp \barr{Y}^{\dag q/2} \lp I - \frac{\barr{Y}^{\dag 1/2}ww^{\rot}\barr{Y}^{\dag 1/2}}{1   + w^{\rot}\barr{Y}^{\dag}w} \rp^q \barr{Y}^{\dag q/2}\rsp \\
      &\leq \Tr \lsp \barr{Y}^{\dag q/2} \lp I - \frac{\barr{Y}^{\dag 1/2}ww^{\rot}\barr{Y}^{\dag 1/2}}{1 + 2\varepsilon/q} \rp^q \barr{Y}^{\dag q/2}\rsp,
    \end{align*}
    where the second equality follows by the fact that $Z\cdot \PV = \PV\cdot Z = Z$, the third equality follows from   Lemma~\ref{lem:pseud-inv SM}, and the last inequality follows by the condition that $w^{\rot}\barr{Y}^{\dag} w\leq 2\varepsilon/q $. 
    We define 
     \[
     E = \frac{\barr{Y}^{\dag 1/2}ww^{\rot}\barr{Y}^{\dag 1/2}}{1 + 2\varepsilon/q}.
     \]From the assumption of the lemma, we know that $E \preceq \frac{2\varepsilon}{q} I$ and hence
      \[
      (I - E)^q \preceq I - qE + \frac{q(q-1)}{2} E^2 \preceq I - \lp q -  \varepsilon(q-1) \rp E.
      \]
      Therefore, it holds that  
      \begin{align*}
         &\Phi_{\hat{\ell}_{i+1}}\left(B^{(i+1)}\right)  \leq \Tr \lsp \barr{Y}^{\dag q/2} \lp I - \frac{\barr{Y}^{\dag 1/2}ww^{\rot}\barr{Y}^{\dag 1/2}}{1 + 2\varepsilon/q} \rp^q \barr{Y}^{\dag q/2}\rsp \\
        &\leq \Tr \lsp \barr{Y}^{\dag q/2} \lp I - \lp q -  \varepsilon(q-1) \rp E\rp \barr{Y}^{\dag q/2}\rsp \\
        &=\Tr \lsp \barr{Y}^{\dag q}\rsp - \lp q - \varepsilon (q-1) \rp \Tr \lsp \barr{Y}^{\dag q/2} E \barr{Y}^{\dag q/2}\rsp \\
        &=\Tr \lsp \barr{Y}^{\dag q}\rsp - \frac{q - \varepsilon (q-1)}{1 + 2\varepsilon/q} \Tr \lsp \barr{Y}^{\dag (q+1)/2}ww^{\rot}\barr{Y}^{\dag (q+1)/2}\rsp \\
        &\leq \Tr \lsp \PV \lp B^{(i)} - \hat{\ell}_{i+1} I \rp \PV \rsp^{\dag q} - q(1-2\varepsilon)w^{\rot} \lp \PV \lp B^{(i)} - \hat{\ell}_{i+1} I \rp \PV \rp^{\dag(q+1)} w\\
        &=\Phi_{\hat{\ell}_{i+1}}\left(B^{(i)}\right) - q(1-2\varepsilon) \cdot w_{i+1}^{\rot} Z \lp \PV \lp B^{(i)} - \hat{\ell}_{i+1} I \rp \PV \rp^{\dag (q+1)}Z w_{i+1},
      \end{align*}
      which proves the second statement.
\end{proof}

\begin{lemma}\label{lem:bound separate sum resistances}
Assuming $q \geq 10, \varepsilon \leq  1/20$ and
\[ w_{i+1}^{\rot}\uRes{\widehat{u}_{i+1}}{A^{(i)}}w_{i+1} \leq \frac{2\varepsilon}{q} \] and 
\[ w_{i+1}^{\rot}\lRes{\widehat{\ell}_{i+1}}{B^{(i)}}w_{i+1} \leq \frac{2\varepsilon}{q},\]
it holds that 
\[
  \rho^{\widehat{u}_{i+1}}\left(A^{(i+1)}\right) \leq \rho^{\widehat{u}_{i+1}}\left(A^{(i)}\right) + \frac{1}{1-2\varepsilon/q} \cdot w_{i+1}^{\rot} \lp \widehat{u}_{i+1} I - A^{(i)}\rp^{-1} \barr{M} \lp \widehat{u}_{i+1} I - A^{(i)}\rp^{-1} w_{i+1}
\]
and
\[
  \rho_{\widehat{\ell}_{i+1}}\left(B^{(i+1)}\right) \leq \rho_{\widehat{\ell}_{i+1}}\left(B^{(i)}\right) - \frac{1}{1+2\varepsilon/q} \cdot w_{i+1}^{\rot}Z \lp \PV \lp B^{(i)} - \widehat{\ell}_{i+1} I \rp \PV \rp^{\dag 2}Zw_{i+1}.
\]
\end{lemma}

\begin{proof}
    To prove the first statement, for simplicity we define 
  $Y = \widehat{u}_{i+1} I - A^{(i)}$ and  $w = w_{i+1}$, for brevity. Then it holds that  
    \begin{align}
      &\rho^{\widehat{u}_{i+1}}\left(A^{(i+1)}\right)
      = \Tr \lsp \lp \widehat{u}_{i+1} I - A^{(i)} - w_{i+1}w_{i+1}^{\rot}\rp^{-1} \barr{M} \rsp \nonumber\\ 
      &= \Tr \lsp \lp Y - ww^{\rot} \rp^{-1} \barr{M} \rsp \nonumber\\
      &= \Tr \lsp \lp Y^{-1} + \frac{Y^{-1}ww^{\rot}Y^{-1}}{1 - w^{\rot}Y^{-1}w}\rp \barr{M}\rsp,  \label{eq:S-M} \\
      &= \Tr \lsp Y^{-1}\barr{M} \rsp + \frac{1}{1- w^{\rot}Y^{-1}w} \cdot \Tr \lsp Y^{-1}ww^{\rot}Y^{-1}\barr{M}\rsp \nonumber\\
      &\leq \Tr \lsp \lp \widehat{u}_{i+1} I - A^{(i)} \rp^{-1}\barr{M} \rsp + \frac{1}{1-2\varepsilon/q} \cdot w^{\rot}Y^{-1}\barr{M}Y^{-1}w \label{eq:hypot} \\
      &= \rho^{\widehat{u}_{i+1}}\left(A^{(i)}\right) + \frac{1}{1-2\varepsilon/q} \cdot w_{i+1}^{\rot} \lp \widehat{u}_{i+1} I - A^{(i)}\rp^{-1} \barr{M} \lp \widehat{u}_{i+1} I - A^{(i)}\rp^{-1} w_{i+1}. \nonumber
    \end{align}
    where \eqref{eq:S-M} comes from the Sherman-Morrison formula~(Lemma~\ref{lem:woodbury}) and \eqref{eq:hypot} from the hypothesis. 
    
    For the second statement, again let $\barr{Y} = \PV \lp B^{(i)} - \widehat{\ell}_{i+1} I \rp \PV$ and $w = Zw_{i+1}$. We have that 
    \begin{align}
      &\rho_{\widehat{\ell}_{i+1}}\left(B^{(i+1)}\right)
      = \Tr \lsp \PV \lp B^{(i)} + Zw_{i+1}w_{i+1}^{\rot}Z - \widehat{\ell}_{i+1} I \rp \PV\rsp^{\dag} \nonumber\\
      &= \Tr \lsp \barr{Y} + ww^{\rot}\rsp^{\dag} \nonumber\\
      &= \Tr \lsp \barr{Y}^{\dag} - \frac{\barr{Y}^{\dag}ww^{\rot}\barr{Y}^{\dag}}{1 + w^{\rot}\barr{Y}^{\dag}w}\rsp \label{eq:PSD S-M}\\
      &= \Tr \lsp \barr{Y}^{\dag} \rsp - \frac{1}{1+w^{\rot}\barr{Y}^{\dag}w} \cdot \Tr \lsp \barr{Y}^{\dag}ww^{\rot}\barr{Y}^{\dag} \rsp \nonumber\\
      &\leq \Tr \lsp \PV \lp B^{(i)} - \widehat{\ell}_{i+1} I \rp \PV \rsp^{\dag} - \frac{1}{1+2\varepsilon/q} \cdot w^{\rot}\barr{Y}^{\dag 2}w \label{eq:hypot2}\\
      &= \rho_{\widehat{\ell}_{i+1}}\left(B^{(i)}\right) - \frac{1}{1+2\varepsilon/q} \cdot w_{i+1}^{\rot}Z \lp \PV \lp B^{(i)} - \widehat{\ell}_{i+1} I \rp \PV \rp^{\dag 2}Zw_{i+1}, \nonumber
    \end{align}
    where \eqref{eq:PSD S-M} follows by Lemma~\ref{lem:pseud-inv SM} and \eqref{eq:hypot2} follows by the hypothesis of the lemma.
\end{proof}

The following lemma states that, assuming the event $W\preceq \frac{1}{2} (u I-A)$ occurs, both the potential functions and the total relative effective resistances are not increasing in expectation.  We remark that, in contrast to \cite{LS15:linearsparsifier}, the relative effective resistances are not only used for random sampling in each iteration of the algorithm, but also used for analysing the algorithm's performance. That's why the fact of the total  relative effective resistances being non-increasing is needed here.

\begin{lemma}\label{lem:potentials decrease in subphases}
It holds for any $0\leq i\leq N-1$ that   
\[
    \widetilde{\Ex} \lsp \Phi_{\hat{u}_{i+1}, \hat{\ell}_{i+1}}\left(A^{(i+1)}, B^{(i+1)}\right) \rsp \leq \Phi_{\hat{u}_i, \hat{\ell}_i} \left(A^{(i)}, B^{(i)}\right),
\]
and   
\[
  \widetilde{\Ex} \lsp \rho^{(i+1)} \rsp \leq \rho^{(i)}.
\]
\end{lemma}
 \begin{proof}
  We assume that the sampled matrix $W$ satisfies $\mathbf{0}\preceq W\preceq (1/2)\cdot (uI-A)$. Then, combining Lemma~\ref{lem:first bound w_i} and Lemma~\ref{lem:bound separate potentials} we have that 
    \[
        \Phi^{\hat{u}_{i+1}}\left(A^{(i+1)}\right) \leq \Phi^{\hat{u}_{i+1}}\left(A^{(i)}\right) + q (1+2\varepsilon) \cdot w_{i+1}^{\rot} \lp \hat{u}_{i+1} I - A^{(i)}\rp^{-(q+1)}w_{i+1}
    \]
    and 
    \[
        \Phi_{\hat{\ell}_{i+1}}\left(B^{(i+1)}\right) \leq \Phi_{\hat{\ell}_{i+1}}\left(B^{(i)}\right) - q(1-2\varepsilon) \cdot w_{i+1}^{\rot} Z \lp \PV \lp B^{(i)} - \hat{\ell}_{i+1} I \rp \PV \rp^{\dag (q+1)}Z w_{i+1}.
    \]
    Combining these and the definition of $\widetilde{\Ex}[\cdot]$, we have that 
    \begin{align*}
    \lefteqn{ \widetilde{\Ex} \lsp \Phi_{\hat{u}_{i+1}, \hat{\ell}_{i+1}}\left(A^{(i+1)}, B^{(i+1)}\right)\rsp }\\
       &\leq\Phi_{\hat{u}_{i+1}, \hat{\ell}_{i+1}}\left(A^{(i)}, B^{(i)}\right) + q(1+2\varepsilon) \cdot \widetilde{\Ex} \lsp w_{i+1}^{\rot} \lp \hat{u}_{i+1} I - A^{(i)}\rp^{-(q+1)}w_{i+1} \rsp \\
       &\qquad \qquad-q(1-2\varepsilon) \cdot \widetilde{\Ex} \lsp w_{i+1}^{\rot} Z \lp \PV \lp B^{(i)} - \hat{\ell}_{i+1} I \rp \PV \rp^{\dag (q+1)}Z w_{i+1} \rsp \\
       &= \Phi_{\hat{u}_{i+1}, \hat{\ell}_{i+1}}\left(A^{(i)}, B^{(i)}\right) + q(1+2\varepsilon) \cdot \Tr \lp \lp \hat{u}_{i+1} I - A^{(i)}\rp^{-(q+1)}  \widetilde{\Ex}\lsp w_{i+1}w_{i+1}^{\rot}\rsp \rp \\
       &\qquad\qquad-q(1-2\varepsilon) \cdot \Tr \lp \lp \PV \lp B^{(i)} - \hat{\ell}_{i+1} I \rp \PV \rp^{\dag (q+1)} Z \widetilde{\Ex}\lsp w_{i+1}w_{i+1}^{\rot}\rsp Z\rp \\
       &\leq \Phi_{\hat{u}_{i+1}, \hat{\ell}_{i+1}}\left(A^{(i)}, B^{(i)}\right) + q(1+2\varepsilon)(1+\varepsilon/2) \cdot \Tr \lp \lp \hat{u}_{i+1} I - A^{(i)}\rp^{-(q+1)} \Ex \lsp w_{i+1}w_{i+1}^{\rot}\rsp \rp \\
       &\qquad\qquad-q(1-2\varepsilon)(1-\varepsilon/2) \cdot \Tr \lp \lp \PV \lp B^{(i)} - \hat{\ell}_{i+1} I \rp \PV \rp^{\dag (q+1)} Z \Ex \lsp w_{i+1}w_{i+1}^{\rot}\rsp Z\rp \\
       &\leq \Phi_{\hat{u}_{i+1}, \hat{\ell}_{i+1}}\left(A^{(i)}, B^{(i)}\right) + q(1+3\varepsilon) \cdot \frac{\varepsilon}{q\cdot \rho} \cdot \Tr \lp \lp \hat{u}_{i+1} I - A^{(i)}\rp^{-(q+1)} \barr{M} \rp \\
       &\qquad\qquad-q(1-3\varepsilon) \cdot \frac{\varepsilon}{q \cdot \rho} \cdot \Tr \lp \lp \PV \lp B^{(i)} - \hat{\ell}_{i+1} I \rp \PV \rp^{\dag (q+1)} Z \barr{M} Z\rp \\
       &\leq \Phi_{\hat{u}_{i+1}, \hat{\ell}_{i+1}}\left(A^{(i)}, B^{(i)}\right) + q\cdot \barr{\delta_u} \cdot \Tr \lp \PL{A^{(i)}} \lp \hat{u}_{i+1} I - A^{(i)}\rp^{-(q+1)} \PL{A^{(i)}} \rp\\
       &\qquad \qquad -q \cdot \barr{\delta_\ell} \cdot \Tr \lp \PV \lp B^{(i)} - \hat{\ell}_{i+1} I \rp \PV \rp^{\dag (q+1)}, \\
        \end{align*}
    where the third inequality follows by Lemma~\ref{lem:expectation w_i}, and the fourth inequality follows by Lemma~\ref{lem: majorisation}. To upper bound $\widetilde{\Ex} \lsp \Phi_{\hat{u}_{i+1}, \hat{\ell}_{i+1}}\left(A^{(i+1)}, B^{(i+1)}\right)\rsp$ with respect to
    $\Phi_{\hat{u}_i, \hat{\ell}_i} \left(A^{(i)}, B^{(i)}\right)$, we define function

    \begin{align*}
      f_i(x) &= \Tr \lp \PL{A^{(i)}} \lp \left(\hat{u}_i + x\barr{\delta_u}\right) I - A^{(i)}\rp^{-q} \PL{A^{(i)}} \rp + \Tr \lp \PV \lp B^{(i)} - (\hat{\ell}_i + x\barr{\delta_{\ell}}) I \rp \PV \rp^{\dag q} \\
           &= \sum_{t=n-T+1}^n \lp \frac{1}{\hat{u}_i + x\barr{\delta_u} - \lambda_t\left(A^{(i)}\right)} \rp^q + \sum_{t=1}^k \lp \frac{1}{\lambda_t\left(B^{(i)}\big|_{S'}\right) - (\hat{\ell}_i + x\barr{\delta_{\ell}})}\rp^q. 
    \end{align*}
    By the convexity of the function $f$ we know that 
    \[
        f_i'(1) \geq f_i(1) - f_i(0) = \Phi_{\hat{u}_{i+1}, \hat{\ell}_{i+1}}\left(A^{(i)}, B^{(i)}\right) - \Phi_{\hat{u}_i, \hat{\ell}_i}\left(A^{(i)}, B^{(i)}\right).
    \]
    Therefore, it holds that
    \begin{align*}
    \lefteqn{\widetilde{\Ex} \lsp \Phi_{\hat{u}_{i+1}, \hat{\ell}_{i+1}}\left(A^{(i+1)}, B^{(i+1)}\right)\rsp}\\
   &  \leq \Phi_{\hat{u}_{i+1}, \hat{\ell}_{i+1}}\left(A^{(i)}, B^{(i)}\right) + q\cdot \barr{\delta_u} \cdot \Tr \lp \PL{A^{(i)}} \lp \hat{u}_{i+1} I - A^{(i)}\rp^{-(q+1)} \PL{A^{(i)}} \rp\\
     &    \qquad \qquad -q \cdot \barr{\delta_\ell} \cdot \Tr \lp \PV \lp B^{(i)} - \hat{\ell}_{i+1} I \rp \PV \rp^{\dag (q+1)}\\
    & = \Phi_{\hat{u}_{i+1}, \hat{\ell}_{i+1}}\left(A^{(i)}, B^{(i)}\right) - f_i'(1) \\
    & \leq \Phi_{\hat{u}_{i+1}, \hat{\ell}_{i+1}}\left(A^{(i)}, B^{(i)}\right) - \Phi_{\hat{u}_{i+1}, \hat{\ell}_{i+1}}\left(A^{(i)}, B^{(i)}\right) + \Phi_{\hat{u}_i, \hat{\ell}_i}\left(A^{(i)}, B^{(i)}\right)\\
    & =  \Phi_{\hat{u}_i, \hat{\ell}_i}\left(A^{(i)}, B^{(i)}\right),
        \end{align*}
    which proves the statement.

Next, we will prove that the conditional expectation of the sum of the relative effective resistances decreases as well. 
Conditioning on the event that $W\preceq (1/2)\cdot (uI-A)$, by  Lemmas~\ref{lem:first bound w_i} and \ref{lem:bound separate sum resistances} we have that 
  \[
    \rho^{\widehat{u}_{i+1}}\left(A^{(i+1)}\right) \leq \rho^{\widehat{u}_{i+1}}\left(A^{(i)}\right) + \frac{1}{1-2\varepsilon/q} \cdot w_{i+1}^{\rot} \lp \widehat{u}_{i+1} I - A^{(i)}\rp^{-1} \barr{M} \lp \widehat{u}_{i+1} I - A^{(i)}\rp^{-1} w_{i+1}
  \]
  and 
  \[
    \rho_{\widehat{\ell}_{i+1}}\left(B^{(i+1)}\right) \leq \rho_{\widehat{\ell}_{i+1}}\left(B^{(i)}\right) - \frac{1}{1+2\varepsilon/q} \cdot w_{i+1}^{\rot}Z \lp \PV \lp B^{(i)} - \widehat{\ell}_{i+1} I \rp \PV \rp^{\dag 2}Zw_{i+1}.
  \]
  We will  upper bound the two terms separately. To upper bound the first term, we have that  
  \begin{align*}
    \lefteqn{\tilde{\Ex} \lsp \rho^{\widehat{u}_{i+1}}\left(A^{(i+1)}\right)\rsp}\\
    &\leq \rho^{\widehat{u}_{i+1}}\left(A^{(i)}\right) + \frac{1}{1-2\varepsilon/q} \cdot \tilde{\Ex} \lsp w_{i+1}^{\rot} \lp \widehat{u}_{i+1} I - A^{(i)}\rp^{-1} \barr{M} \lp \widehat{u}_{i+1} I - A^{(i)}\rp^{-1} w_{i+1} \rsp \\
    &= \rho^{\widehat{u}_{i+1}}\left(A^{(i)}\right) + \frac{1}{1-2\varepsilon/q} \cdot \Tr \lsp \lp \widehat{u}_{i+1} I - A^{(i)}\rp^{-1} \barr{M} \lp \widehat{u}_{i+1} I - A^{(i)}\rp^{-1} \tilde{\Ex} \lsp w_{i+1}w_{i+1}^{\rot}\rsp \rsp \\
    &\leq \rho^{\widehat{u}_{i+1}}\left(A^{(i)}\right) + \frac{1 + \varepsilon/2}{1-2\varepsilon/q} \cdot \Tr \lsp \lp \widehat{u}_{i+1} I - A^{(i)}\rp^{-1} \barr{M} \lp \widehat{u}_{i+1} I - A^{(i)}\rp^{-1} \Ex \lsp w_{i+1}w_{i+1}^{\rot}\rsp \rsp \\
    &= \rho^{\widehat{u}_{i+1}}\left(A^{(i)}\right) + \frac{1 + \varepsilon/2}{1-2\varepsilon/q} \cdot \frac{\varepsilon}{q \cdot \rho}\cdot \Tr \lsp \lp \widehat{u}_{i+1} I - A^{(i)}\rp^{-1} \barr{M} \lp \widehat{u}_{i+1} I - A^{(i)}\rp^{-1} \barr{M} \rsp \\
    &\leq \rho^{\widehat{u}_{i+1}}\left(A^{(i)}\right) + \cdot \frac{(1 + 3\varepsilon) \varepsilon}{q \cdot \rho}\cdot \Tr \lsp \lp \widehat{u}_{i+1} I - A^{(i)}\rp^{-1} \barr{M} \lp \widehat{u}_{i+1} I - A^{(i)}\rp^{-1} \barr{M} \rsp \\
    &\leq \rho^{\widehat{u}_{i+1}}\left(A^{(i)}\right) + \barr{\delta_u}\cdot \Tr \lsp \lp \widehat{u}_{i+1} I - A^{(i)}\rp^{-2} \barr{M} \rsp,
  \end{align*}
  where the second inequality follows by Lemma~\ref{lem:bound expectation},   and the last equality follows by our choice   of $\delta_u$. 
     We use a similar technique to upper bound the second term, and have that  
  \begin{align*}
    &\tilde{\Ex} \lsp \rho_{\widehat{\ell}_{i+1}}\left(B^{(i+1)}\right)\rsp
    \leq \rho_{\widehat{\ell}_{i+1}}\left(B^{(i)}\right) - \frac{1}{1+2\varepsilon/q} \cdot \tilde{\Ex} \lsp w_{i+1}^{\rot}Z \lp \PV \lp B^{(i)} - \widehat{\ell}_{i+1} I \rp \PV \rp^{\dag 2}Zw_{i+1} \rsp \\
    &= \rho_{\widehat{\ell}_{i+1}}\left(B^{(i)}\right) - \frac{1}{1+2\varepsilon/q} \cdot \Tr \lsp \lp \PV \lp B^{(i)} - \widehat{\ell}_{i+1} I \rp \PV \rp^{\dag 2}Z\tilde{\Ex} \lsp w_{i+1}w_{i+1}^{\rot} \rsp Z\rsp \\
    &\leq \rho_{\widehat{\ell}_{i+1}}\left(B^{(i)}\right) - \frac{1-\varepsilon/2}{1+2\varepsilon/q} \cdot \Tr \lsp \lp \PV \lp B^{(i)} - \widehat{\ell}_{i+1} I \rp \PV \rp^{\dag 2}Z \Ex \lsp w_{i+1}w_{i+1}^{\rot} \rsp Z\rsp \\
    &= \rho_{\widehat{\ell}_{i+1}}\left(B^{(i)}\right) - \frac{1-\varepsilon/2}{1+2\varepsilon/q} \cdot \frac{\varepsilon}{q\cdot \rho}\cdot \Tr \lsp \lp \PV \lp B^{(i)} - \widehat{\ell}_{i+1} I \rp \PV \rp^{\dag 2}Z \barr{M} Z\rsp \\
    &\leq \rho_{\widehat{\ell}_{i+1}}\left(B^{(i)}\right) - \frac{(1-3\varepsilon) \varepsilon}{q\cdot \rho}\cdot \Tr \lsp \lp \PV \lp B^{(i)} - \widehat{\ell}_{i+1} I \rp \PV \rp^{\dag 2}Z \barr{M} Z\rsp \\
    &= \rho_{\widehat{\ell}_{i+1}}\left(B^{(i)}\right) - \barr{\delta_{\ell}}\cdot \Tr \lsp \lp \PV \lp B^{(i)} - \widehat{\ell}_{i+1} I \rp \PV \rp^{\dag 2}\rsp,   
  \end{align*}
  where the second inequality follows by Lemma~\ref{lem:bound expectation}. 
  
  To prove the statement of the lemma, we introduce function $f$ defined by \begin{align*}
    f(x) &\triangleq \Tr \lsp \lp (\widehat{u}_{i} + x\barr{\delta_u}) I - A^{(i)}\rp^{-1} \barr{M} \rsp + \Tr \lsp \lp \PV \lp B^{(i)} - (\widehat{\ell}_{i} + x\barr{\delta_{\ell}}) I \rp \PV \rp^{\dag}\rsp,
\end{align*}
and have that  
  \[
    f'(x) = -\barr{\delta_u} \cdot \Tr \lsp \lp (\widehat{u}_{i} + x\barr{\delta_u}) I - A^{(i)}\rp^{-2} \barr{M}\rsp + \barr{\delta_{\ell}} \cdot \Tr \lsp \lp \PV \lp B^{(i)} - (\widehat{\ell}_{i} + x\delta_{\ell}) I \rp \PV \rp^{\dag 2}\rsp.
  \]
  Since $f$ is a  convex function,  we have that 
  \[
    f'(1) \geq f(1) - f(0) = \lp \rho^{\widehat{u}_{i+1}}\left(A^{(i)}\right) + \rho_{\widehat{\ell}_{i+1}}\left(B^{(i)}\right) \rp - \lp \rho^{\widehat{u}_{i}}\left(A^{(i)}\right) + \rho_{\widehat{\ell}_{i}}\left(B^{(i)}\right) \rp.
\]
  Combining everything together, we have that 
  \begin{align*}
  \lefteqn{\tilde{\Ex} \lsp \rho^{(i+1)}\rsp }\\
  &= \tilde{\Ex}\lsp \rho^{\widehat{u}_{i+1}}\left(A^{(i+1)}\right) \rsp + \tilde{\Ex} \lsp\rho_{\widehat{\ell}_{i+1}}\left(B^{(i+1)}\right)\rsp \\
  &\leq \rho^{\widehat{u}_{i+1}}\left(A^{(i)}\right) + \barr{\delta_u}\cdot \Tr \lsp \lp \widehat{u}_{i+1} I - A^{(i)}\rp^{-2} \barr{M}\rsp\\
  &\qquad\qquad+ \rho_{\widehat{\ell}_{i+1}}\left(B^{(i)}\right) - \barr{\delta_{\ell}}\cdot \Tr \lsp \lp \PV \lp B^{(i)} - \widehat{\ell}_{i+1} I \rp \PV \rp^{\dag 2}\rsp \\
  &= \rho^{\widehat{u}_{i+1}}\left(A^{(i)}\right) + \rho_{\widehat{\ell}_{i+1}}\left(B^{(i)}\right) - f'(1)\\
  & \leq \rho^{\widehat{u}_{i+1}}\left(A^{(i)}\right) + \rho_{\widehat{\ell}_{i+1}}\left(B^{(i)}\right) - \lp \rho^{\widehat{u}_{i+1}}\left(A^{(i)}\right) + \rho_{\widehat{\ell}_{i+1}}\left(B^{(i)}\right) \rp + \lp \rho^{\widehat{u}_{i}}\left(A^{(i)}\right) + \rho_{\widehat{\ell}_{i}}\left(B^{(i)}\right) \rp \\
  & = \rho^{(i)},
  \end{align*}
  which proves the claimed statement.
\end{proof}


\subsection{On the total number of iterations and  the   number of chosen vectors}\label{sec:total number of iterations and vectors}

In this subsection we prove that with constant probability the algorithm samples $\Theta(k/\varepsilon^2)$ vectors.  The following technical lemmas will be used in our analysis.

\begin{lemma}\label{lem:Phi_0}
  It holds that
  \[
    \Phi_{u_0, \ell_0}(A_0, B_0)
    \leq T \cdot \lp u_0 - \lambdaMax(X)\rp^{-q} + k \cdot \lp -\ell_0\rp^{-q}
  \]
  and hence 
  \[
    \Phi_{u_0, \ell_0}(A_0, B_0)^{1/q} \leq T^{1/q} \cdot \lp u_0 - \lambdaMax(X)\rp^{-1} + k^{1/q} \cdot \lp -\ell_0\rp^{-1}.
  \]
\end{lemma}
\begin{proof}
By definition, we have that $A_0=X$, and $B_0=Z(A_0-X)Z=\mathbf{0}$. Therefore, it holds that 
\[
    \Phi^{u_0}(A_0) = \sum_{i=n-T+1}^n \lp \frac{1}{u_0 - \lambda_i(X)}\rp^q \leq \sum_{i=n-T+1}^n \lp \frac{1}{u_0 - \lambdaMax(X)}\rp^q = T\cdot \lp u_0 - \lambdaMax(X)\rp^{-q},
\]
and 
\[
    \Phi_{\ell_0}(B_0) = \sum_{i=1}^k \lp \frac{1}{\lambda_i(\mathbf{0}) - \ell_0}\rp^q = k\cdot (-\ell_0)^{-q}.
\]
Combining the two inequalities above gives us that 
\[
    \Phi_{u_0, \ell_0}(A_0, B_0) = \Phi^{u_0}(A_0) + \Phi_{\ell_0}(B_0) \leq T\cdot \lp u_0 - \lambdaMax(X)\rp^{-q} + k\cdot (-\ell_0)^{-q}.
\]
To prove the second statement, notice that it holds for $a,b\in\mathbb{R}^+$ and $q\in\mathbb{Z}^+$ that 
\[
    (a + b)^{1/q} \leq a^{1/q} + b^{1/q}.
\]
By setting $a= \Phi^{u_0}(A_0)$ and $b=\Phi_{\ell_0}(B_0)$ and applying the inequality above, we prove the second statement of the lemma.
\end{proof}

\begin{lemma}\label{lem:rho_0}
  It holds that 
  \[
    \rho_0 \leq T \cdot \lp u_0 - \lambdaMax(X)\rp^{-1} + k \cdot (-\ell_0)^{-1} \leq \frac{T + \Lambda}{2}.
  \]
\end{lemma}
\begin{proof} 
By the definition of $\rho_0$, we have that 
  \begin{align*}
    \rho_0 &= \Tr \lsp (u_0 I - X)^{-1}\barr{M} \rsp + \Tr \lsp \PV(\mathbf{0} - \ell_0I)\PV\rsp^{\dag} \\
    &\leq \sum_{i=n-T+1}^n \frac{1}{u_0 - \lambda_i(X)} + \sum_{i=1}^k \frac{1}{0 - \ell_0} \\
    &\leq T \cdot (u_0 - \lambdaMax(X))^{-1} + k \cdot (-\ell_0)^{-1} \\
    &=\frac{T}{2} + \frac{\Lambda}{2}, 
  \end{align*}
  where  the first inequality follows by   Lemma~\ref{lem: majorisation}  and the last equality follows by our choice of $u_0$ and $\ell_0$. 
\end{proof}

\begin{lemma}\label{lem:lower bound rho_j}
It holds for any iteration $j$ that  
  \[
    \rho_j \geq \frac{T + k - 1}{u_0 - \ell_0 + \frac{1+3\varepsilon}{6\varepsilon} \cdot \sum_{t=0}^{j-1}(\delta_{u, t} - \delta_{\ell, t})}
  \]
\end{lemma}
\begin{proof}
By the definition of $\rho_j$, we have that 
  \[
    \rho_j = \Tr \lsp (u_j I - A_j)^{-1}\barr{M}\rsp + \Tr \lsp \PV (B_j - \ell_j I)\PV\rsp^{\dag}. 
  \]
  We will analyse the two terms of $\rho_j$ separately. To study the first term, we use the fact of $(u_j I - A_j)^{-1} \succeq (u_j I)^{-1}$ and  Lemma~\ref{fact:PSD-invariant} to obtain that  
  \begin{align*}
    \Tr \lsp (u_j I - A_j)^{-1}\barr{M}\rsp 
    &= \Tr \lsp \barr{M}^{1/2}(u_j I - A_j)^{-1}\barr{M}^{1/2}\rsp\geq \Tr \lsp \barr{M}^{1/2}(u_j I)^{-1}\barr{M}^{1/2}\rsp \\
    & = \frac{1}{u_j}\cdot \Tr[\barr{M}]
     \geq \frac{T-1}{u_j},
  \end{align*}
  where the last  inequality follows by the fact that $T=\left\lceil \Tr\left[\overline{M}\right]\right\rceil \leq \Tr\left[\overline{M}\right]+1$.
  Combining this with the fact of 
  \[
    \frac{\delta_{u, j} - \delta_{\ell, j}}{\delta_{u, j}} = \frac{6\varepsilon}{1+3\varepsilon},
  \]
  which implies that 
  \[
    \delta_{u, j} = \frac{1+3\varepsilon}{6\varepsilon} \cdot (\delta_{u, j} - \delta_{\ell, j}),
  \]
  we obtain that   
  \begin{equation}\label{eq:bound1}
    \Tr \lsp (u_j I - A_j)^{-1}\barr{M}\rsp 
    \geq \frac{T-1}{u_j}
    = \frac{T-1}{u_0 + \sum_{t=0}^{j-1}\delta_{u, t}}
    \geq \frac{T-1}{u_0 - \ell_0 + \sum_{t=0}^{j-1} \frac{1+3\varepsilon}{6\varepsilon}\cdot (\delta_{u, t} - \delta_{\ell, t})},
  \end{equation}
  where the last inequality uses the fact that $\ell_0<0$.
  
  Now we study the second term.  We will first argue that $\lambdaMax(B_j \big|_{S'}) < u_j$. Note that
    \[
        B_j = Z(A_j - X)Z = Z A_j Z - ZXZ.
    \]
    Let $v \in S'$  be the eigenvector corresponding to the largest eigenvalue of $B_j\big|_{S'}$. We have that
    \[
        v^{\rot} B_j v + (v^{\rot}Z)X(Zv) = (v^{\rot}Z)A_j(Zv). 
    \]
    Since the matrix $X$ is PSD, we have that 
   $
        \lambdaMax(B_j\big|_{S'}) = v^{\rot} B_j v \leq \lambdaMax(A_j) < u_j. 
    $
    Therefore, it holds that 
    \begin{equation}\label{eq:bound2}
        \Tr \lsp \PV (B_j - \ell_j I)\PV\rsp^{\dag} = \sum_{i=1}^k \frac{1}{\lambda_i(B_j\big|_{S'}) - \ell_j} 
        \geq \frac{k}{u_j - \ell_j}
        = \frac{k}{u_0 - \ell_0 + \sum_{t=0}^{j-1}(\delta_{u, t} - \delta_{\ell, t})}
    \end{equation}
    Combining (\ref{eq:bound1}) with (\ref{eq:bound2}), we have that 
    \begin{align*}
        \rho_j &\geq \frac{T-1}{u_0 - \ell_0 +  \sum_{t=0}^{j-1}\frac{1+3\varepsilon}{6\varepsilon}\cdot (\delta_{u, t} - \delta_{\ell, t})} + \frac{k}{u_0 - \ell_0 + \sum_{t=0}^{j-1}(\delta_{u, t} - \delta_{\ell, t})} \\
        &\geq \frac{T + k-1}{u_0 - \ell_0 + \sum_{t=0}^{j-1}\frac{1+3\varepsilon}{6\varepsilon}\cdot  \left(\delta_{u, t} - \delta_{\ell, t}\right)},
    \end{align*}
    where the last inequality follows by our choice of $\varepsilon$.
\end{proof}

\begin{lemma}\label{lem:upper bound rho_j}
It holds for any iteration $j$  that   
\[
    \rho_j \leq  \Phi_{u_j, \ell_j}^{1/q}(A_j, B_j) \cdot (T+k)^{1-1/q}.
\]
\end{lemma}
\begin{proof}
  \begin{align*}
    \rho_j &  = \sum_{i=1}^m v_i^{\rot}\uRes{u_j}{A_j} v_i + v_i^{\rot}\lRes{\ell_j}{B_j}v_i \\
    &= \Tr \lp \uRes{u_j}{A_j} \sum_{i=1}^m v_iv_i^{\rot} \rp + \Tr \lp \lRes{\ell_j}{B_j} \sum_{i=1}^m v_iv_i^{\rot}\rp \\
    &= \Tr \lp \uRes{u_j}{A_j} \barr{M} \rp + \Tr \lp \lRes{\ell_j}{B_j} \barr{M}\rp \\
    &\leq \sum_{i=n-T+1}^n \lambda_i \uRes{u_j}{A_j} + \Tr \lp \lp \PV(B_j - \ell_j I)\PV\rp^\dag Z \barr{M} Z \rp \\
    &= \sum_{i=n-T+1}^n \lp \frac{1}{u_j - \lambda_i(A_j)} \rp + \sum_{i=1}^k \lp \frac{1}{\lambda_i(B_j\big|_{S'}) - \ell_j} \rp \\
    &\leq \lsp \sum_{i=n-T+1}^n \lp \frac{1}{u_j - \lambda_i(A_j)}\rp^q + \sum_{i=1}^k \lp \frac{1}{\lambda_i(B_j\big|_{S'}) - \ell_j} \rp^q \rsp^{1/q} \cdot (T+k)^{1-1/q}\\
    &= \Phi_{u_j, \ell_j}^{1/q}(A_j, B_j) \cdot \lp T+k\rp^{1-1/q},
  \end{align*}
  where the first inequality comes from Lemma~\ref{lem: majorisation} and the second
  inequality follows by the   H\"{o}lder's inequality. 
\end{proof}

We will first prove that there are a sufficient number of vectors sampled in each iteration.

\begin{lemma}\label{lem:lower bound N_j}
The total number of vectors sampled in each iteration $j$ satisfies 
    \[
        N_j \geq c_N \cdot \rho_j^{1-2\varepsilon/q} \cdot \Phi_{u_j, \ell_j}(A_j, B_j)^{-1/q},
    \]
    for some $c_N=\Omega\left( \left(1/(\mathrm{poly}(n)\right)^{2\varepsilon/q} \right)$.
\end{lemma}
\begin{proof}
We first recall that 
  \begin{align*}
  N_j &=  \lp \frac{\varepsilon}{4\rho_j}\cdot \lambdaMin\lsp(u_j I-A_j)^{-1}\barr{M}\rsp \cdot \frac{\lambdaMax\lp(u_j I - A_j)^{-1}\barr{M}\rp}{\Tr \lsp(u_j I - A_j)^{-1}\barr{M}\rsp}\rp ^{2\varepsilon/q} \cdot \rho_j \\
&\cdot \mathrm{min} \left\{ \frac{1}{\lambdaMax\lp (u_j I - A_j)^{-1}\barr{M}\rp}, \frac{1}{\lambdaMax\lp \PV (B_j - \ell_j I) \PV\rp^{\dag}} \right\}.
  \end{align*}
  To prove the claimed statement, we lower bound the terms involved above separately.  For simplicity, let us denote
  \[
    \Phi_j \triangleq \Phi_{u_j, \ell_j}(A_j, B_j).
  \]
  We will first prove that  
  \[
     \mathrm{min} \left\{ \frac{1}{\lambdaMax\lp (u_j I - A_j)^{-1}\barr{M}\rp}, \frac{1}{\lambdaMax\lp \PV (B_j - \ell_j I) \PV\rp^{\dag}} \right\}
     \geq \Phi_j^{-1/q},
  \]
  which is equivalent to showing
  \begin{equation}\label{eq:bound111}
    \Phi_j^{1/q} \geq \max \left\{ \lambdaMax\lp (u_j I - A_j)^{-1}\barr{M}\rp, \lambdaMax\lp \PV (B_j - \ell_j I) \PV\rp^{\dag} \right\}.
  \end{equation}
  Note that 
  \begin{align*}
      \Phi_j^{1/q} &= \lp \Tr \lsp \PL{A_j} \lp u_j I - A_j\rp^{-q}\PL{A_j}\rsp + \Tr \lsp \PV \lp B_j - \ell_j I\rp \PV\rsp^{\dag q} \rp^{1/q} \\
      &\geq \max \left\{ \Tr \lsp \PL{A_j} \lp u_j I - A_j\rp^{-q}\PL{A_j}\rsp, \Tr \lsp \PV \lp B_j - \ell_j I\rp \PV\rsp^{\dag q} \right\}^{1/q}\\
      &= \max \left\{ \sum_{i=n-T+1}^n \lp \frac{1}{u_j - \lambda_i(A_j)}\rp^q , \sum_{i=1}^k \lp \frac{1}{\lambda_i(B_j) - \ell_j}\rp^{q}\right\}^{1/q}\\
      &\geq \max \left\{ \lp \frac{1}{u_j - \lambdaMax(A_j)}\rp^q , \lp \frac{1}{\lambdaMin(B_j) - \ell_j}\rp^q \right\}^{1/q}\\
      &= \max \left\{ \frac{1}{u_j - \lambdaMax(A_j)} , \frac{1}{\lambdaMin(B_j) - \ell_j} \right\}\\
      &= \max \left\{ \lambdaMax\lp (u_j I - A_j)^{-1}\rp, \lambdaMax\lp \PV (B_j - \ell_j I) \PV\rp^{\dag} \right\}\\
      &\geq \max \left\{ \lambdaMax\lp (u_j I - A_j)^{-1}\barr{M}\rp, \lambdaMax\lp \PV (B_j - \ell_j I) \PV\rp^{\dag} \right\},
  \end{align*}
  where the last inequality comes from the fact that $\barr{M} \preceq I.$ With this we proved (\ref{eq:bound111}).

  We continue with lower bounding other terms involved in the definition of $N_j$.
We have that   \[
    \frac{\lambdaMax\lp(u_j I - A_j)^{-1}\barr{M}\rp}{\Tr \lsp(u_j I - A_j)^{-1}\barr{M}\rsp} \geq \frac{1}{n},
  \]
  and
  \begin{align*}
    \lambdaMin \lsp (u_j I - A_j)^{-1}\barr{M}\rsp
    &\geq \lambdaMin \lsp (u_j I)^{-1}\barr{M}\rsp
    = \frac{\lambdaMin\left(\barr{M}\right)}{u_j} \\
    &= \frac{\lambdaMin\left(\barr{M}\right)}{u_0 + \sum_{t=0}^{j-1}\delta_{u, t}}
    \geq \frac{\lambdaMin\left(\barr{M}\right)}{u_0 - \ell_0 + \frac{1+3\varepsilon}{6\varepsilon}\cdot \sum_{t=0}^{j-1}(\delta_{u, t} - \delta_{\ell, t})} \\
    &\geq \frac{\lambdaMin\left(\barr{M}\right)}{u_0 - \ell_0 + \frac{1+3\varepsilon}{6\varepsilon} \cdot \alpha} \geq \frac{\lambdaMin\left(\barr{M}\right)}{3 + \frac{1 + 3\varepsilon}{\varepsilon}}
    = \frac{\varepsilon}{1+6\varepsilon} \cdot \lambdaMin\left(\barr{M}\right).
  \end{align*}
 Combining  everything together, we have  that
  \begin{align*}
    N_j &\geq \lp\frac{\varepsilon}{4\rho_j} \cdot \frac{\varepsilon}{1+6\varepsilon}\cdot \lambdaMin\left(\barr{M}\right) \cdot \frac{1}{n}\rp^{2\varepsilon/q} \cdot \rho_j \cdot \Phi_j^{-1/q} \\
    &\geq \frac{1}{16} \cdot \lambdaMin\left(\barr{M}\right)^{2\varepsilon/q} \cdot n^{-2\varepsilon/q}\cdot \rho_j^{1-2\varepsilon/q} \cdot \Phi_j^{-1/q}.
  \end{align*} 
  To prove the claimed statement, we define
  \[
        c_N = \frac{1}{16}\cdot\lambdaMin(\barr{M})^{2\varepsilon/q} \cdot n^{-2\varepsilon/q},
    \] 
    and it suffices to prove that $\lambda_{\min}(\overline{M}) = \Omega(1/\mathrm{poly}(n))$.       By definition, we know that 
    \[
    \overline{M}=I-X = L_{G+W}^{\dagger/2} \left( L_{G+W} - L_G \right) L_{G+W}^{\dagger/2}.
    \]
    Combining this with the Courant-Fischer formulation of the eigenvalues, we know that 
    \begin{align*}
        \lambda_{\min} \left( \overline{M} \right) & =  \min_{\substack{ x\in\mathbb{R}^n \\ x\neq 0}}
        \frac{x^{\rot} L_{G+W}^{\dagger/2} \left( L_{G+W} - L_G \right) L_{G+W}^{\dagger/2} x}{x^{\rot} x}.
    \end{align*}
    By setting $ y = L_{G+W}^{\dagger/2} x$, we have   $x=L^{1/2}_{G+W} y$ and 
    \begin{align}
    \lefteqn{\lambda_{\min} \left( \overline{M} \right)} \nonumber \\ & = 
    \min_{\substack{ y\in\mathbb{R}^n \nonumber  \\ y\neq 0, y\neq 1}}
    \frac{y^{\rot} (L_{G+W} - L_G) y }{y^{\rot} L_{G+W} y} \nonumber \\
    & = \min_{\substack{ y\in\mathbb{R}^n \\ y\neq 0, y\neq 1}} \frac{\sum_{  \{u,v\}\in E(  W)} w_{u,v}\cdot \left(y_u-y_v\right)^2  + \sum_{u\in V} \gamma y_u^2}{\sum_{  \{u,v\}\in E(G \cup  W)} w_{u,v}\cdot\left(y_y-u_v\right)^2+\sum_{u\in V} \gamma y_u^2   }\nonumber \\
    & = \min_{\substack{ y\in\mathbb{R}^n \\ y\neq 0, y\neq 1}} \left(\frac{\sum_{  \{u,v\}\in E(  W)} w_{u,v}\cdot \left(y_u-y_v\right)^2  + \sum_{u\in V} \gamma y_u^2  }{ y^{\rot} y} \cdot \frac{ y^{\rot } y }{\sum_{  \{u,v\}\in E(G \cup  W)} w_{u,v}\cdot\left(y_y-u_v\right)^2 + \sum_{u\in V} \gamma y_u^2} \right), \nonumber 
    \end{align}
    where $\gamma$ corresponds to the $n$ self-loops we introduced artificially to ensure that $\overline{M}$ has full rank.
 To bound the term above,
 notice that
 \begin{equation}\label{eq:boundlambmin}
 \min_{\substack{ y\in\mathbb{R}^n \\ y\neq 0, y\neq 1}} \left(\frac{\sum_{  \{u,v\}\in E(  W)} w_{u,v}\cdot \left(y_u-y_v\right)^2  + \sum_{u\in V} \gamma y_u^2  }{ y^{\rot} y}   \right)\geq \gamma,
 \end{equation}
 and 
 \begin{align*}
\lefteqn{ \max_{\substack{ y\in\mathbb{R}^n \\ y\neq 0, y\neq 1}} \left( \frac{ \sum_{  \{u,v\}\in E(G \cup  W)} w_{u,v}\cdot\left(y_y-u_v\right)^2  + \sum_{u\in V} \gamma_u y_u^2}{ y^{\rot }y} \right) }\\
 \leq & \max_{\substack{ y\in\mathbb{R}^n \\ y\neq 0, y\neq 1}} \left( \frac{\sum_{\{u,v\}\in E(G+W)} 2\left(\deg(u)\cdot y_u^2 + \deg(v) y^2_v\right) + \sum_{ u\in V}\gamma y_u^2 }{\sum_{u\in V[G]} y_u^2 } \right)\\
 &\leq 2\Delta(G+W) +\gamma,
     \end{align*}
where $\Delta(G+W) $ is the maximum degree of $G+W$ and $\Delta(G+W) =O(\mathrm{poly}(n))$ since the edge weight of $G$ is polynomial bounded.    Combining these with (\ref{eq:boundlambmin}) proves the statement.
\end{proof}
Next we prove that the algorithm finishes in a sub-linear number of iterations.

\begin{lemma}\label{lem:number of rounds}
    Assume the   number of sampled vectors in iteration $j$   satisfies
    \[
        N_j \geq c_N \cdot \rho^{1-2\varepsilon/q}_j \cdot \Phi_{u_j, \ell_j}(A_j, B_j)^{-1/q},
    \]
    for some coefficient $c_N$ independent of $j$. Then, with probability at least $4/5$, Algorithm~\ref{alg:subgraph sparsifier} finishes in at most 
    \[
        \tau = \frac{10 \cdot \alpha\cdot q \cdot\rho_0^{2\varepsilon/q}\cdot \Phi_{u_0, \ell_0}(A_0, B_0)^{1/q}}{6\varepsilon^2 \cdot c_N} \leq \frac{80q}{3\varepsilon^2}\cdot \frac{1}{c_N} \cdot \Lambda^{(1+2\varepsilon)/q}
    \]
    iterations, where $c_N=\Omega\left( \left(1/(\mathrm{poly}(n)\right)^{2\varepsilon/q} \right)$.
\end{lemma}
\begin{proof}
  For simplicity let us denote
  \[
    \Phi_j \triangleq \Phi_{u_j, \ell_j}(A_j, B_j).
  \]
  First, notice that for every round $j$, it holds that 
  \[
    \delta_{u, j} - \delta_{\ell, j} = \frac{6\varepsilon^2}{q} \cdot \frac{N_j}{\rho_j}
    \geq \frac{6\varepsilon^2}{q} \cdot c_N \cdot \rho_j^{-2\varepsilon/q} \cdot \Phi_j^{-1/q}, 
  \]
  where the last inequality comes from our assumption. Suppose $\tau$ is the last round in which we sample vectors. Then, it holds that 
  \begin{align}
     \Pro \lsp \text{the algorithm finishes after $\tau$ rounds}\rsp  
    &\geq \Pro \lsp \sum_{j=0}^{\tau-1}\delta_{u,j} - \delta_{\ell, j} > \alpha \rsp \nonumber\\
    &\geq \Pro \lsp \sum_{j=0}^{\tau-1} \frac{6\varepsilon^2}{q} \cdot c_N \cdot \rho_j^{-2\varepsilon/q} \cdot \Phi_j^{-1/q} \geq \alpha \rsp \nonumber\\
    &= \Pro \lsp \sum_{j=0}^{\tau-1} \rho_j^{-2\varepsilon/q} \cdot \Phi_j^{-1/q} \geq \frac{q\cdot \alpha}{6\varepsilon^2 \cdot c_N} \rsp \nonumber\\
    &\geq \Pro \lsp \sum_{j=0}^{\tau-1} \rho_j^{2\varepsilon/q} \cdot \Phi_j^{1/q} \leq \tau^2 \cdot \frac{6\varepsilon^2 \cdot c_N}{  \alpha\cdot q} \rsp \label{eq:alg ends2},
  \end{align}
  where \eqref{eq:alg ends2} comes from the following inequality
  \[
    \lp \sum_{j=0}^{\tau - 1} \lp\rho_j^{2\varepsilon}\cdot \Phi_j\rp^{1/q} \rp \lp \sum_{j=0}^{\tau - 1} \lp\rho_j^{2\varepsilon}\cdot \Phi_j\rp^{-1/q}\rp \geq \tau^2.
  \]
  
Now we bound the probability of the opposite event. By applying the Cauchy-Schwarz inequality for random variables, which states that 
\[
    \Ex \lsp XY \rsp \leq \sqrt{\Ex \lsp X^2\rsp \cdot \Ex \lsp Y^2\rsp}.  
  \]
for any random variables $X$ and $Y$, we have that 
\begin{align*}
    \tilde{\Ex} \lsp \sum_{j=0}^{\tau-1} \rho_j^{2\varepsilon/q} \cdot \Phi_j^{1/q}\rsp 
    &= \sum_{j=0}^{\tau-1} \tilde{\Ex} \lsp \rho_j^{2\varepsilon/q} \cdot\Phi_j^{1/q}\rsp
    \leq \sum_{j=0}^{\tau-1} \sqrt{\tilde{\Ex}\lsp \rho_j^{4\varepsilon/q}\rsp \cdot \tilde{\Ex}\lsp \Phi_j^{2/q}\rsp}\\
    &\leq \sum_{j=0}^{\tau-1} \sqrt{\lp\tilde{\Ex}\lsp \rho_j\rsp\rp^{4\varepsilon/q} \cdot \lp\tilde{\Ex}\lsp \Phi_j\rsp\rp^{2/q}}
    = \sum_{j=0}^{\tau-1} \lp\tilde{\Ex}\lsp \rho_j\rsp\rp^{2\varepsilon/q} \cdot \lp\tilde{\Ex}\lsp \Phi_j\rsp\rp^{1/q}\\
    &\leq \sum_{j=0}^{\tau-1} \rho_0^{2\varepsilon/q}\cdot \Phi_0^{1/q} 
    \leq \tau \cdot \rho_0^{2\varepsilon/q}\cdot \Phi_0^{1/q},
  \end{align*}
  where the second inequality follows by the Jensen's inequality.  By Lemma~\ref{lem:bounded W} and the union bound, it holds with constant probability that the sampled matrix $W_j$ in iteration $j$ satisfies $W_j\preceq (1/2)\cdot \left( u_j I - A_j \right)$ over the first $O(n)$ iterations. Therefore, we  have that 
  \begin{align*}
  &\Pro \lsp \sum_{j=0}^{\tau-1} \Phi_j^{1/q} \geq \tau^2 \cdot \frac{6\varepsilon^2 \cdot c_N}{  \alpha\cdot q} \rsp \\
  &\leq \Pro \lsp \sum_{j=0}^{\tau-1} \Phi_j^{1/q} \geq \tau^2 \cdot \frac{6\varepsilon^2 \cdot c_N}{  \alpha\cdot q} \Big| \forall j: W_j \preceq \frac{1}{2}\cdot (u_j I - A_j)\rsp + \Pro \lsp \exists j: W_j \npreceq \frac{1}{2} \cdot (u_j I - A_j)\rsp\\
  &\leq\frac{\tilde{\Ex} \lsp \sum_{j=0}^{\tau-1} \rho_j^{2\varepsilon/q}\Phi_j^{1/q} \rsp}{\tau^2 \cdot   6\varepsilon^2 \cdot c_N/ (\alpha\cdot q) }+ \frac{1}{10}
  \leq \frac{  \alpha  \cdot \rho_0^{2\varepsilon/q}\cdot \Phi_0^{1/q}}{6\varepsilon^2 \cdot c_N\cdot \tau} + \frac{1}{10} \leq \frac{1}{5},
  \end{align*}
  where the last inequality comes from the choice of $\tau$.
\end{proof}

By Lemma~\ref{lem:number of rounds}, with constant probability the algorithm finishes in a sub-linear number of iterations. Combining this with the union bound, we know that   the sampled matrix  $W_j$ satisfies $W_j\preceq (1/2)\cdot (u_j I - A_j) $ for all iterations.  This allows us to show that the algorithm terminates after choosing $\Theta(k/\varepsilon^2)$ vectors.

\begin{lemma}\label{lem:number of sampled vectors} 
With probability at least $3/4$, Algorithm~\ref{alg:subgraph sparsifier} terminates after choosing at most  
    \[
        K  = \frac{20\cdot  q\ k}{3\cdot \varepsilon^2}
    \]
vectors.  
\end{lemma}
\begin{proof}
Without loss of generality, we assume that $v_1,\ldots, v_K$ are the vectors sampled by the algorithm and every vector $v_i$ is sampled in iteration $\tau_i$. By the algorithm description, we know that the probability that the algorithm terminates after sampling $K$ vectors is at least 
  \begin{align*}  
  \Pro \lsp \widehat{u} - \widehat{\ell} > \alpha + u_0 - \ell_0\rsp  
  &= \Pro \lsp \sum_{j=1}^{K} \frac{6\varepsilon^2}{q} \cdot \frac{1}{\rho_{\tau_j}} \geq \alpha \rsp  
  = \Pro \lsp\sum_{j=1}^{K} \rho_{\tau_j}^{-1} \geq \frac{  \alpha\cdot q}{6\varepsilon^2} \rsp \\ 
  &\geq \Pro \lsp \sum_{j=1}^{K} \rho_{\tau_j} \leq   \frac{6\varepsilon^2 K^2}{  \alpha\cdot q} \rsp 
  \end{align*}
  where the last inequality follows by the fact that 
  \[
    \lp \sum_{j=1}^{K} \rho_{\tau_j} \rp \lp \sum_{j=1}^{K} \rho_{\tau_j}^{-1}\rp \geq K^2.
  \]
  
  To prove the claimed statement, we upper bound the probability that $\sum_{j=1}^{K} \rho_{\tau_j} \leq   6\varepsilon^2 K^2 /  ( \alpha\cdot q)$ occurs. By Lemma~\ref{lem:bounded W} we know that with probability at least $1-\varepsilon/2$ it holds for all  the iterations $j$ that 
  $W_j \preceq \frac{1}{2} \cdot (u_j I - A_j)$, under which condition by Lemma~\ref{lem:potentials decrease in subphases} we have that   \[
    \widetilde{\Ex} \lsp \sum_{j=1}^{K} \rho_{\tau_j}\rsp 
    = \sum_{j=1}^{K} \tilde{\Ex} \lsp \rho_{\tau_j}\rsp
    \leq K \cdot \rho_0 \leq K \cdot \Lambda
  \]
  where the last  inequality follows by 
   Lemma~\ref{lem:rho_0}. Therefore, by Markov inequality we have that 
  \begin{align*}
  &\Pro \lsp \sum_{j=1}^{K} \rho_{\tau_j} \geq K^2 \cdot \frac{6\varepsilon^2}{  \alpha\cdot q} \rsp \\
  &\leq \Pro \lsp \sum_{j=1}^{K} \rho_{\tau_j} \geq K^2 \cdot \frac{6\varepsilon^2}{\alpha\cdot q} ~\Big|~ \forall j: W_j \preceq \frac{1}{2}\cdot (u_j I - A_j)\rsp + \Pro \lsp \exists j: W_j \npreceq \frac{1}{2} \cdot (u_j I - A_j)\rsp\\
  &\leq\frac{\tilde{\Ex} \lsp \sum_{j=1}^{K} \rho_{\tau_j} \rsp}{K^2 \cdot \frac{6\varepsilon^2}{\alpha\cdot q}}+ \frac{1}{10}
  \leq \frac{\alpha\cdot q \cdot \Lambda}{6\varepsilon^2 \cdot K} + \frac{1}{10} = \frac{1}{5},
  \end{align*}
  where the last equality comes from the choice of $K$.  Therefore, with probability at least $4/5$, the algorithm samples at most $20qk/(3\varepsilon^2)$ vectors. 
\end{proof}

  
\subsection{Runtime analysis}\label{sec:runtime analysis}

Now we discuss fast approximation of the quantities needed for our subgraph sparsification algorithm, and the impact of our approximation on the overall algorithm's performance.  
For simplicity, we drop the subscript $j$ representing the iteration. By the algorithm description, we know that the number of vectors sampled by the algorithm in the iteration is 
\begin{align*}
  N =  &\lp \frac{\varepsilon}{4\rho}\cdot \lambdaMin\lsp(u I-A)^{-1}\barr{M}\rsp \cdot \frac{\lambdaMax\lp(u I - A)^{-1}\barr{M}\rp}{\Tr \lsp(u I - A)^{-1}\barr{M}\rsp}\rp ^{2\varepsilon/q} \cdot \rho \\
  &\cdot \mathrm{min} \left\{ \frac{1}{\lambdaMax\lp (u I - A)^{-1}\barr{M}\rp}, \frac{1}{\lambdaMax\lp \PV (B - \ell I) \PV \rp^{\dag}} \right\},
\end{align*}
and within the iteration every vector $v_i$ is sampled with probability proportional to its relative effective resistance defined by  \[
  R_i(A, B, u, \ell) = v_i^{\rot}\uRes{u}{A} v_i + v_i^{\rot}Z\lp \PV \lp B - \ell I\rp \PV\rp^{\dag} Zv_i.
\]
Therefore, the efficiency of our subgraph sparsification algorithm  is based  on the fast approximation of the following   quantities: 
\begin{enumerate}\label{enum: approximated quantities}
   \item $v_i^{\rot}Z\lp \PV \lp B - \ell I\rp \PV\rp^{\dag} Zv_i$\label{enum: resistance lower}
   \item $\lambdaMax\lp \PV (B - \ell I) \PV\rp^{\dag}$\label{enum: lambda_max lower}
   
    \item $\lambdaMin\lsp(u I-A)^{-1}\barr{M}\rsp$ \label{enum: lambda_min upper}
    \item $\lambdaMax\lsp(u I - A)^{-1}\barr{M}\rsp$\label{enum: lambda_max upper}
    \item $\Tr \lsp(u I - A)^{-1}\barr{M}\rsp$\label{enum: Trace upper}
    \item $v_i^{\rot}\uRes{u}{A} v_i$\label{enum: resistance upper}
\end{enumerate}
Without loss of generality, we first assume the following   assumption  holds when computing the required quantities, and this assumption will be addressed when we prove  Theorem~\ref{thm:main-formal}.

\begin{assumption}\label{assumption:runtime assumption}
    Let $L \triangleq L_{G+W}$ and $L_A = L_G + \tilde{L}$ be the Laplacian matrices of the graph $G+W$ and its subgraph after reweighting. Let $A = L^{-1/2} L_A L^{-1/2}$, $X = L^{-1/2} L_G L^{-1/2}$ and $\barr{M} = L^{-1/2} L_W L^{-1/2}$. We assume that
    \[
    A \prec (1-\eta)u \cdot I,
\]
 and 
    \[
        \PV\lp A - X - \ell \barr{M}\rp \PV \succeq  |\ell| \eta \cdot \PV\barr{M}\PV
    \]
    hold  for some $0 <  \eta < 1$.
\end{assumption}

\subsubsection{Approximating the projection matrix}\label{sec:approximated projection}

In comparison with previous algorithms for spectral sparsification~\cite{LS15:linearsparsifier,LS17}, one of the challenges for our problem is the need to compute the quantities that involve a projection matrix, which makes it difficult to  apply a nearly-linear time Laplacian solver directly. To address this issue, let us have a close look at $P_S$, which is the projection onto the bottom $k$ eigenspace of $X$. Following the previous notation, we set   $L \triangleq L_{G+W}$, and assume   that the eigenvalues of $X = L^{-1/2}L_GL^{-1/2}$ are $\lambda_1(X) \leq \lambda_2(X) \leq \dots \leq \lambda_n(X) \leq 1$ with corresponding eigenvectors $u_1, \dots, u_n$.    Since $P_S$ is the projection on the space spanned by $\{u_1 \dots u_k \}$, we can write  $P_S$ as 
    \[
        P_S = UU^{\rot},
    \]
    where 
    \[
        U \triangleq [ u_1, u_2, \dots u_k]
    \]
    is the $n \times k$ matrix whose $i$-th column is the vector $u_i$.  On the other hand, by definition we know that   the matrix $\barr{M} = I - X$ has eigenvalues $1 \geq \mu_1 \geq \dots \geq \mu_n$ with corresponding eigenvectors $u_1, \dots, u_n$, such that each $\mu_i = 1 - \lambda_i(X)$. Therefore, the columns of $U$ are also the top $k$ eigenvectors of  $\barr{M}$. To put it differently, \emph{$P_S$ is the projection on the top $k$ eigenspace of $\barr{M}$}.

Since it is computationally expensive to compute $P_S$, what we use in our analysis is the projection matrix $\PV$ which behaves similar to $P_S$.  We remark that, while $\PV = VV^{\rot}$ for some unitary matrix $V$ is used in our previous analysis, we do not need to compute the matrices $V$ or $\PV$ explicitly. Instead, we will show that it suffices to compute the matrix $\mathbb{V} \triangleq L^{-1/2}V$ in order to approximate our required quantities (1), (2). In this subsubsection, we discuss an efficient method for getting the matrix $\mathbb{V}$.
    The following result will be used in our analysis.

    \begin{theorem}[Restatement of Theorem~4.2 and Theorem~4.4,  \cite{pmlr-v70-allen-zhu17b}]\label{thm:proj guarantees}
        Let $A, B \in \mathbb{R}^{n \times n}$ be two symmetric matrices satisfying $B \succ 0$ and $-B \preceq A \preceq B$. Suppose the eigenvalues of $B^{-1/2}AB^{-1/2}$ are $1 \geq \mu_1 \geq \dots \geq \mu_n \geq 0$. For fixed $\varepsilon, p > 0$, we can find an $n \times k$ matrix 
        \[
            V = [v_1, \dots, v_k],
        \]
        such that with probability at least $1-p$ the following statements hold: 
        \begin{enumerate}
            \item $V^{\rot} V = I_{k\times k}$; 
            \item   we have $\left|v_i^{\rot} B^{-1/2}AB^{-1/2}v_i\right| \in \lsp (1-\varepsilon) \mu_i, \frac{\mu_i}{1-\varepsilon}\rsp $ for any $1\leq i\leq k$;
            \item $\max_{\substack{u\in \mathbb{R}^n \\ u^{\rot}V = 0}} \left|\frac{u^{\rot}B^{-1/2}AB^{-1/2}u}{u^{\rot}u}\right| \leq \frac{\mu_{k+1}}{1-\varepsilon}$.
        \end{enumerate}
        Moreover, we can obtain an $n \times k$ matrix 
        \[
            \mathbb{V} = B^{-1/2} V
        \]
        in   $\tilde{O}\left(\frac{k \mathrm{nnz}(B) + nk^2 + k\Upsilon}{\sqrt{\varepsilon}}\right)$ time, where $\Upsilon$ is the time needed to compute $(B^{-1}A)u$ for some vector $u$ with error $\delta$ such that $\log (1/\delta) = \tilde{O}(1)$. Here, the $\tilde{O}$ notation hides polylogarithmic factors with respect to $1/\varepsilon, 1/p, \kappa_B, n$.
    \end{theorem}
    
    We summarise the properties of our approximate projection $\PV$ in the following result.  
     \begin{theorem}\label{thm:our projection}
        We can compute a matrix $\mathbb{V} = L^{-1/2} V$, for some matrix $V$ such that, with constant probability, the following two properties hold:
        \begin{enumerate}
            \item $\PV = VV^{\rot}$ is a projection matrix  on a $k$-dimensional subspace $S'$ of $\mathbb{R}^n$;
            \item for any $u\in\mathbb{R}^n$ satisfying $u^{\rot}V = 0$ we have that  
                \[
                    \frac{u^{\rot}Xu}{u^{\rot}u} \geq \frac{\lambda_{k+1}(X)}{2} = \frac{\lambda^*}{2}.
                \]
        \end{enumerate}
        Moreover, the running time is  $t_{\ref{thm:our projection}} = \min \Big\{O(n^{\omega}), \tilde{O}\lp\frac{mk + nk^2}{\sqrt{\lambda^*}}\rp \Big\}$.  
    \end{theorem}
\begin{proof}
    
     The proof is by case distinction.  If $\frac{k}{\sqrt{\lambda^*}}= O(n^{\omega - 2})$, we apply Theorem~\ref{thm:proj guarantees} to get $\mathbb{V}$ and show that both listed properties are satisfied. However, in case $\frac{k}{\sqrt{\lambda^*}} = \Omega(n^{\omega - 2})$, it is more efficient to compute $\mathbb{V}$ directly from the spectral decomposition of $L^{-1}L_G$. 
     
    \textbf{Case 1:} $\frac{k}{\sqrt{\lambda^*}}= O\left(n^{\omega-2}\right)$
    
    We will apply Theorem~\ref{thm:proj guarantees} with $A = L_W$, $B = L$ and $\varepsilon = \frac{\lambda^*}{2 - \lambda^*}$. Since  
    \[
        B^{-1/2}AB^{-1/2} = L^{-1/2}L_WL^{-1/2} = \barr{M},
    \]
    we have  $\mu_i = 1 - \lambda_i(X)$, for all $i$.         The guarantees in Theorem~\ref{thm:proj guarantees} ensure that, with constant probability, the matrix $\PV \triangleq VV^{\rot}$ is a projection matrix. Moreover, it holds for  any   $u\in \mathbb{R}^n$ with $u^{\rot} V = 0$   that 
    \[
        \frac{u^{\rot} \barr{M}u}{u^{\rot}u} \leq \frac{1 - \lambda_{k+1}(X)}{1-\varepsilon} = 1 - \frac{\lambda^*}{2},
    \]
    where the last equality comes from the choice of our $\varepsilon$. Using the fact that $\barr{M} = I - X$ and rearranging the above inequality, we get that 
    \[
        \frac{u^{\rot}Xu}{u^{\rot}u} \geq \frac{\lambda^*}{2}.
    \]
    Finally, the running time $\tilde{O}\lp\frac{mk + nk^2}{\sqrt{\lambda^*}}\rp$ follows from Theorem~\ref{thm:proj guarantees} for the choice of our parameters and the fact that we have a nearly-linear time solver for Laplacian systems.
    
    \textbf{Case 2:} $\frac{k}{\sqrt{\lambda^*}} = \Omega\left(n^{\omega - 2}\right)$
    
    Recall that $P_S = UU^{\rot}$ is the projection onto the bottom $k$ eigenspace of $X = L^{-1/2}L_GL^{-1/2}$. For $V = U$, it is clear that $\PV = P_S$ is a projection matrix.
    Moreover, it's easy to see that $\mathbb{V} = L^{-1/2} U = [L^{-1/2} u_1, \dots, L^{-1/2}u_k]$ has columns the bottom $k$ eigenvectors of the matrix $L^{-1}L_G$. Hence, if $u^{\rot} V = 0$, then it must be the case that $ u \in S^{\perp}$, and therefore 
    \[
        \frac{u^{\rot} X u}{u^{\rot}u} \geq \lambda_{k+1}(X) > \frac{\lambda^*}{2}.
    \]
    This proves the second condition. For the running time, we first compute the spectral decomposition of $L$ and then use that to get the spectral decomposition of $L^{-1}L_G$. From this, we can get the bottom $k$ eigenvectors which form $\mathbb{V}$. The running time is $O(n^{\omega})$ \cite{Demmel2007}. 
        
\end{proof}

\begin{lemma}\label{lem:eigendecomposition}
    Under Assumption~\ref{assumption:runtime assumption}, let 
    \[
        \mathbb{V}^{\rot} \lp \tilde{L} - \ell L_W\rp \mathbb{V} = \sum_{i=1}^k \lambda_i f_i f_i^{\rot}
    \]
    be the spectral decomposition of the $k \times k$ matrix $\mathbb{V}^{\rot} \lp \tilde{L} - \ell L_W\rp \mathbb{V}$. Then, we can get $\{\lambda_i\}_{i=1}^k$ and $\left\{\mathbb{V}f_i\right\}_{i=1}^k$ in time
    $  t_{\ref{lem:eigendecomposition}} = O(\min\{n^{\omega}, mk + nk^2 + k^{\omega}\} )$.
\end{lemma}
\begin{proof}

    Let $        F = [f_1, \dots, f_k]$
     be the matrix whose columns are the eigenvectors $f_i$.      Then, the proof can be summarised in the following three  steps:
    \begin{enumerate}
        \item Perform matrix multiplication to get $\mathbb{V}^{\rot}(\tilde{L} - \ell L_W)\mathbb{V}$;
        \item Perform the spectral decomposition of $\mathbb{V}^{\rot}(\tilde{L} - \ell L_W)\mathbb{V}$ to get $\lambda_i$ and $f_i$, for all $1 \leq i \leq k$;
        \item Perform matrix multiplication to get $\mathbb{V} F$.
    \end{enumerate} 
    The set $\{\lambda_i\}_{i=1}^k$ is obtained in the second step, while the set $\{\mathbb{V}f_i\}_{i=1}^k$ consists of the columns of the matrix $\mathbb{V}F$, which we will get at the end of the third step.
        It is easy to see that the three steps can be computed in time   $O\lp \min \{n^{\omega}, mk + nk^2 \} \rp$, $O(\min \{ n^{\omega}, k^{\omega}\})$, and  $O(\min \{n^{\omega}, nk^2\})$ respectively. Therefore,  the total running time is $O(\min\{n^{\omega}, mk + nk^2 + k^{\omega}\} )$.
\end{proof}

 \subsubsection{Approximating the quantities involving the lower barrier value $\ell$}\label{sec:approx lower barrier}

We first  present efficient algorithms for approximately computing all the quantities that involve the lower barrier $\ell$, i.e., the quantities (1) and (2).
Our results are summarised in Lemma~\ref{lem:runtime l1} and Lemma~\ref{lem:runtime l2} below.
   \begin{lemma}\label{lem:runtime l1}
        Under the Assumption~\ref{assumption:runtime assumption}, we can compute numbers $\{r_i\}_{i=1}^m$ such that  
        \[
            (1-\varepsilon) r_i \leq v_i^{\rot}Z \lp \PV (B - \ell I) \PV\rp^{\dag} Z v_i \leq (1+\varepsilon) r_i,
        \] 
        for all $\{v_i\}$ in time $ t_{\ref{lem:runtime l1}} = \tilde{O}\lp \lp\min \{n^{\omega}, mk + nk^2 + k^{\omega}\}\rp/\varepsilon^2\rp$.
    \end{lemma}
Before proving Lemma~\ref{lem:runtime l1}, we need the following technical result.
\begin{lemma}\label{lem:from B to A in S} It holds that 
    $\lRes{\ell}{B} = \PV \lp \PV \lp A - X - \ell \barr{M} \rp \PV \rp^{\dag} \PV $.
\end{lemma}

\begin{proof}
    By the definition of $B$, we have that
    \[
        \lRes{\ell}{B} = \lRes{\ell}{Z(A-X)Z} = Z \lp Z \lp \PV (A-X) \PV - \ell Z^{\dag2} \rp Z\rp^{\dag} Z,
    \]
    where the last equality follows by the fact that $Z\cdot \PV  = \PV \cdot Z = Z$.
    We define $C = \PV (A-X) \PV - \ell Z^{\dag2}$, which gives us that   $\PV C = C \PV = C$. Also, since 
    \[
        Z^{\dag} C^{\dag} Z^{\dag} \cdot ZCZ 
        = Z^{\dag} C^{\dag} \cdot \PV \cdot CZ
        = Z^{\dag} C^{\dag} \cdot CZ
        = Z^{\dag} \cdot \PV \cdot Z
        = \PV,
    \]
    we have that   $(ZCZ)^{\dag} = Z^{\dag} C^{\dag} Z^{\dag}$.
    Therefore, it holds that
    \begin{align*}
       \lRes{\ell}{B} & = Z \lp ZCZ \rp^{\dag} Z  = \PV \lp \PV (A-X) \PV - \ell Z^{\dag2}\rp^{\dag} \PV \\
       &
        = \PV \lp \PV \lp A-X - \ell \barr{M} \rp \PV \rp^{\dag} \PV,
    \end{align*}
    where the last equality follows from the fact that $Z = \lp \PV \barr{M} \PV\rp^{\dag/2}.$
\end{proof}

\begin{proof}[Proof of Lemma~\ref{lem:runtime l1}] 

    First of all, we observe that  
        \begin{align}
           Z \lp \PV (B - \ell I) \PV\rp^{\dag} Z 
           &= \PV \lp \PV (A - X - \ell \barr{M}) \PV\rp^{\dag} \PV \label{eq:from B to A}\\
           &= \PV \lp VV^{\rot} L^{-1/2}(\tilde{L} - \ell L_W)L^{-1/2} VV^{\rot} \rp^{\dag} \PV \nonumber\\
           &= \PV \lp V \mathbb{V}^{\rot}(\tilde{L} - \ell L_W)\mathbb{V} V^{\rot} \rp^{\dag} \PV \nonumber\\
           &= \PV \lp V \lp \sum_{j=1}^k \lambda_i f_j f_j^{\rot}\rp V^{\rot} \rp^{\dag} \PV \nonumber\\
           &= \PV \lp \sum_{j=1}^k \lambda_i \lp Vf_j\rp \lp Vf_j\rp^{\rot} \rp^{\dag} \PV \nonumber\\
           &= \PV \lp \sum_{j=1}^k \lambda_i^{-1} \lp Vf_j\rp \lp Vf_j\rp^{\rot} \rp \PV \label{eq:spec decomp2},
        \end{align}
        where \eqref{eq:from B to A} follows by Lemma~\ref{lem:from B to A in S},    and  (\ref{eq:spec decomp2}) follows by  Lemma~\ref{lem:spec decomp preserved under unitary}.  
        Therefore, for any 
        $v_i$ with the corresponding $b_e$ we have  that 
        \begin{align*}
            v_i^{\rot} Z \lp \PV (B - \ell I) \PV\rp^{\dag} Z v_i
            &= v_i^{\rot} \PV \lp \sum_{j=1}^k \lambda_j^{-1} \lp Vf_j\rp \lp Vf_j\rp^{\rot} \rp \PV v_i \\
            &= v_i^{\rot} VV^{\rot} \cdot V \lp \sum_{j=1}^k \lambda_j^{-1} f_j f_j^{\rot}\rp V^{\rot} \cdot VV^{\rot} v_i\\ 
            &= b_e^{\rot}L^{-1/2}V \lp \sum_{j=1}^k \lambda_j^{-1} f_j f_j^{\rot}\rp V^{\rot} L^{-1/2} b_e\\
            &= b_e^{\rot} \mathbb{V} \lp \sum_{j=1}^k \lambda_j^{-1} f_j f_j^{\rot}\rp \mathbb{V}^{\rot} b_e\\
            &= \sum_{j=1}^k \left\langle b_e, \lambda_j^{-1/2}\mathbb{V} f_j\right\rangle^2.
        \end{align*}
        
        Next, we use Lemma~\ref{lem:eigendecomposition} to get $\{\lambda_j\}_{j=1}^k$ and   $\{ \mathbb{V}f_j\}_{j=1}^k$ in time $t_{\ref{lem:eigendecomposition}} = O(\min\{n^{\omega}, mk + nk^2 + k^{\omega}\})$. To compute $\sum_{j=1}^k \left\langle b_e, \lambda_j^{-1/2}\mathbb{V} f_j\right\rangle^2$ we make the following case distinction based on the value of $k$.
      
      \underline{\textbf{Case 1: $k = O(\log n/\varepsilon^2)$}}. 
       We set
       \[
            r_i = \sum_{j=1}^k \left\langle b_e, \lambda_j^{-1/2}\mathbb{V} f_j\right\rangle^2
       \]
       and compute $r_i$ by explicitly computing the right hand side. This will take time  $O(mk)$ for all edges $e$.
       
      \underline{\textbf{Case 2: $k = \Omega(\log n/\varepsilon^2)$}}. 
       Let
       \[
            J = [\mathbb{V}\lambda_1^{-1/2}f_1, \dots, \mathbb{V}\lambda_k^{-1/2}f_k ],
       \]
       and fix one edge $e= \{x, y\}$. We can view $\sum_{j=1}^k \langle b_e, \mathbb{V}\lambda_j^{-1/2} f_j\rangle^2$ as the squared norm of the difference between the $x$-row  and the $y$-row  of the matrix $J$. We invoke the Johnson-Lindenstrauss Lemma to reduce the matrix $J \in \mathbb{R}^{n \times k}$ to a matrix $JR \in \mathbb{R}^{n \times O(\log n/\varepsilon^2)}$ such that 
       \[
            \sum_{j=1}^k \langle b_e, \mathbb{V}\lambda_j^{-1/2} f_j\rangle^2 \approx_{\varepsilon} r_i,
       \]
        where $r_i$ is the squared norm of the difference between the $x$-row  and the $y$-row  of the matrix $JR$. The running time to compute $JR$ is $O(nk \log n/\varepsilon^2)$, and given $JR$ the time to compute $r_i$ for all edges $e$ is $O(m \log n/\varepsilon^2)$. Therefore the total running time for Case 2 is  $O\lp (m + nk) \log n/\varepsilon^2\rp$.
    
        Thus, we can upper bound the running time for Case 1 and Case 2 by 
        \[
             t_{\mathrm{cases}} \triangleq O\lp (m + nk) \log n/\varepsilon^2\rp,
        \]
        and the   overall runtime is given by 
        \begin{align*}
          t_{\ref{lem:eigendecomposition}} + t_{\mathrm{cases}}
            &= O(\min\{n^{\omega}, mk + nk^2 + k^{\omega}\}) +
            O\lp (m + nk) \log n/\varepsilon^2\rp \\
            &= \tilde{O} \lp (\min\{n^{\omega}, mk + nk^2 + k^{\omega}\})/\varepsilon^2\rp. \qedhere
        \end{align*}
             
\end{proof}

      \begin{lemma}\label{lem:runtime l2}
        Under   Assumption~\ref{assumption:runtime assumption}, we can compute a number $\alpha$ such that
        \[
            (1-\varepsilon)\alpha \leq \lambdaMax\lp \PV (B - \ell I) \PV\rp^{\dag} \leq (1+\varepsilon) \alpha,
        \]
        in time 
        $
            t_{\ref{lem:runtime l2}} = \tilde{O}\lp (\min\{n^{\omega}, mk + nk^2 + k^{\omega}\})/\varepsilon^3\rp.
        $
    \end{lemma}
\begin{proof}

    By Lemma~\ref{lem:from B to A in S}, we have that  
    \begin{align*}
    \lambdaMax(\PV (B - \ell I) \PV)^{\dag}  & = \lambdaMax\left(Z^{\dag}\lp \PV(A - X -\ell \barr{M})\PV \rp^{\dag}Z^{\dag}\right) \\
    & = \lambdaMax\left(\lp \PV(A - X - \ell \barr{M})\PV \rp^{\dag}\cdot \PV \barr{M}\PV\right),
    \end{align*}
    where the last equality uses the fact that the eigenvalues are preserved through circular permutations. Therefore, we have that 
    \begin{align*}
        \lefteqn{\lambdaMax\left(\lp \PV \lp A - X - \ell \barr{M}\rp \PV\rp^{\dag} \cdot \PV \barr{M}\PV\right)}\\
        &=\lambdaMax\left(\lp V \mathbb{V}^{\rot}\lp \tilde{L} - \ell L_W\rp \mathbb{V} V^{\rot}  \rp^{\dag} \PV \barr{M} \PV\right)\\
        &= \lambdaMax \lp \lp \sum_{i=1}^k \lambda_i^{-1} (Vf_i) (Vf_i)^{\rot} \rp VV^{\rot} L^{-1/2} L_W L^{-1/2} VV^{\rot}\rp \\
        &= \lambdaMax \lp \lp \sum_{i=1}^k \lambda_i^{-1} f_i f_i^{\rot}\rp \mathbb{V}^{\rot} L_W \mathbb{V}\rp \\
        &= \lambdaMax \lp \lp \sum_{i=1}^k \lambda_i^{-1} \lp \mathbb{V}f_i\rp \lp \mathbb{V}f_i\rp^{\rot}\rp L_W\rp \\
        &= \lambdaMax \lp J L_W\rp,
    \end{align*}  
    where the matrix $J$ is defined as  
    \[
    J \triangleq  \sum_{i=1}^k (1/\lambda_i) \lp \mathbb{V}f_i\rp \lp \mathbb{V}f_i\rp^{\rot}.
    \]
    Now we observe that for any $t\in\mathbb{N}$ it holds that
    \[
        \lambdaMax(JL_W) \leq \lp \Tr (JL_W)^{2t} \rp^{1/2t} \leq n^{1/2t} \lambdaMax(J L_W). 
    \]
    Thus, for $t = \Theta(\log n / \varepsilon)$, we have that $$ \lp \Tr (JL_W)^{2t}\rp^{1/2t} \approx_{\varepsilon/2} \lambdaMax(JL_W).$$ 
    Based on this, in order to approximate  $\lambdaMax\left(JL_W\right)$ it suffices to approximate  $\Tr \left(JL_W\right)^{2t}$ instead. To that extent, we have that 
    \begin{align*}
        \Tr \lsp (JL_W)^{2t} \rsp 
        &= \tr \lsp (JL_W)^{t-1} JL_W^{1/2} L_W^{1/2}J (L_W J)^{t-1} L_W\rsp \\
        &= \sum_{e \in W} \Tr \lsp  (JL_W)^{t-1} JL_W^{1/2} L_W^{1/2}J (L_W J)^{t-1} \cdot b_e b_e^{\rot} \rsp\\
        &= \sum_{e \in W} \left\| L_W^{1/2} J (L_W J)^{t-1} b_e\right\|^2\\
        &= \sum_{e \in W} \left\| B J (L_W J)^{t-1} b_e\right\|^2,
    \end{align*}
    where we used that $L_W = BB^{\rot}$ for some  incidence matrix $B$. Now we apply the standard technique of approximating the distances via the Johnson-Lindenstrauss Lemma. To that extent, let $R$ be the $O(\log n / \varepsilon^2) \times n$ random projection matrix, and it holds    for every edge $e \in E_W$ that
    \[
        \| R BJ (L_W J)^{t-1} b_e \|^2 \approx_{\varepsilon/10}  \| BJ (L_W J)^{t-1} b_e \|^2. 
    \]
    We focus on computing the matrix $RBJ(L_WJ)^{t-1}$. To that extent, let $r^{\rot}$ be some row of $R$ and using the symmetry of all matrices, it suffices to compute $(J L_W)^{t-1} J B r$. This can be done sequentially by performing the following steps:
    \begin{enumerate}
        \item Perform the multiplication $B u$ for some vector $u$;
        \item Perform the multiplication $L_W u$ for some vector $u$;
        \item Perform the multiplication $J u$ for some vector $u$.
    \end{enumerate}
    We notice that computing $B u$ and $L_W u$ takes  $O(m +n)$ time. For the last step, recall that   
    \[
        J = \sum_{i=1}^k \lp\lambda_i^{-1/2}\mathbb{V}f_i\rp \lp\lambda_i^{-1/2}\mathbb{V}f_i\rp^{\rot},
    \]
   and by Lemma~\ref{lem:eigendecomposition} we can get $\{\lambda_i\}_{i=1}^k$ and  $\{\mathbb{V}f_i\}_{i=1}^k$ in    $t_{\ref{lem:eigendecomposition}} = O(\min\{n^{\omega}, mk + nk^2 + k^{\omega}\})$ time. 
   
   Then,  computing $G u$ for any vector $u$ can be done in $O(nk)$ time. Since $t = \Theta(\log n/\varepsilon)$, the total time needed to compute $(J L_W)^{t-1} J B r$ is $O((m + nk)\log n/\varepsilon)$. Since there are 
    $O(\log n/\varepsilon^2)$ rows of $R$,
     the total time needed to 
      compute   $RBJ(L_WJ)^{t-1}$ is $O((m + nk)\log^2 n/\varepsilon^3)$. Therefore,  the total running time for computing $\| R BJ (L_W J)^{t-1} b_e \|^2$ for all edges $e$ is  $
       O \lp  (m + nk)\log^2n/\varepsilon^3 \rp$, and the  overall runtime is 
    \[
        t_{\ref{lem:eigendecomposition}} +  O \lp  (m + nk)\log^2n/\varepsilon^3 \rp 
        = \tilde{O}\lp (\min\{n^{\omega}, mk + nk^2 + k^{\omega}\})/\varepsilon^3\rp . \qedhere
    \]
\end{proof}

\subsubsection{Approximating the quantities 
involving the upper barrier value $u$}

Next we  present efficient algorithms for computing all the quantities that involve the upper barrier value $u$, i.e., the quantities (3), (4), (5)  and (6).   We remark that the discussions here   follow the analysis of  \cite{zhu15,LS15:linearsparsifier}, and our main point here  is to show that their techniques can be adapted in our setting which involves $\overline{M}$ in the computation.

We will first show that $\Tr \left[ (uI-A)^{-1}\barr{M}\right]$ and the values $v_i^{\rot}(uI-A)^{-1}v_i$ for all the $v_i$ can be approximately computed in almost-linear time.

\begin{lemma}[\cite{LS15:linearsparsifier}] \label{lem: approximate matrices}
  Under Assumption~\ref{assumption:runtime assumption}, the following statement holds:
  we can construct a matrix $S_u$ such that
  \[
    S_u\approx_{\eps/10} (uI-A)^{-1/2},
  \]
  and $S_u=p(A)$ for a polynomial $p$ of degree $O\left(\frac{\log(1/\eps\eta)}{\eta}\right)$.
\end{lemma}

\begin{lemma}\label{lem:aptrace}
Under  Assumption~\ref{assumption:runtime assumption},
there is an algorithm that  computes  $\{r_i\}_{i=1}^m$ in $t_{\ref{lem:aptrace}} = \tilde{O}\lp\frac{m}{\eta\eps^2}\rp$ time such that
\[
  (1-\varepsilon) \cdot r_i\leq  v_i^{\rot}(uI-A)^{-1}v_i \leq (1+\varepsilon) \cdot r_i.
\]
Moreover, it holds that 
\[
  (1-\varepsilon) \cdot \sum_{i=1}^m r_i \leq \Tr \lsp (uI-A)^{-1}\barr{M}\rsp \leq (1+\varepsilon) \cdot \sum_{i=1}^m r_i.
\]
\end{lemma}
\begin{proof}
The first statement is from Lemma~4.9 of \cite{LS15:linearsparsifier}. 
The second statement follows by the first statement and the fact that
\[
\sum_{i=1}^m v_i^{\rot} \left(u I-A \right)^{-1} v_i = \sum_{i=1}^m  \tr\left[\left(u I-A \right)^{-1} v_i v_i^{\rot}\right] = \tr\left[\left(u I-A \right)^{-1} \overline{M}\right].\qedhere
\]
\end{proof}

The next lemma shows that both of $\lambdaMax \lsp (uI-A)^{-1}\barr{M}\rsp $
 and $\lambdaMin \lsp (uI-A)^{-1}\barr{M}\rsp $ can be approximately computed in almost-linear time.

\begin{lemma}\label{lem:running time lambda max 1}
    Under   Assumption~\ref{assumption:runtime assumption}, there is an algorithm that computes   values $\alpha_1$, $\alpha_2$ in $t_{\ref{lem:running time lambda max 1}} = \tilde{O}\lp\frac{m}{\eta\varepsilon^3}\rp$  time such that
    \[
        (1-\varepsilon)\cdot \alpha_1 \leq \lambdaMax \lsp (uI-A)^{-1}\barr{M}\rsp \leq (1+\varepsilon)\cdot \alpha_1,
    \]
    and 
    \[
        (1-\varepsilon)\cdot \alpha_2 \leq \lambdaMin \lsp(uI-A)^{-1}\barr{M}\rsp \leq (1+\varepsilon)\cdot \alpha_2.
    \]
\end{lemma} 

\begin{proof}
    By Lemma~\ref{lem: approximate matrices}, we have a matrix $S_u$ such that  $S_u \approx_{\varepsilon/10} (uI - A)^{-1/2}$. This implies that 
    \[
        \barr{M}^{1/2}S_u^2\barr{M}^{1/2} \approx_{3\varepsilon/10} \barr{M}^{1/2}(uI-A)^{-1}\barr{M}^{1/2},
    \]
    and hence
    \[
        \lambdaMax \left( S_u^2 \barr{M}\right) \approx_{3\varepsilon/10} \lambdaMax \lsp (uI-A)^{-1}\barr{M} \rsp.
    \]
    Therefore, it suffices to approximate $\lambdaMax \lp S_u^2 \barr{M} \rp$. To achieve this, notice that
    \[
        \lambdaMax \lp S_u^2 \barr{M}\rp 
        \leq \lp \Tr \left(S_u^2\barr{M}\right)^{2t+1}\rp^{1/(2t+1)} 
        \leq n^{1/(2t+1)}\lambdaMax\left(S_u^2\barr{M}\right),
    \]
    by choosing $t = \Theta(\log n /\varepsilon)$ we have that 
    \[
        \lp \Tr \left(S_u^2\barr{M}\right)^{2t+1}\rp^{1/(2t+1)} \approx_{\varepsilon/2}\lambdaMax\left(S_u^2\barr{M}\right).
    \]
    Based on this, in order to approximate  $\lambdaMax\left(S_u^2\barr{M}\right)$ it suffices for us to approximate  $\Tr \left(S_u^2\barr{M}\right)^{2t+1}$ instead,
    which is achieved by 
    exploiting the structure of $S_u$ and $\barr{M}$. Our proof technique is similar with the proof of Lemma~G.3 of  \cite{zhu15}, and we present the proof here for completeness.  
    By definition, we have that 
    \[
        \barr{M} = L_{G+W}^{-1/2}L_WL_{G+W}^{-1/2} = L_{G+W}^{-1/2}\lp\sum_{e \in W}b_eb_e^{\rot} \rp L_{G+W}^{-1/2},
    \]
    For simplicity, we write $L\triangleq L_{G+W}$ in the remaining part of the proof. 
    We have that 
    \begin{align*}
        \Tr \lp S_u^2 \barr{M} \rp^{2t+1} 
        &= \Tr \lsp p^2(L^{-1/2} L_A L^{-1/2}) \cdot L^{-1/2} L_W L^{-1/2}\rsp^{2t+1} \\
        &= \Tr \lsp L^{1/2}p^2(L^{-1}L_A) L^{-1} L_W L^{-1/2}\rsp^{2t+1}\\
        &= \Tr \lsp p^2(L^{-1}L_A) L^{-1} L_W \rsp^{2t+1} \\
        &= \Tr \lsp \lp p^2(L^{-1}L_A) L^{-1} L_W \rp^{2t} \cdot p^2(L^{-1}L_A) L^{-1} L_W \rsp \\
        &= \Tr \lsp \lp p^2(L^{-1}L_A) L^{-1} L_W \rp^{t} \cdot \lp p^2(L^{-1}L_A) L^{-1} L_W \rp^{t} \cdot p^2(L^{-1}L_A) L^{-1} L_W \rsp \\
        &= \Tr \lsp \lp p^2(L^{-1}L_A) L^{-1} L_W \rp^{t} \cdot p^2(L^{-1}L_A) L^{-1} \cdot \lp L_W p^2(L^{-1}L_A) L^{-1} \rp^{t}L_W\rsp \\
        &= \Tr \lsp \lp p^2(L^{-1}L_A) L^{-1} L_W \rp^{t} \cdot p^2(L^{-1}L_A) L^{-1} \cdot \lp L_W L^{-1} p^2(L^{-1}L_A)   \rp^{t}L_W\rsp.
    \end{align*}
    Let   $D\triangleq  \lp L_W L^{-1} p^2(L^{-1}L_A)   \rp^{t}$, and we rewrite the equality  above as 
    \begin{align*}
        \Tr \lp S_u^2 \barr{M} \rp^{2t+1}
        &= \Tr \lsp D^{\rot}   p^2(L^{-1}L_A) L^{-1}   D L_W\rsp \\
        &= \Tr \lsp D^{\rot}   L^{-1/2}p^2(L^{-1/2}L_AL^{-1/2})L^{-1/2}  D L_W\rsp\\
        &= \sum_{e \in W} \Tr \lsp D^{\rot}   L^{-1/2}p^2(L^{-1/2}L_AL^{-1/2})L^{-1/2}  D b_eb_e^{\rot}\rsp\\
        &= \sum_{e \in W} \left\| p(L^{-1/2}L_AL^{-1/2})L^{-1/2}   D   b_e\right\|^2 \\
        &= \sum_{e \in W} \left\| L^{-1/2}p(L_A L^{-1})   D  b_e\right\|^2\\
        &= \sum_{e \in W} \left\| L^{1/2}L^{-1}p(L_A L^{-1})   D \ b_e\right\|^2.
    \end{align*}
    Now we use the fact that $L = B^{\rot}B$ for the edge-incidence matrix $B$ (assuming without loss of generality that the graph is unweighted), and we also invoke the Johnson-Lindenstrauss lemma to obtain a matrix $Q\in\mathbb{R}^{O(\log n/\eps^2)\times m}$ such that the above quantity can be approximated by 
        \[
          \Tr \lp S_u^2 \barr{M} \rp^{2t+1} \approx_{\varepsilon/10}
          \sum_{e \in W} \left\| QBL^{-1}p(L_A L^{-1})   D   b_e\right\|^2,
        \]
    which can be approximately computed in    $\tilde{O}\lp \frac{m}{\eta\varepsilon^3}\rp$ time  using a  nearly-linear time Laplacian solver. We refer the reader to Lemma $G.3$ in the appendix of \cite{zhu15} for a more detailed discussion on the fast computation of the above quantity.

        Now we prove the second statement. Notice that  
    \[
        \lambdaMin \lsp (uI-A)^{-1}\barr{M}\rsp = \frac{1}{\lambdaMax \lsp (uI-A)\barr{M}^{-1}\rsp},
    \]
    and we can rewrite $\lambdaMax \lsp (uI-A)\barr{M}^{-1}\rsp$ as   
    \begin{align*}
      \lambdaMax \lsp (uI-A)\barr{M}^{-1}\rsp 
      &= \lambdaMax \lsp \lp uI - L^{-1/2}L_AL^{-1/2}\rp \cdot \lp L^{-1/2}L_WL^{-1/2}\rp^{-1} \rsp \\
      &= \lambdaMax \lsp \lp L^{-1/2}(uL - L_A)L^{-1/2} \rp \cdot \lp L^{1/2}L_W^{-1}L^{1/2}  \rp\rsp \\
      &= \lambdaMax \lsp (uL - L_A)\cdot L_W^{-1}\rsp,
    \end{align*}
    which can be approximately computed  in $\tilde{O}\left(\frac{m}{\eta \varepsilon^3}\right)$ time in a similar way as before.
\end{proof}

\subsection{Approximation guarantee}\label{sec:approx guarantee}

Now we study the  approximation ratio of the algorithm's returned sparsifier.  
The following statement will be used in our analysis.

\begin{lemma}\label{lem:sparsification guarantee}
   
    Suppose that the algorithm returns the matrix $M = A_{K}$ such that 
    $
      \lambdaMax(A_K) \leq \theta_{\mathrm{max}}$ and $ \lambdaMin\left(B_K\big|_{S'}\right) \geq \theta_{\mathrm{min}}$. Then it holds that         \[
      \lambdaMin\left(A_K\right) \geq \frac{\theta_{\mathrm{min}} \lambda^{*}/2}{\lp \lp\lambda^{*}/2\rp^{1/2} + \theta_{\mathrm{min}}^{1/2} + \theta_{\mathrm{max}}^{1/2}\rp^2},
    \]
    where $\lambda^*=\lambda_{k+1}(X)$.

\end{lemma} 
\begin{proof}

    The proof is a direct adaptation of the proof of Theorem~3.3 in \cite{KMST10:subsparsification}. 
    Let $v$ be an arbitrary unit vector, and we write   $v = v_{\mathcal{V}} + v_{\mathcal{V}^{\perp}}$ such that $v_{\mathcal{V}} \in S'$ and $v_{\mathcal{V}^{\perp}} \perp S'$. We will give two lower bounds for $v^{\rot}A_{K}v$, one increasing and one decreasing with $\norm{v_{\mathcal{V}^{\perp}}}$. The statement will follow when we equalise the two bounds. 
    
    First of all,  by the Lemma's preconditions we have that 
    $
        A_K \preceq \theta_{\mathrm{max}} I$ and $ \PV B_K \PV \succeq \theta_{\mathrm{min}} \PV$.
    To derive the first bound, we   have that 
    \begin{align*}
        v_{\mathcal{V}}^{\rot} A_K v_{\mathcal{V}} 
        &= v_{\mathcal{V}}^{\rot} \lp \PV X \PV + Z^{\dag/2}B_K Z^{\dag/2} \rp v_{\mathcal{V}}\\
        &= v_{\mathcal{V}}^{\rot} \lp \PV X \PV + \lp \PV \barr{M} \PV \rp^{1/2} B_K \lp \PV \barr{M} \PV \rp^{1/2} \rp v_{\mathcal{V}}\\
        &\geq v_{\mathcal{V}}^{\rot} \lp \PV X \PV + \lp \PV \barr{M} \PV \rp^{1/2} \theta_{\mathrm{min}} \cdot \PV \lp \PV \barr{M} \PV \rp^{1/2} \rp v_{\mathcal{V}}\\
        &= v_{\mathcal{V}}^{\rot} \lp \PV X \PV + \theta_{\mathrm{min}} \PV \barr{M} \PV \rp v_{\mathcal{V}}\\
        &= v_{\mathcal{V}}^{\rot} \lp \theta_{\mathrm{min}} \PV + (1-\theta_{\mathrm{min}})\PV X \PV \rp v_{\mathcal{V}}\\
        &\geq \theta_{\mathrm{min}} \norm{v_{\mathcal{V}}}^2.
    \end{align*}
    On the other hand, we have that
    \[
        v_{\mathcal{V}^{\perp}}^{\rot} A_K v_{\mathcal{V}^{\perp}} \leq \theta_{\mathrm{max}}\norm{v_{\mathcal{V}^{\perp}}}^2.
    \]
    Hence,   by the triangle inequality we have that 
    \[
        \lp v^{\rot} A_K v\rp^{1/2} 
        \geq \theta_{\mathrm{min}}^{1/2}\norm{v_{\mathcal{V}}} - \theta_{\mathrm{max}}^{1/2} \norm{v_{\mathcal{V}^{\perp}}}
        \geq \theta_{\mathrm{min}}^{1/2} - \lp \theta_{\mathrm{max}}^{1/2} + \theta_{\mathrm{min}}^{1/2}\rp \norm{v_{\mathcal{V}^{\perp}}}.
    \]
    
    To derive the second bound, we have by  Theorem~\ref{thm:our projection} that
    \[
        \lp v^{\rot} A_K v\rp^{1/2} 
        \geq \lp v^{\rot} X v\rp^{1/2} 
        \geq \lp v_{\mathcal{V}^{\perp}}^{\rot} X v_{\mathcal{V}^{\perp}}\rp^{1/2} \geq \lp \frac{\lambda^*}{2} \rp^{1/2} \norm{v_{\mathcal{V}^{\perp}}}.
    \]
    Equalising our two lower bounds gives the desired result.  
\end{proof}

The lemma below summaries the spectral properties of the resulting sparsifier, which is essentially the same as the one from \cite{KMST10:subsparsification} up to a constant factor.   
\begin{lemma} \label{lem:approximation guarantee}  
The condition number of the returned matrix $A_{\tau}$ after $\tau$ iterations is at most $ 1+ O(\varepsilon)\cdot \max\{1, T/k \}$.  Moreover, it holds that 
   \[
        \lambdaMin\left(A_{\tau} \right) \geq c\cdot (1-O(\varepsilon))\cdot\lambda^*\min\{1, k/T\},
    \]
    for some  constant $c$. 
\end{lemma}
\begin{proof}

Notice that it holds for any iteration $j$ that
\[
\frac{\overline{\delta}_{u,j} - \overline{\delta}_{\ell,j}}{\overline{\delta}_{u,j}} = \frac{6\varepsilon}{1+3\varepsilon},
\]
which implies that 
\[
\overline{\delta}_{u,j} = \frac{1+ 3\varepsilon}{ 6\varepsilon} \left( \overline{\delta}_{u,j} - \overline{\delta}_{\ell,j} \right)  \geq \frac{1}{ 6\varepsilon} \left( \overline{\delta}_{u,j} - \overline{\delta}_{\ell,j} \right).
\]
Let $u_{\tau}$ and $\ell_{\tau}$ be the barrier values when the algorithm terminates, and our goal is to show that 
\[
\frac{u_{\tau}}{\ell_{\tau}} = \left( 1 -\frac{u_{\tau} - \ell_{\tau}}{ u_{\tau}} \right)^{-1} = 1+ O(\varepsilon)\cdot \max\{1, T/k \},
\]
which suffices to prove that 
\[
 \frac{u_{\tau} - \ell_{\tau}}{ u_{\tau}} = O(\varepsilon)\cdot \max\{1, T/k \}.
\]
By definition, we know that
\begin{align*}
 \frac{u_{\tau} - \ell_{\tau}}{ u_{\tau}} & \leq \frac{u_0 -\ell_0 + \alpha}{ u_0 + (6\varepsilon)^{-1} \alpha} \leq \frac{3 + 6k/\Lambda}{2 + (6\varepsilon)^{-1} 4k/\Lambda} \leq 1+ O(\varepsilon)\cdot \max\{1, T/k\},
\end{align*}
where the last inequality holds by the definition of $\Lambda$. Now we set $\theta_{\min} = \ell_{\tau}$, $\theta_{\max}= u_{\tau}$, and apply  Lemma~\ref{lem:sparsification guarantee}. This gives us that  \[
        \lambdaMin(A_{\tau}) \geq c\cdot (1-O(\varepsilon))\cdot\lambda^*\min\{1, k/T\},
    \]
    for some  constant $c$. This proves the statement.
\end{proof}
 
Finally, the lemma below analyses   $\sum_{i=1}^m w_i cost_i$.

\begin{lemma}\label{lem:cost}
It holds with constant probability that 
$
          \sum_{i=1}^m c_i \cdot cost_i  = O (1/\varepsilon^2)\cdot \min\{ 1, k/T\}
             $.
\end{lemma}
\begin{proof}
Without loss of generality, let $v_i$ be the vector sampled in iteration $j$. Then the contribution of $v_i$ towards the total cost function in iteration $j$, denoted by $\sigma_{i,j}$, can be written as 
\[
        \sigma_{i,j} = \frac{\varepsilon}{q \cdot R_i(A_j, B_j, u_j, \ell_j)} \cdot cost_i. 
    \]
By the algorithm's sampling scheme, we know that
\[
        \Ex [\sigma_{i,j}] = \sum_{i=1}^{N_j} \frac{R_i(A_j, B_j, u_j, \ell_j)}{\rho_j} \cdot \frac{\varepsilon}{q \cdot R_i(A_j, B_j, u_j, \ell_j)} \cdot cost_i
        \leq \frac{\varepsilon}{q \cdot \rho_j}.
    \]
    We assume that   the algorithm finishes after sampling $K$ vectors $v_1, \dots v_K$ and each $v_i$ is sampled in iteration $\tau_i$. Then, it holds that 
    \[
        C = \sum_{i=1}^m c_i \cdot cost_i = \sum_{i=1}^{K} \sigma_{i, \tau_i},
    \]
    which implies that 
    \begin{align}
        \Ex \lsp C\rsp 
        &= \sum_{i=1}^{K} \Ex \lsp \sigma_{i, \tau_i}\rsp
        \leq \sum_{i=1}^K \frac{\varepsilon}{q \rho_{\tau_i}}\nonumber \\
        &\leq \frac{\varepsilon}{q} \cdot \frac{K}{T+k-1} \cdot \lp \frac{1+3\varepsilon}{6\varepsilon} \cdot \frac{4k}{\Lambda} + 2 + \lambdaMax(X) + \frac{2k}{\Lambda}\rp \label{eq:cost1}\\
        &\leq \frac{1+3\varepsilon}{6q} \cdot \frac{K}{T+k-1} \cdot \lp \frac{6k}{\Lambda} + 3 \rp \label{eq:cost2}\\
        &= \frac{1+3\varepsilon}{6q} \cdot \frac{K}{\Lambda} \cdot \lp \frac{6k}{T+k-1} + \frac{3\Lambda}{T+k-1}\rp \nonumber \\
        &\leq \frac{3(1+3\varepsilon)}{2q} \cdot \frac{K}{\Lambda} \nonumber\\
        & = O(1/\varepsilon^2) \cdot \min\{ 1, k/T\}, \nonumber 
    \end{align}
    where \eqref{eq:cost1} comes from Lemma~\ref{lem:lower bound rho_j}, and \eqref{eq:cost2} holds by the fact that $(1+3\varepsilon)/{6\varepsilon} > 1$ and $\lambdaMax(X) \leq 1$.  Therefore, the statement follows by applying the Markov inequality.
\end{proof}

\subsection{Proof of Theorem~\ref{thm:main-formal}}\label{sec:proof of main theorem}

Now we are ready to combine everything together, and prove Theorem~\ref{thm:main-formal}.

\begin{proof}[Proof of Theorem~\ref{thm:main-formal}]
We first prove that  Assumption~\ref{assumption:runtime assumption} is always satisfied for each iteration, i.e., there is some parameter $\eta>0$ such that 
 $A_j \preceq (1- \eta) u_jI$ and \[\PV(A_j - X - \ell \barr{M})\PV \succeq |\ell| \eta \cdot \PV \barr{M} \PV.\]
   By Lemma~\ref{lem:bounded W}, Lemma~\ref{lem:number of rounds} and the union bound, with probability at least $3/4$, all matrices picked in
\[
    \tau \leq \frac{80q}{3\varepsilon^2} \cdot \frac{1}{c_N} \cdot \Lambda^{(1+2\varepsilon)/q} \leq \frac{80q}{3\varepsilon^2} \cdot n^{c/q}
\]
iterations for some small constant $c < q$
  satisfy
\[
    W_j \preceq \frac{1}{2} (u_j I - A_j)
\]
for all iterations $j$.
  Moreover, by   the proof of Lemma~\ref{lem:potentials decrease in subphases}
  we know that both the   upper and lower potential functions decrease in expectation individually, i.e., it holds for any $j$ that
\[
    \tilde{\Ex} \lsp \Phi^{u_{j+1}}(A_{j+1})\rsp \leq \Phi^{u_j}(A_j)
    \quad \text{and} \quad
    \tilde{\Ex} \lsp \Phi_{\ell_{j+1}}(B_{j+1})\rsp \leq \Phi_{\ell_j}(B_j).
\]
Therefore, conditioning on the event that $\forall i : W_i \preceq 1/2 \cdot (u_i I - A_i)$ we have that 
\[
    \Ex \lsp \Phi^{u_j}(A_j) \big| \forall i : W_i \preceq (1/2) \cdot (u_i I - A_i)\rsp
    \leq \Phi^{u_0} (A_0) \leq \frac{T}{2^q}
\]
and 
\[
    \Ex \lsp \Phi_{\ell_j}(B_j) \big| \forall i : W_i \preceq (1/2) \cdot (u_i I - A_i)\rsp
    \leq \Phi_{\ell_0} (B_0) \leq k \cdot \lp \frac{\Lambda}{2k} \rp^q.
\]
By Markov's inequality, it holds with high probability  that
\[
    \lp \Phi^{u_j}(A_j)\rp^{1/q} = O \lp  T^{1/q}  \cdot \tau^{1/q} \rp
    \quad \text{and} \quad
    \lp \Phi_{\ell_j}(B_j) \rp^{1/q} = O \lp k^{1/q}\cdot \frac{\Lambda}{ k} \cdot \tau^{1/q} \rp.
\]
For any eigenvalue of $A_j$, say $\lambda_i$, we have 
\[
    (u_j-\lambda_i)^{-q} \leq (u_j-\lambdaMax(A_j))^{-q} < \sum_{t=n-T+1}^n \lp u_j-\lambda_t(A_j)\rp^{-q} = \Phi^{u_j}(A_j). 
\]
Therefore, it holds that 
\[
    \lambda_i < u_j - \lp \Phi^{u_j}(A_j)\rp^{-1/q} \leq u_j - O \lp \frac{1}{T^{1/q}} \cdot \frac{1}{\tau^{1/q}} \rp \leq u_j - O \lp \frac{2}{T^{1/q}} \cdot \lp \frac{\varepsilon^2}{q n^{c/q}} \rp^{1/q} \rp. 
\]
Since $u_j$ is $O(1/\varepsilon^2)$ and $T \leq n$, we can choose \[
    \eta = O \lp \frac{\varepsilon^{2+2/q}}{n^{2/q}}\rp
\]
such that
\[
    A_j \prec (1-\eta) u_j I.
\]

The second statement of Assumption~\ref{assumption:runtime assumption} can be shown in a similar way, i.e., we  show that for any nonzero eigenvalue $\lambda_i$ of $B_j$, it holds that $\lambda_i \geq \ell + \Phi_{\ell_j}^{-1/q}$. Hence
\[
    \lambda_i \geq \ell + \Omega \lp \frac{1}{k^{1/q}} \cdot \frac{k}{\Lambda} \cdot \lp \frac{\varepsilon^2}{q n^{c/q}} \rp^{1/q}\rp.
\]
Since $|\ell| = O \lp \frac{k}{\Lambda} \cdot 1/\varepsilon \rp$ and $k \leq n$, we can choose 
\[
    \eta = O \lp \frac{\varepsilon^{2 + 2/q}}{n^{2/q}} \rp
\]
 such that
\[
    \PV B_j \PV \succeq (\ell + |\ell|\eta) \PV.
\]
Multiplying on the left and right by $Z^{\dag}$ we have that
\[
    \PV (A - X) \PV \succeq (\ell + |\ell|\eta) \PV \barr{M} \PV,
\]
which implies that both claims are satisfied.

 Now we will analyse the running time of our algorithm. First of all, by Theorem~\ref{thm:our projection} we can compute the matrix $\mathbb{V}$ in time $t_{\ref{thm:our projection}} = \min \Big\{O(n^{\omega}), \tilde{O}\lp\frac{mk + nk^2}{\sqrt{\lambda^*}}\rp \Big\}$. This is computed only once and will be used throughout every iteration. 
For the computational cost in each iteration, we combine Lemma~\ref{lem:runtime l1}, Lemma~\ref{lem:runtime l2}, Lemma~\ref{lem:aptrace} and Lemma~\ref{lem:running time lambda max 1} to get the overall running time of  
\begin{align*}
    t_{\mathrm{iteration}} &\triangleq t_{\ref{lem:runtime l1}} + t_{\ref{lem:runtime l2}} + t_{\ref{lem:aptrace}} + t_{\ref{lem:running time lambda max 1}} \\
    &= \tilde{O}\lp \lp\min \{n^{\omega}, mk + nk^2 +  k^{\omega}\}\rp/\varepsilon^2\rp  
    + \tilde{O}\lp (\min\{n^{\omega}, mk + nk^2 + k^{\omega}\})/\varepsilon^3\rp \\
    &\qquad \qquad+ \tilde{O}\lp\frac{m}{\eta\eps^2}\rp 
    + \tilde{O}\lp\frac{m}{\eta\varepsilon^3}\rp\\
    &= \tilde{O}\lp \lp \frac{mn^{2/q}}{\varepsilon^{2+2/q}} + \min \left\{ n^{\omega}, mk + nk^2 + k^{\omega}\right\} \rp/\varepsilon^3 \rp.
\end{align*}
Combining this with Lemma~\ref{lem:number of rounds}, which states that the number of iterations is 
\[
    \tau = O \lp q \cdot n^{O(1/q)}/\varepsilon^2\rp,
\]
the algorithm's overall running time is
\begin{align*}
    t_{\mathrm{alg}} 
    &\triangleq t_{\ref{thm:our projection}} + \tau \cdot t_{\mathrm{iteration}}\\
    &= \widetilde{O} \lp \min \left\{ n^{\omega},  \frac{mk + nk^2}{\sqrt{\lambda^*}}\right\} + q\cdot n^{O(1/q)}
    \lp \frac{mn^{2/q}}{\varepsilon^{2+2/q}} + \min \left\{ n^{\omega}, mk + nk^2 + k^{\omega}\right\} \rp/\varepsilon^5 \rp.
\end{align*}
Finally, the total number of edges picked by the algorithm follows by Lemma~\ref{lem:number of sampled vectors};  the lower bound on the minimum eigenvalue of the returned matrix follows by Lemma~\ref{lem:approximation guarantee}; the total cost of the returned edges follows by Lemma~\ref{lem:cost}.  
\end{proof}

\subsection{Further discussion}\label{sec:generalisation of LS15}

Finally, at the end of this section  we will show that our algorithm generalises the almost-linear time algorithm for constructing spectral sparsifiers presented in  \cite{LS15:linearsparsifier}. To see this, notice that for the problem of constructing a spectral sparsifier we have that $X=\mathbf{0}$,   which implies that $T=k=n$ and   $P_S=\PL{\cdot}=I$, the identity matrix in the entire space. Furthermore, we have that 
$ 
u_0 = 2, \ell_0=-2$, and $\alpha=4$.  
Hence, by  Lemma~\ref{lem:number of rounds} it holds that with probability at least $4/5$ the algorithm finishes in $O\left(qn^{3/q}/\varepsilon^2 \right)$ iterations, and 
by Lemma~\ref{lem:number of sampled vectors} it holds that with probability at least $4/5$ our algorithm finishes after choosing $O\left( qn/\varepsilon^2 \right)$ vectors. That is, up to a constant factor, the total number of iterations needed before the algorithm terminates   and the total number of vectors sampled by our algorithm are exactly the same as the one in \cite{LS15:linearsparsifier}.

It remains to show that the output sparsifier is a    $(1 \pm O(\varepsilon))$-approximation of the identity.  To show this, we use the same approach presented in 
  \cite{LS15:linearsparsifier}: notice that it holds for any iteration    $j$ that
\[
    \frac{\barr{\delta}_{u,j} - \barr{\delta}_{\ell, j}}{\barr{\delta}_{u, j}} 
    = \frac{6\varepsilon}{1+3\varepsilon},
\]
and hence
$
    \barr{\delta}_{u, j} \geq (6\varepsilon)^{-1} \cdot \lp\barr{\delta}_{u,j} - \barr{\delta}_{\ell, j} \rp. 
$ Since the condition number of the resulting matrix is upper bounded by 
\[
    \frac{\hat{u}}{\hat{\ell}} = \lp 1 - \frac{\hat{u} - \hat{\ell}}{\hat{u}}\rp^{-1},
\]
the algorithm's returned matrix is a $(1+O(\varepsilon))$-sparsifier because of the fact that
\begin{align*}
    \frac{\hat{u} - \hat{\ell}}{\hat{u}} 
    &= \frac{u_0 - \ell_0 + \sum_{j=1}^{K} \left(\barr{\delta}_{u, j} - \barr{\delta}_{\ell, j}\right)}{u_0 + \sum_{j=1}^{K} \barr{\delta}_{u, j}} 
     \leq \frac{u_0 - \ell_0 + \alpha}{u_0 + (6\varepsilon)^{-1} \cdot \alpha} 
    = \frac{8}{2 + 4(6\varepsilon)^{-1}} \leq 12 \varepsilon. 
\end{align*}

%% file: MainResult.tex
\section{Proof of the main theorem}\label{sec:ProofMainResult}
Finally  we apply our fast \textsf{SDP} solver and the subgraph sparsification algorithm to design an algorithm for the spectral-augmentability problem, and  prove Theorem~\ref{thm:mainAlgConn}. 
We first  give an overview of the main algorithm:   for any input  $G=(V,E)$, the  set $E_W$ of candidate edges, and parameter $k$, our algorithm applies the doubling technique to enumerate all the possible $\gamma$ under which the input instance is $(k,\gamma)$-spectrally augmentable:
starting with the initial $\gamma$, which is set to be $1/n^{1/q}$ and increases by a factor of $2$ each time,  the algorithm runs the \textsf{SDP} solver, a subgraph sparsification algorithm, and a Laplacian solver to verify the algebraic connectivity of the output of our subgraph sparsification algorithm.  The algorithm terminates if the algebraic connectivity is greater than some threshold at some iteration, or it is below the initial threshold.
See Algorithm~\ref{alg:algebraicconn} for formal description.

\begin{algorithm}[!h]
	\caption{Algorithm for augmenting the algebraic connectivity}
	\label{alg:algebraicconn}
	\begin{algorithmic}[1]
	 	\Require the base graph $G=(V,E)$, and the set $E_W$ of $m$ edges defined on $V$, and   $k\in \mathbb{Z}^+$.
		\State $\gamma_0\gets 1/n^{\frac{1}{q}}$;
		\State $\gamma\gets \gamma_0$; 
		\State $\alpha \gets 0$;
		\State $F\gets \emptyset$; \Comment{the set of edges added to $G$}
		\While{$\gamma <1$}
		    \State $\gamma\gets 2\cdot \gamma$, and run the SDP solver from Theorem~\ref{thm:mainSDPlambda2} for $\PSDP(G,W,k,\gamma)$
		    \If {the solver certifies that $\PSDP(G,W,k,\gamma)$ is infeasible}
		    \If{$\alpha=0$}
		    \State{Abort and output \textbf{Reject}.}
		    \Else{}
		    \State{ \Return graph $H=(V,E(G)\cup F)$.} \Comment{$\lambda_2(L_H)\geq c_1\alpha^2 \Delta$}
		    \EndIf
		    \Else~{the solver finds a feasible solution for $\PSDP(G,W,k,0.9\gamma)$ with weights $\{w_e\}_{e\in E_W}$ }
\State{$\alpha\gets \gamma$}
		    \State Let $H=(V,E(G)\cup F)$ be the output of our subgraph sparsification algorithm with edge weights $\{w_e\}_{e\in E_W}$, $q,k$,  and a sufficiently small constant $\varepsilon$.
	        \State $\eta_2\gets $ a $1.1$-approximation of $\lambda_2(L_H)$
	        \Comment{apply the Laplacian solver to compute $\eta_2$}
	        \If{$\eta_2\leq O\left(\Delta\cdot n^{-2/q}\right)$}
	        \State {Abort and output \textbf{Reject}.}
	       \EndIf
	   \EndIf
\EndWhile
	\end{algorithmic}
\end{algorithm}

We will need the following lemma in our analysis.

\begin{lemma}\label{lemma:augmentablespectra}
Let $\gamma>0$. If $G=(V,E)$ is $(k,\gamma \Delta)$-spectrally-augmentable with respect to $W=(V,E_W)$, then the SDP solver finds a feasible solution $(\widehat{\lambda},{w})$ to $\PSDP(G,W,k,(1-\delta')\gamma)$, and the subgraph sparsification algorithm with input $G,E_W,k,\varepsilon,q$ and weights $\{w_e:e\in E\cup E_W\}$ will find a graph $H=(V,E\cup F)$ with $F\subseteq E_W$, $\lambda_2(L_H)\geq c_1 \gamma^2 \cdot \Delta$, $|F|\leq O(qk/\varepsilon^2)$ and total new weights of edges in $F$ at most $O(k/\varepsilon^2)$.
\end{lemma}

\begin{proof} 
If $G$ is $(k,\gamma \Delta)$-spectrally-augmentable with respect to $W$, then there exists a feasible solution to $\PSDP(G,W,k,{\gamma})$ and our SDP solver will find a solution $(\widehat{\lambda},{w})$ to $\PSDP(G,W,k,(1-\delta')\gamma)$, 
for any constant $\delta'>0$. Note that $\widehat{\lambda}\geq (1-\delta')\gamma$. 
Now we use the subgraph sparsification algorithm to sparsify the SDP solution. 

We apply Theorem~\ref{thm:main-formal} to graphs $G,W$, by setting $V=\textrm{im}(L_{G+W})=\ker(L_{G+W})^{\bot}$,  
$X=\left(L^{\dagger/2}_{G+W} L_G L^{\dagger/2}_{G+W}\right)_{|V}$ and 
$Y_e=w_e \left(L^{\dagger/2}_{G+W} L_e L^{\dagger/2}_{G+W}\right)_{|V}$, and $\overline{M}=\sum_{e\in E_W} Y_e$, $K=O(qk/\varepsilon^2)$, $\lambda^*=\lambda_{k+1}(X)$ and $cost_e=\frac{w_e}{\sum_{f\in E_W} w_f}$. Note that  $T=\lceil\textrm{tr}(\overline{M})\rceil \leq k$.  This is true as $\sum_{e\in E_W}w_e\leq k$ , $L^{\dagger/2}_{G+W} L_e L^{\dagger/2}_{G+W} \preceq I$,  $L^{\dagger/2}_{G+W} L_e L^{\dagger/2}_{G+W}$ is a rank one matrix and thus has trace at most $1$. We get a set of coefficients $\{c_e\}$ supported on at most $K$ edges, such that 
\begin{eqnarray*}
C(1-O(\varepsilon))\cdot \min\{1,K/T\}\cdot \lambda_{k+1}(X) &\leq& \lambda_{\min}\left(X+\sum_{e\in E_W}c_e Y_e\right)\\
&\leq& \lambda_{\max}\left(X+\sum_{e\in E_W}c_e Y_e\right) \leq 1+O(\varepsilon)
\end{eqnarray*}

From the above and the fact that $T\leq k$, $K=O(qk/\varepsilon^2)$, we get that
\begin{eqnarray*}
\lambda_2\left(L_G+\sum_{e}c_e w_e L_e\right) &\geq &C(1-O(\varepsilon))\cdot \min\{1,K/T\}\cdot \lambda_{k+1}(X)\cdot \lambda_2\left(L_G + \sum_e w_e L_e\right) \\
&\geq& C'\cdot \frac{\lambda_{k+2}(L_G)}{4\Delta}\cdot\widehat{\lambda}\cdot \Delta = \frac{C'}{4} \cdot \widehat{\lambda} \cdot \lambda_{k+2}(L_G) 
\end{eqnarray*}
for some constant $C'>0$, where the last inequality follows from the fact that $\lambda_i(X)=\lambda_i\left(\left(L^{\dagger/2}_{G+W} L_G L^{\dagger/2}_{G+W}\right)_{|V}\right)\geq \frac{\lambda_{i+1}(L_G)}{4D}$, for any $i\geq 1$.

\begin{claim}
It holds that  $\lambda_{k+2}(L_G)\geq \lambda_{\OPT}$, the  maximum algebraic connectivity of adding a subset set of $k$ edges from $E_W$ to $G$. 
\end{claim}
\begin{proof}
Let $L_R$ be the Laplacian matrix of the graph which is formed by the optimum solution $R$. Then $\dim\ker(L_R)\geq n-k$ as $\textrm{rank}(L_R)\leq |E|\leq k$. Consider the space $S$ spanned by all the eigenvectors of $L_G$ corresponding to $\lambda_2(L_G),\cdots,\lambda_{k+2}(L_G)$. Since $\dim(S)+\dim\ker(L_R)> n$, there exists a unit vector $v\in \ker(L_R)\cap \dim(S)$ such that $v\bot \1$, and $v^{\rot}(L_G+L_R)v\leq \lambda_{k+2}(L_G) + 0 =\lambda_{k+2}(L_G)$. This further implies that $\lambda_{\OPT}=\lambda_2(L_G+L_R)\leq \lambda_{k+2}(L_G)$.  
\end{proof}

Therefore, if we let $F=\{e:e\in E_W, c_e>0 \}$ and set the edge weights to be $\{c_e\cdot w_e: e\in F \}$, then the resulting graph $H=(V,E+F)$ with the corresponding weights satisfies that
\[
    \lambda_2(L_H)=\lambda_2\left(L_G+\sum_{e}c_e w_e L_e\right)\geq 
    c\cdot \gamma\cdot  \lambda_{\OPT} \geq c\cdot \gamma^2 \Delta
\]
for some constant $c>0$, where the last inequality follows from the assumption $G$ is $(k,\gamma\Delta)$-spectrally-augmentable with respect to $W$ and thus $\lambda_{\OPT}\geq \gamma \Delta$.  
Since $$\sum_{e\in E_W} cost_e\cdot c_e \leq O(1/\varepsilon^2) \min\{1,K/T\}=O(1/\varepsilon^2),$$ the total weights of added edges become 
\[
\sum_{e\in E_W} c_e w_e = \left(\sum_{e\in E_W} w_e\right) \cdot \left(\sum_{e\in E_W} cost_e\cdot c_e\right) O(1/\varepsilon^2)\cdot k=O(k/\varepsilon^2). \qedhere
\]
\end{proof} 

Now we are ready to prove Theorem \ref{thm:mainAlgConn}.


\begin{proof}[Proof of Theorem~\ref{thm:mainAlgConn}]
Let $G$ and $W$ be the input to Algorithm~\ref{alg:algebraicconn}. Note that the algorithm only returns a subgraph $H$ with $\lambda_2(L_H)\geq c_1 \gamma_0^2\Delta $, and $H$ contains at most $K=O(kq)$ edges from $E_W$. Hence, if $G$ is not $(O(kq), c_1\gamma_0^2 \Delta)$-spectrally-augmentable with respect to $W$, then the algorithm will reject the input instance.

Without loss of generality, in the following analysis we  assume that $G$ is $\left(k,\lambda_{\star}\Delta\right)$-augmentable for some $\lambda_{\star}>\gamma_0 $, where $\lambda_{\star}\Delta$ is the optimum solution. In this case, by the geometric search over $\gamma$ in the algorithm, when $\gamma \in (\frac{\lambda_{\star}}{2}, \lambda_{\star})$,   the \textsf{SDP} solver will find a feasible solution for $\PSDP(G,W,k,0.9\gamma)$ and the graph $H$ returned by the subgraph sparsification algorithm with input $G,W,q,k$ and constant $\varepsilon$ satisfies that $\lambda_2(L_H)\geq c_1\gamma^2\Delta \geq c_1'\lambda_{\star}^2 \Delta$. If $\gamma\geq \lambda_{\star}$, then the algorithm will either return the graph $H$ that we constructed corresponding to the value $\gamma \in (\frac{\lambda_{\star}}{2}, \lambda_{\star})$, or finds a graph $H$ with $\lambda_2(L_H)\geq c_1\gamma^2\Delta \geq c_1'\lambda_{\star}^2 \Delta$. By Lemma~\ref{lemma:augmentablespectra}, the number of added edges and the total sum of their weights are $O(qk)$ and $O(k)$, respectively.

Furthermore, since $\lambda_{\star}\geq \gamma_0$, it only takes $O(\log n)$ iterations to reach $\gamma$ with $\gamma\in (\frac{\lambda_{\star}}{2},\lambda_{\star})$. In each iteration, by Theorem~\ref{thm:mainSDPlambda2}, the running time for solving $\PSDP(G,W,k,0.9\gamma)$ 
for $\gamma \geq \gamma_0$ is $\tilde{O}(m+n)/\gamma^2)=\tilde{O}((m+n)n^{O(1/q)})$; by Theorem~\ref{thm:main-formal}, the time for applying subgraph sparsification with input $G,W$ and constant $\varepsilon$ is $\widetilde{O} (\min \left\{ qn^{\omega+O(1/q)}, q(m+n)n^{O(1/q)}k^2\right\})$. For the latter, we note that whenever we apply the subgraph sparsification from Theorem~\ref{thm:main-formal}, the corresponding matrix $X$ satisfies that $$\lambda_{k+1}(X)\geq \frac{\lambda_{k+2}(L_G)}{4\Delta}\geq \frac{\lambda_{\OPT}}{4\Delta}=\frac{\lambda_{\star}\Delta}{4\Delta}\geq \frac{\gamma_0 \Delta}{4\Delta}=\Omega(n^{-1/q})$$ and thus we obtain the claimed running time. Furthermore, we can compute an estimate $\eta_2$ of $\lambda_2(L_H)$ by the algorithm given in \cite{Vishnoi13}, which takes $\tilde{O}(|E(H)|+n)=\tilde{O}(n+k)$ time. Thus, the total running time is $\widetilde{O} (\min \left\{ qn^{\omega+O(1/q)}, q(m+n)n^{O(1/q)}k^2\right\})$. This completes the proof of the theorem.
\end{proof}

%% file: appendix.tex
\appendix

\section{Omitted details from Section~\ref{sec:proofmainSDPlambda2}}
\label{sec:OmittedSDP}
\begin{proof}[Proof of Fact \ref{fact:infeasibledual}]
	Suppose that $\langle Z, v,\beta \rangle$ were feasible for $\DSDP(G,W,k,\gamma)$, then it holds that \[
\mat{Z}\bullet \mat{L}_G + k v + \sum_{e\in E_W} \beta_e <\gamma, \quad \mat{Z}\bullet \Delta\mat{P}_{\bot}  =1
	\]
	and 
\[Z \bullet L_e\leq v+\beta_e, \qquad \textrm{for any $e\in E_W$}.\]
Let $(\lambda, {w})$ be the output of the oracle. Then, it holds that 
	\begin{eqnarray*}
		&& A(\lambda, {w})  \bullet Z + B(\lambda, {w}) \cdot v + c(\lambda, {w}) \cdot \beta \\
		&=& \left(\mat{L}_G +\sum_{e\in E_W} w_e \mat{L}_e - \lambda \Delta\mat{P}_{\bot}\right) \bullet Z + v\left(k-\sum_{e\in E_W}w_e\right) + \sum_{e\in E_W} \beta_e \left(1-w_e\right) \\
		&=&  \mat{Z}  \bullet \mat{L}_G  + k v + \sum_{e\in E_W} \beta_e + \sum_{e\in E_W} w_e \cdot (\mat{Z} \bullet \mat{L}_e  -v - \beta_e ) - \lambda \Delta \mat{P}_{\bot}  \bullet Z\\
		&< &\gamma - \lambda <0,
	\end{eqnarray*}
	which contradicts to  the definition of an $(\ell,\rho)$-oracle. Thus, $\langle Z, v,\beta \rangle$ must be infeasible for $\DSDP(G,W,k,\gamma)$.
\end{proof}	

In the following, we give the proof of Theorem~\ref{thm:sdp_MWU}. In order to do so, we introduce a useful theorem. For a matrix $C$, let $\lambda_{\min,N}(C)=\lambda_{\min}(N^{-1/2}CN^{-1/2})$. The following directly follows from Theorem 3.3.3 in \cite{Orecchia:11}. 
\begin{theorem}\label{thm:generalmwu}
Let $\varepsilon \in (0,1/2)$, and let $\{M^{(t)}\}$ be a sequence of loss functions such that $-\ell N \preceq M^{(t)} \preceq \rho N$, where $\rho \geq \ell\geq 0$ for all $t$. We define the update by
	\begin{eqnarray*}
		X^{(t)} = U_{\varepsilon} \left(\frac{1}{2\rho} \sum_{s=1}^{t-1}M^{(s)} \right)
	\end{eqnarray*}
	Then, it holds  for any $T\geq 1$ that
	\begin{eqnarray*}
		\sum_{t=1}^T X^{(t)}\bullet M^{(t)}  \leq \frac{2\rho \log n}{\varepsilon} + T\varepsilon \ell + (1+\varepsilon) \cdot \lambda_{\min, N}\left(\sum_{t=1}^T M^{(t)} \right)
	\end{eqnarray*}
\end{theorem}
Given the above theorem, we are ready to prove Theorem \ref{thm:sdp_MWU}.
\begin{proof}[{Proof of Theorem~\ref{thm:sdp_MWU}}]
By the oracle, we are guaranteed that $-\ell N \preceq M^{(t)}\preceq \rho N$. We apply our Theorem \ref{thm:generalmwu} to obtain that after $T$ rounds, 
	\begin{eqnarray*}
		\sum_{t=1}^T X^{(t)}\bullet M^{(t)}  \leq \frac{2\rho \log n}{\varepsilon} + T\varepsilon \ell + (1+\varepsilon) \cdot \lambda_{\min, N}\left(\sum_{t=1}^T M^{(t)} \right)
	\end{eqnarray*}
	
	As \textsc{Oracle} is $(\ell, \rho)$-oracle, so for any $t\geq 1$, it holds that  $$X^{(t)}\bullet M^{(t)}=A^{(t)}\bullet Z^{(t)} + B^{(t)} \cdot v^{(t)}  + c^{(t)} \cdot \beta^{(t)}\geq 0,$$ where $A^{(t)},B^{(t)},c^{(t)}$ correspond to the decomposition of $M^{(t)}$ and $Z^{(t)}, v^{(t)}, \beta^{(t)}$ correspond to the decomposition of $X^{(t)}$. Thus,
	$$(1+\varepsilon) \cdot \lambda_{\min, N}\left(\sum_{t=1}^T M^{(t)}\right) \geq -\frac{2\rho \log n}{\varepsilon} - T\varepsilon \ell.$$
	By dividing the above inequality by $(1+\varepsilon)T$ we have that 
		$$\lambda_{\min, N}\left(M(\bar{\lambda},\bar{w})\right) \geq -\frac{2\rho \log n}{\varepsilon (1+\varepsilon)T} - \frac{\varepsilon \ell}{1+\varepsilon}\geq -\frac{2\rho \log n}{\varepsilon T}-\varepsilon \ell.$$
By setting $T=\frac{4\rho\log n}{\delta \varepsilon}$ and the assumption that $\varepsilon \leq \frac{\delta}{2\ell}$, we have $$\varepsilon\ell \leq \delta /2, \quad \frac{2\rho\log n}{\varepsilon T}\leq \frac{\delta}{2}.$$
	Thus, 
	$$\lambda_{\min, N}\left(M(\bar{\lambda},\bar{w})\right) \geq -\delta,$$
which gives that 
	$$M(\bar{\lambda},\bar{w})\succeq -\delta N.$$
Thus, for any $x \in \mathbb{R}^{n+1+m}$, it holds that  $x^{\rot}M(\bar{\lambda},\bar{w}) x \geq -\delta x^{\rot} N x$. 

Now we write $M(\bar{\lambda},\bar{w})=\Diag(A(\bar{\lambda},\bar{w}),B(\bar{\lambda},\bar{w}),C(\bar{\lambda},\bar{w}))$. For any $x\in \mathbb{R}^n$, we extend it to $\bar{x}\in \mathbb{R}^{n+1+m}$ by adding $0$ at corresponding entries. Then, it holds that  
\begin{eqnarray*}
x^{\rot} A\left(\bar{\lambda},\bar{w}\right)x &=& \bar{x}^{\rot} M(\bar{\lambda},\bar{w}) \bar{x} \geq -\delta \bar{x}^{\rot} N \bar{x} 
= -\delta 
\Delta x^{\rot} P_\bot x,
\end{eqnarray*} 
which implies that $A(\bar{\lambda},\bar{w})\succeq -\delta {\Delta} P_\bot$, and thus
$$\mat{L}_G +\sum_{e\in E_W} \bar{w}_e \mat{L}_e - \bar{\lambda} \Delta \mat{P}_{\bot} \succeq -\delta\Delta P_{\bot}.$$
Similarly, by restricting $N$ to the sub-matrices corresponding to $B$ and $C$, we get that 
$$B(\bar{\lambda},\bar{w})=k-\sum_{e}\bar{w}_e \geq -\delta m$$ and $$C(\bar{\lambda},\bar{w})\succeq -\delta I,$$
where the second inequality holds as for any $x\in \mathbb{R}^{m}$, $x^{\rot}C(\bar{\lambda},\bar{w}) x\geq -\delta x^{\rot} I x$.
The latter implies that we have  for any $e\in E_W$ that  $$1-\bar{w}_e\geq -\delta.$$
Furthermore, we have  $V(\bar{\lambda},\bar{w})=\bar{\lambda}\geq \gamma$. 

Now we consider the primal SDP  with candidate solution $(\bar{\lambda}-{3\delta },\bar{w}-{\delta})$. We have that $V(\bar{\lambda}-{3\delta },\bar{w}-{\delta})=\bar{\lambda}-{3\delta}\geq \gamma-3\delta$. Moreover, it holds that   $k-\sum_e(\bar{w}_e-\delta)\geq -\delta m + \delta m=0$, $1-(\bar{w}_e-{\delta})\geq 0$ and by the assumption that $(V,E)$ and $(V,E_W)$ have the degree at most $\Delta$, we have  that
$$\mat{L}_G +\sum_{e\in E_W} \left(\bar{w}_e -{\delta}\right) \mat{L}_e - \left(\bar{\lambda}-{3\delta }\right) \Delta\mat{P}_{\bot} \succeq -2\delta {\Delta} P_\bot-\delta {\Delta} P_\bot + {3\delta\Delta} P_\bot \succeq 0.$$
Therefore, $(\bar{\lambda}-{3\delta },\bar{w}-{\delta})$ is a feasible solution to $\PSDP(G,W,k,\gamma-{3\delta })$.
\end{proof}


